\documentclass[twoside]{article}
\usepackage{fullpage}

\usepackage[round]{natbib}
\renewcommand{\bibname}{References}

\setcitestyle{authoryear}

\usepackage[USenglish]{babel}
\usepackage[utf8]{inputenc}
\usepackage{amsthm}
\usepackage{amssymb}
\usepackage{amsmath}
\usepackage{leftidx}
\usepackage[all]{xy}
\usepackage{microtype}
\usepackage{hyperref}
\hypersetup{colorlinks=true,allcolors=blue}
\usepackage{tikz}
\usetikzlibrary{arrows}
\usepackage{listings}
\usepackage{graphicx}
\usepackage{todonotes}

\usepackage{algorithm}
\usepackage{bbm}
\usepackage[noend]{algpseudocode}
\usepackage{varwidth}
\newcounter{algsubstate}


\newcommand\poly{\mathrm{poly}}
\newcommand\ranges{\mathrm{ranges}}

\newcommand\rng[1]{\mathrm{range}_{#1}}
\newcommand\sgn{\mathrm{sgn}}
\newcommand\E{\mathbb{E}}
\newcommand\nnz{\mathrm{nnz}}

\newcommand{\norm}[1]{\ensuremath{\left\| #1\right\|_2}}

\newcommand{\normp}[1]{\ensuremath{\left\| #1\right\|_p}}
\newcommand{\normq}[1]{\ensuremath{\left\| #1\right\|_q}}

\newcommand{\norminf}[1]{\ensuremath{\left\| #1\right\|_\infty}}
\newcommand\R{\ensuremath{R^{-1}}}
\newcommand{\RL}{\ensuremath{\mathbb{R}}}

\newtheorem{mydef}{Definition}[section]
\newtheorem{lem}[mydef]{Lemma}      
      
\newtheorem{thm}{Theorem}
\newtheorem{pro}[mydef]{Proposition}
\newtheorem{cor}[mydef]{Corollary}

\theoremstyle{definition}

\title{\vspace{-2em}$p$-Generalized Probit Regression and Scalable Maximum\\Likelihood Estimation via Sketching and Coresets}

\author{Alexander Munteanu\thanks{Dortmund Data Science Center, Faculties of Statistics and Computer Science, TU Dortmund University, Dortmund, Germany. Email: \texttt{alexander.munteanu@tu-dortmund.de}.}
\and Simon Omlor \thanks{Faculty of Statistics, TU Dortmund University, Dortmund, Germany. Email: \texttt{simon.omlor@tu-dortmund.de}.}
\and Christian Peters \thanks{Faculty of Statistics, TU Dortmund University, Dortmund, Germany. Email: \texttt{christian2.peters@tu-dortmund.de}.}}

\begin{document}
\allowdisplaybreaks

\maketitle

\begin{abstract}
  We study the $p$-generalized probit regression model, which is a generalized linear model for binary responses. It extends the standard probit model by replacing its link function, the standard normal cdf, by a $p$-generalized normal distribution for $p\in[1, \infty)$. The $p$-generalized normal distributions \citep{Sub23} are of special interest in statistical modeling because they fit much more flexibly to data. Their tail behavior can be controlled by choice of the parameter $p$, which influences the model's sensitivity to outliers. Special cases include the Laplace, the Gaussian, and the uniform distributions. We further show how the maximum likelihood estimator for $p$-generalized probit regression can be approximated efficiently up to a factor of $(1+\varepsilon)$ on large data by combining sketching techniques with importance subsampling to obtain a small data summary called coreset.
\end{abstract}

\section{INTRODUCTION}
Probit regression is arguably one of the most successful models when considering predictive models with binary responses. Historically it preceded the logit model and those two are still the gold standards in many statistical domains \citep{Cramer02}. Probit regression enjoys large popularity in toxicology \citep{LeiS18}, biostatistics \citep{VarinC09}, and also in econometrics \citep{Gu09,Moussa19}.
Viewing it as a latent variable model yields an efficient Gibbs sampler for Bayesian analysis \citep{AlbertC93}.

The probit model can be defined as a generalized linear model \citep{McCullaghN89} where the expected value of the response is connected to a linear predictor using the cumulative distribution function (cdf) of a standard normal distribution $\Phi(\cdot)$ as its link function, i.e.,
\[ \mathbb{E}[Y] = \Phi(Z\beta). \]
In principle, the link function can be replaced by any cdf, e.g., by the cdf of a standard logistic distribution for logistic regression. In what follows, we extend the probit model to a $p$-\emph{generalized probit model} by employing the cdf of a $p$-generalized normal distribution \citep{Sub23} as its link function. We work with a standardized form given in \citep{KalkeR13}:
\[\Phi_p(x) = \frac{p^{1-1/p}}{2\Gamma(1/p)} \int_{-\infty}^x  \exp(-|t|^p/p) \,dt , x\in \mathbb{R}, p>0.\]
This family of distributions is of special interest in statistical modeling. It fits more flexibly to data because the tail behavior (kurtosis) can be controlled by choice of the parameter $p \geq 1$\footnote{Extending to $p\in (0,\infty)$ is in principal possible, but $p<1$ leads to non-convex level sets in the multivariate case and the maximum likelihood estimation problem thus becomes inefficient to solve.}. 
Probability density functions (pdf) and cumulative distribution functions (cdf) for various values of $p$ are shown in Figure \ref{fig:pdfcdf}.
The standard normal distribution and the standard probit model are obtained in the case $p=2$ with {squared} exponential tails. The case $p=1$ corresponds to the Laplace distribution with exponential tails, akin to the logistic distribution, and thus models more robustness to outliers. Finally, in the limiting case for $p\rightarrow \infty$ we obtain a uniform distribution over $[-1,1]$ whose tails are cut-off. Such a model is very sensitive to outliers and will fit the extreme data points, dominating other points with average behavior.

After introducing the generalized probit model, we set up the likelihood in order to learn the parameter $\beta\in \mathbb{R}^d$ via maximum likelihood estimation resp. minimization of the negative log-likelihood. The resulting $p$-probit loss function is strictly monotonic and convex for $p\geq 1$.
\begin{figure}[t!]
    \centering
    \includegraphics[width=0.49\linewidth]{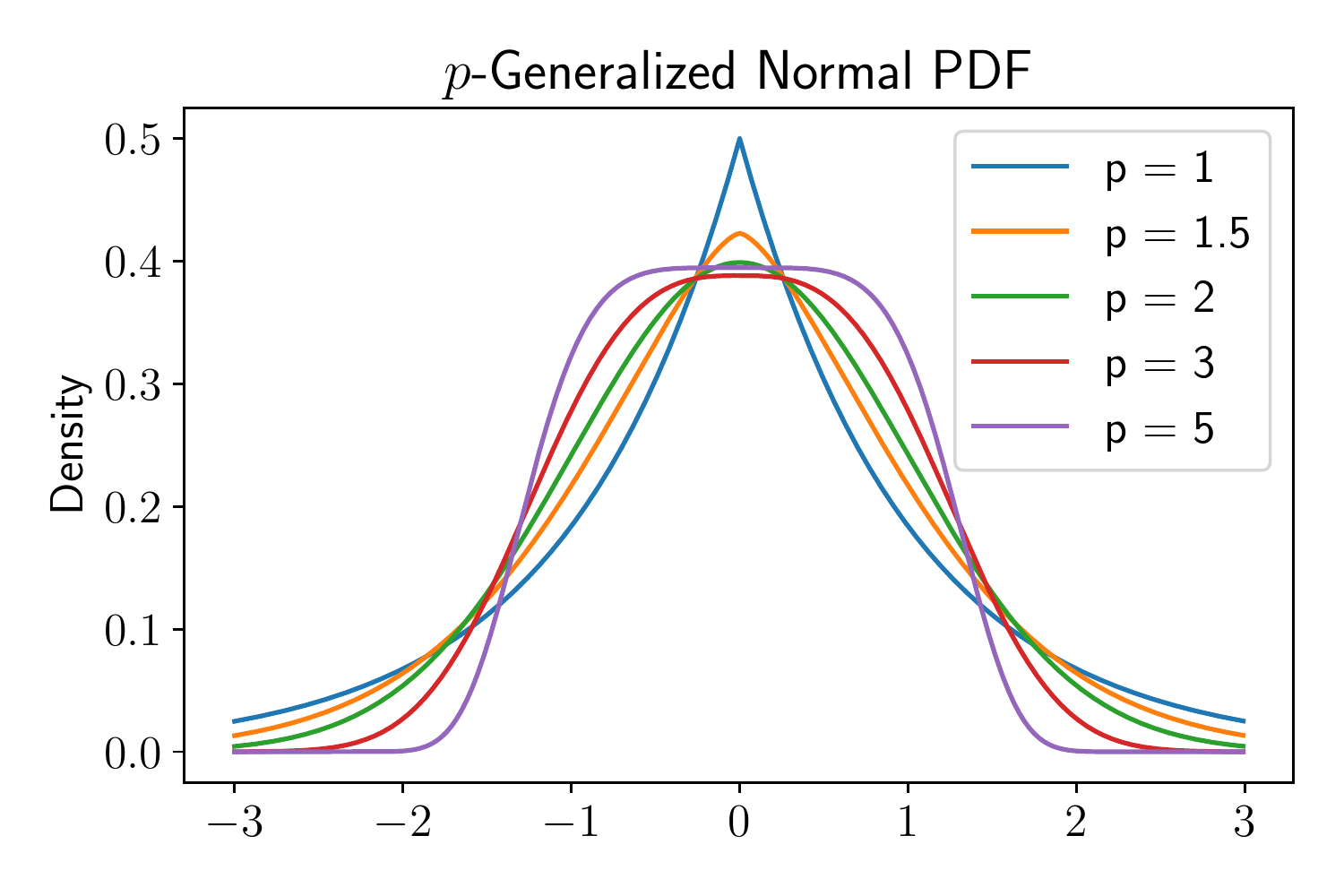}
    \includegraphics[width=0.49\linewidth]{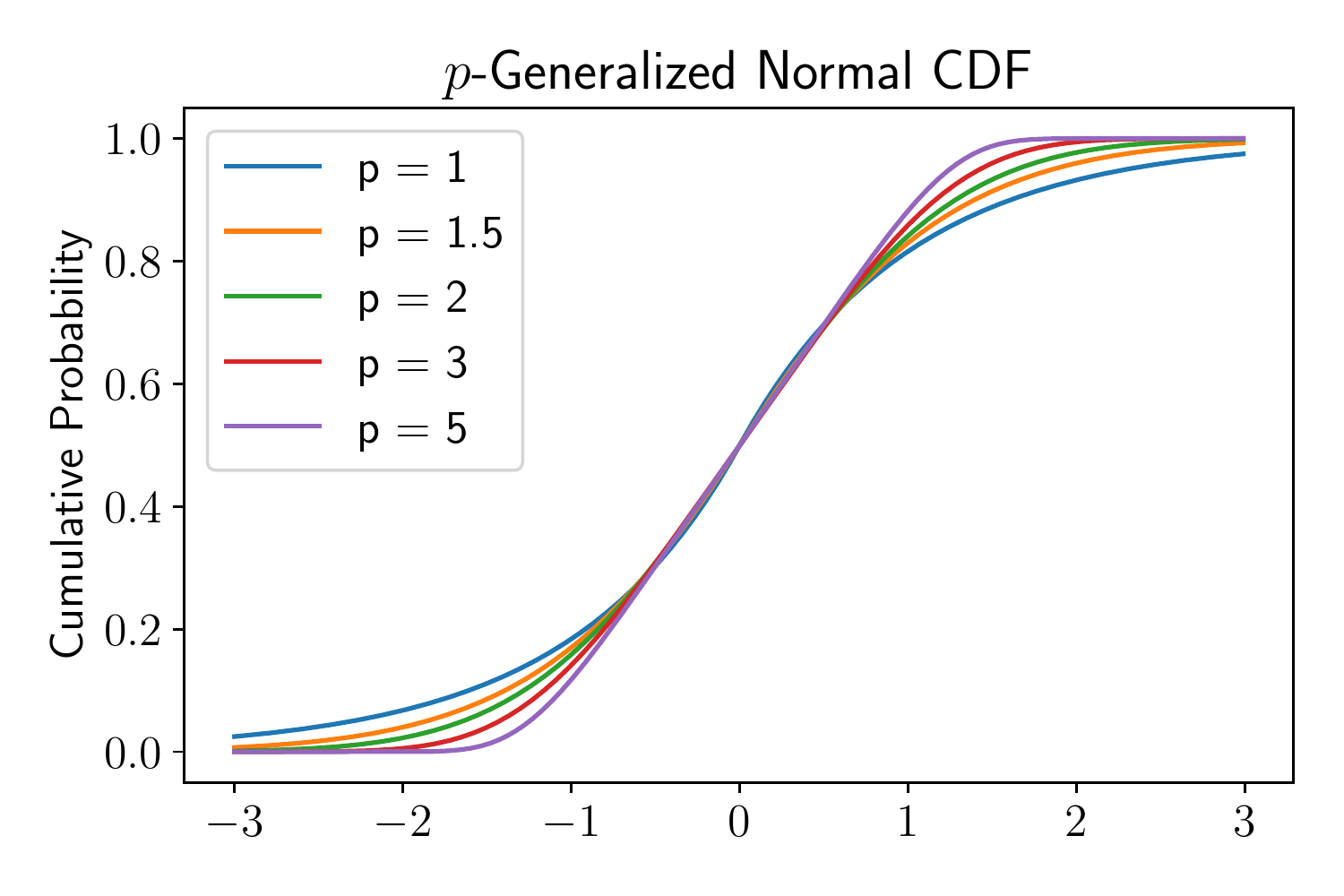}
    \caption{Probability density functions (pdf) and cumulative distribution functions (cdf) of $p$-generalized normal distributions for various values of $p$.}
    \label{fig:pdfcdf}
\end{figure}
We note that efficient routines are available for evaluating the probability density function, the cumulative distribution function and drawing random variables \citep{KalkeR13}\footnote{Implementations are available in the Python package \texttt{scipy.stats.gennorm} as well as in the R package \texttt{pgnorm}.}. This allows to employ standard optimization methods like gradient descent or Newton's method \citep{Bubeck15}. However, working with large scale data poses a severe limitation since each iteration scales at least linearly with the number of data points $n$.
We note that stochastic gradient descent (SGD) and mini batching strategies provide scalable alternatives, and often work well in practice, but do not give the desired accuracy guarantees that we pursue in this paper. Examples and experiments where those methods fail were given in \citep{MunteanuSSW18, MunteanuOW21}. Further, well-known random projections (e.g. Johnson-Lindenstrauss, Count Sketch, etc.) do not work for asymmetric functions since the sign of vectors are not preserved, and they have near linear lower bounds for $p>2$. Another disadvantage of random projections is that linear combinations of data decrease their interpretability.
We thus focus on approximating the \emph{full} gradient via importance sampling to cope with the limitations arising with large $n$.

Two methods called \emph{coresets} and \emph{sketching} \citep{Phillips17,MunteanuS18,Feldman20} received a lot of attention recently and led to the so called \emph{sketch and solve} paradigm \citep{Woodruff14}. The first step of sketch and solve is a data reduction: we compute coresets and sketches which are concise data summaries that approximate the original full data with respect to a given loss function. In a second step, we pass the reduced data to the standard optimization algorithm that we would have used on the full data. But the algorithm now performs much more efficiently on the summary due to its reduced size. The theoretical analysis ensures that the result is close to the result of analyzing the original large data. In this paper we combine methods from the sketching and coreset literature to facilitate an efficient estimation for $p$-generalized probit regression on large data.

\subsection{Related Work}
The $p$-generalized normal distribution was introduced by \citet{Sub23} and became widely popular in the late twentieth century \citep{Goodman73,OsiewalskiS93,Kotz94unicont}. We refer to \citep{Dytso18} for an extensive survey on applications and analytical properties of the generalized normal distribution. Among other results, this reference provides an asymptotic characterization of the tails of generalized normal distributions, which we concretize in a non-asymptotic way, similar to the classic work of \citet{Gordon41} on the standard normal distribution. Another nice property is the decomposability into independent marginals, which characterizes the class of multivariate generalized normal distributions \citep{Sinz09,Dytso18}. In summary, the class of $p$-generalized normal distributions naturally extends the standard normal distribution and retains several of its useful and desirable analytical properties. Hereby, it offers more parametric flexibility allowing for tails that are either heavier $(p < 2)$ or lighter $(p > 2)$ than normal $(p=2)$ which makes it an excellent choice in many modeling scenarios \citep{Dytso18}.

Most related to our work are coreset and sketching algorithms for \emph{linear} $\ell_p$ regression \citep{Clarkson05,DasguptaDHKM09,SohlerW11,MengM13,WoodruffZ13,ClarksonDMMMW16}, which aims at minimizing $\|Z\beta-Y\|_p$, and can be seen as a standard linear model $Y=Z\beta + \eta$, where the error term $\eta$ follows a $p$-generalized normal distribution\footnote{The connection is not explicitly elaborated in those references.}. The earlier works relied on subsampling according to $\ell_p$ norms derived from a well-conditioned basis, whose approximation posed the computational bottleneck. Subsequent works improved the previous results significantly by approximating those bases via fast linear sketching techniques. To our knowledge, we are the first to study coresets and sketching for $\ell_p$ regression in the setting of \emph{generalized} linear models.

Coresets for \emph{generalized} linear models were studied by \citet{KerstingMM18}, who gave an impossibility result for Poisson regression and a heuristic based on a latent variable log-normal count model. Starting with \citep{ReddiPS15,HugginsCB16}, a series of works focused on logistic regression. In \citep{MunteanuSSW18} it was shown that there are no sublinear sketches or coresets for logistic regression, and to overcome this limitation the authors introduced a complexity parameter $\mu$ for compressing the data, which is related to the statistical modeling and their assumption that $\mu$ is small is naturally met on real world data and applications. The theoretical bounds have been recently improved \citep{MaiRM21} using $\ell_1$ Lewis weights instead of sampling according to the square root of $\ell_2$ leverage scores. However, the practical performance was only slightly improved at the cost of a $O(\log\log n)$ factor increase in the running time and number of passes over the data, required for approximating the Lewis weights. %Another milestone
Recently, \citet{MunteanuOW21} developed the first oblivious linear sketch for solving logistic regression up to $O(1)$ error in a single pass over a turnstile stream. 

A different series of related works studied regularization as a means to overcome the lower bounds \citep{FeldmanT18,TukanMF20}. Recently it was shown that if the regularization is strong enough then a uniform sample suffices \citep{SamadianPMIC20}. In this paper we do not rely on regularization because a coreset for the regularized version of a problem does not yield a coreset for the unregularized version. On the other hand, if we have a coreset for the unregularized setting, it remains valid under any regularization term that is non-negative and does not depend on the data.
This holds for the most common cases, e.g. for ridge, LASSO and elastic net regularization.

\subsection{Our Contribution}
A simple observation is that the tails of the logistic and the $1$-generalized normal distributions are both exponential, and thus very similar up to a constant factor\footnote{The distributions differ more significantly near their means, i.e., around zero.}. We thus expect similar results for the generalized linear models employing those cdfs as their link functions with respect to both, lower and upper bounds. Moreover, the intuition behind this paper is that we can generalize the aforementioned notion of $\mu$-complexity and hereby extend the class of generalized linear models that admit small coresets in a natural way, to obtain results for the broader class of $p$-generalized probit regression.
Specifically, our contributions are:

\begin{itemize}
    \item We introduce the $p$-generalized probit model as a flexible framework for modeling binary data, and for classification.
    \item To facilitate an efficient and scalable maximum likelihood estimation for the parameters of the $p$-generalized probit model, we develop coreset constructions for the associated $p$-probit loss function via the sensitivity framework \citep{LangbergS10,FeldmanL11} and combine with sketching techniques of \citet{WoodruffZ13} to obtain well-conditioned $\ell_p$ bases and to approximate $\ell_p$ leverage scores efficiently.
    \item To this end we derive analytical properties of the $p$-generalized normal distribution that continue and concretize the asymptotic tail bounds of \citet{Dytso18} using similar techniques as the classic non-asymptotic work of \citet{Gordon41} on the tails of the standard normal distribution. This result may be of broader independent interest.
    \item We provide a new VC dimension bound for our weighted loss functions by a novel fine-grained analysis, that also improves the $O(d\log n)$ bound of \citet{MunteanuSSW18} for logistic regression to $O(d\log(\mu/\varepsilon))$\footnote{It is widely believed that $O(d)$ suffices but there is no formal proof for this, except in the case of equal weights.}.
    \item We conduct an empirical evaluation on benchmark data. We compare the case $p=1$ to logistic regression, assessing their proximity. We assess the results obtained from $p$-probit regression for different values of $p$. We further evaluate our coreset constructions for different values of $p$ with respect to their approximation accuracy and algorithmic efficiency.
\end{itemize}

In summary, we advance the statistical modeling of binary data, make contributions in the field of analytical properties of the $p$-generalized normal distributions, generalize existing coreset constructions to a broader class of loss functions, and demonstrate the practical relevance of our methods.

\section{TECHNICAL OVERVIEW}
Several details are omitted due to page limitations. All missing pieces can be found in the appendix.

\subsection{Preliminaries}
For a fixed constant $p\in[1,\infty)$, and an unknown paramter $\beta\in\mathbb{R}^d$ we define the $p$-generalized probit model as a generalized linear model \citep{McCullaghN89}:
\[ \E[Y] = \Phi_p(Z\beta), \]
whose link function is given by $\Phi_p(r)= \frac{p^{1-1/p}}{2\Gamma(1/p)}\int_{-\infty}^{r} \exp(-|t|^p/p)\,dt$, i.e., the cumulative distribution function of the $p$-generalized normal distribution.
Suppose we observe data $\{(z_i,y_i)\}_{i=1}^n$, where we have \emph{row vectors} $z_i\in\mathbb{R}^d$ and $y_i\in \{0,1\}$ for each $i\in[n]$. We have $\Pr[y_i = 1] = \E[y_i] = \Phi_p(z_i\beta)$ and similarly $\Pr[y_i = 0] = 1-\Phi_p(z_i\beta) = \Phi_p(-z_i\beta)$. The likelihood of the model is given by 
\begin{align*}
\mathcal{L}(\beta|Z,Y) &= \prod_{i=1}^n \Phi_p(z_i\beta)^{y_i} \; \Phi_p(-z_i\beta)^{(1-y_i)} \\
&= \prod_{i=1}^n \Phi_p((2y_i-1)\cdot z_i\beta).
\end{align*}
We set $x_i = -(2y_i-1)z_i$ for all $i\in[n]$ for convenience of presentation\footnote{Indeed, $y_i$ and $z_i$ always appear in that combination so we can assume the data consists only of $x_i\in\mathbb{R}^d$ for $i\in [n]$.}. Consequently the log-likelihood simplifies to
$\ell(\beta|X,Y) = \sum_{i=1}^n \ln(\Phi_p(-x_i\beta))$.
Maximizing the likelihood is thus equivalent to minimizing the following $p$-probit loss function, where we add a weight vector $w\in\mathbb{R}^n_{>0}$ for technical reasons detailed below:
\begin{align*}
f_w(X\beta) &
= \sum_{i=1}^n -\ln(\Phi_p(-x_i\beta)) \cdot w_i.
\end{align*}
We omit the subscript whenever the weights are uniform, i.e., $w_i=1$ for all $i\in[n]$.
Moreover, to simplify notations we define the individual loss function
\begin{align}\label{eq:g_function}
    g(r)=-\ln(\Phi_p(-r)) .
\end{align}
The above loss function $f_w$ can be minimized within the framework of convex optimization \citep{Bubeck15} and can be regarded as solved for small to medium sized data. We provide a derivation of the gradient and the Hessian in Appendix \ref{app:grad_hess} for completeness.

\subsection{Coresets for \texorpdfstring{$p$}{p}-Probit Regression}
We focus on the situation where the number of data points $n \gg d$ is very large.
In this case, applying a standard algorithm directly is not a viable option due to being either too slow or even impossible requiring too much working memory. Following the sketch and solve paradigm, our goal is to reduce the data without losing much information, before we approximate the problem efficiently on the reduced data using a standard solver. More formally our goal is now to develop an $\varepsilon$-coreset:

\begin{mydef}\label{def:coreset}
A weighted $\varepsilon$-coreset $C=(X', w)$ for $f$ is a matrix $X' \in \mathbb{R}^{k \times d}$ together with a weight vector $w \in \mathbb{R}^{k}_{> 0}$ such that for all $\beta \in \mathbb{R}^d$ it holds that
\[  |f_w(X'\beta)-f(X\beta)|\leq \varepsilon\cdot f(X\beta). \]
\end{mydef}
Small $\varepsilon$-coresets with $k \ll n$ cannot be obtained in general. There are examples where no $\varepsilon$-coreset of size $k=o(n/\log n)$ exist, even if $\varepsilon$ attains an arbitrarily \emph{large} value \citep{MunteanuSSW18}. When $X'$ is constrained to be a subset of the input, then the bound can be strengthened to $\Omega(n)$ \citep{FeldmanT18}. Those impossibility results rely on the monotonicity of the loss function and thus extend to all $p$-generalized probit losses.
To get around those strong limitations, \citet{MunteanuSSW18} introduced a parameter $\mu$ as a natural notion for the complexity of compressing the input matrix $X$ for logistic regression, which we adapt to $p$-generalized probit models to parameterize our results.

\begin{mydef}\label{def:mu_complex}
	Let $X \in \mathbb{R}^{n \times d }$ be any matrix. For a fixed $p\geq 1$ we define
	\[ \mu_p(X)=\sup_{\beta \in \mathbb{R}^d\setminus\{0\}} \frac{\sum_{x_i\beta >0}|x_i\beta|^p}{\sum_{x_i\beta <0}|x_i\beta|^p}. \]
	We say that $X$ is $\mu$-complex if $\mu_p(X)\leq \mu < \infty$.
\end{mydef}

Computing $\mu$ is hard in general and is mainly for parameterizing the size of coresets in the theoretical analysis. In practice we spend as much memory as we can afford and do not precompute $\mu$. Fortunately, for many real world data sets $\mu$ turns out to be sufficiently small \citep{MunteanuSSW18}.

Our main result shows that if $\mu$ is small, then there exists a small coreset $C$.
In fact, the size of $C$ does not depend on $n$ at all and it can be computed efficiently in two passes over the data:

\begin{thm}\label{thm:mainoverview}
If $X \in \mathbb{R}^{n \times d}$ is $\mu$-complex for any fixed $p \in [1, \infty)$ then with constant probability %at most $\delta \in (0, \frac{1}{2})$ 
we can compute an $\varepsilon$-coreset $C=(X', w)$ for $p$-probit regression of size $k={O}(\frac{S}{\varepsilon^2}(d\ln(\varepsilon^{-1} \mu)\ln S))$ in two passes over the data, where
\begin{align*}
S= \begin{cases}
{O}\left(\mu d \right), & \text{for }p=2\\
{O}\left(\mu d^p(d\log d)^{2}\right), & \text{for }p \in [1, 2)\\
{O}\left(\mu d^{2p}(d\log d)^{2}\right), & \text{for }p \in (2, \infty).
\end{cases}
\end{align*}
Algorithm \ref{alg:main} runs in ${O}(\nnz(X)d+\poly(d))$\footnote{Here, $T=\poly(d)$ means that there exists a constant $c\geq 1$ such that $T = \Theta(d^c)$.} %\footnote{TODO Here ${O}(\poly(d))$ means that there exists some polynomial $f(t)$ such that the coreset size is bound by ${O}(f(d))$.}
time for $p\in[1,2]$ and in ${O}(\nnz(X)d+\poly(d)n^{1-\frac{2}{p}}\log n)$ time for $p>2$, were $\nnz{}$ denotes the number of non-zeros.
\end{thm}

In the case $p=2$, which is of special importance since it corresponds to the standard probit regression model, we have the following improvements:

\begin{cor}\label{cor:p2}
Consider the setting of Theorem \ref{thm:mainoverview}, for $p=2$. The running time can be reduced to ${O}(\nnz(X)\log n + \poly(d))$. Moreover there exists a single pass online algorithm (Algorithm \ref{alg:online}) that runs in time $O(nd^2 + \poly(d))$ and computes a coreset of size $${O}\left(\frac{\mu d^2 \ln(\|X\|_2)}{\varepsilon^2}\ln(\varepsilon^{-1}\mu)\ln(\mu d \ln(\|X\|_2))\right),$$ where $\|X\|_2$ denotes the largest singular value of $X$.
\end{cor}

The coreset of Theorem \ref{thm:mainoverview} or Corollary \ref{cor:p2} can then be used to compute a $(1+\varepsilon)$-approximation for the optimal maximum likelihood estimator for $\beta$:

\begin{cor}\label{cor:minimization}
Let $(X',w)$ be a weighted $\varepsilon$-coreset for $f$. Let $\tilde{\beta}\in \operatorname{argmin}_{\beta\in\mathbb{R}^d} f_w(X'\beta)$. Then it holds that $f(X\tilde\beta )\leq (1+3\varepsilon)\min_{\beta\in\mathbb{R}^d} f(X\beta)$.
\end{cor}
\paragraph*{High Level Description of the Algorithm}
Before getting into the details we outline Algorithm \ref{alg:main}:
\begin{enumerate}
    \item We make a first pass to sketch the data for the purpose of estimating their individual importance.
    \item We make another pass to subsample the data proportional to their importance to obtain a coreset.
    \item We solve the reduced problem on the coreset using a standard algorithm for convex optimization.
\end{enumerate}

This approach implements the sensitivity sampling framework. The importance measure that it builds upon is called sensitivity, which measures the worst case contribution of each input point to the objective function. For efficiency reasons we first compute a sketch of the data in one pass. In the second pass, the sketch is used to approximate the $\ell_p$ leverage scores, which upper bound the sensitivities of the input points. Hereby, we pass them one-by-one to a reservoir sampler to obtain the coreset. In the following paragraphs we will prove that the output has the desired coreset property. To this end we will outline the sensitivity framework, then establish the connection of our loss function to $\ell_p$ spaces. We will further bound the VC dimension for our loss function and bound the sensitivities in terms of $\ell_p$ leverage scores. We will then show how we can approximate the $\ell_p$ leverage scores by means of well-conditioned bases via sketching, to use them as importance measure in the actual algorithm. Putting all those pieces together will prove Theorem \ref{thm:mainoverview}.
Finally, we can solve the original problem approximately using gradient descent or other standard methods for convex optimization \citep[see][]{Bubeck15} on the resulting coreset.
We continue with a more detailed description of all pieces and then put them together.

\paragraph*{The Sensitivity Framework}
To prove Theorem \ref{thm:mainoverview}, we start with the so called sensitivity framework \citep{LangbergS10,FeldmanL11,BravermanFL16,FeldmanSS20}. This framework provides a meta theorem, that can be used to construct coresets via importance subsampling. The main parameters that we need to bound is the VC dimension $\Delta$ of a weighted set of functions that can roughly be thought of as the dimension of the parameter space, and the sensitivities of the input points. The latter quantify the worst-case contributions of single input points to the loss function. Recall Eqn. (\ref{eq:g_function}). 
\begin{mydef}[\citet{LangbergS10}]\label{def:sensitivity}
For all $i\in[n]$, we define the sensitivity $\zeta_i$ of point $x_i$ by
\begin{align*}
\zeta_i=\sup_{\beta \in \mathbb{R}^d} \frac{g(x_i\beta)}{f(X\beta)}.
\end{align*}
The total sensitivity is given by $Z=\sum_{i=1}^n \zeta_i$
\end{mydef}

Calculating the sensitivities is often difficult in the sense that it requires to solve the problem to optimality first, before it allows us to obtain an approximation. Fortunately it suffices to calculate upper bounds $S=\sum\nolimits_{i=1}^{n} s_i \geq \sum\nolimits_{i=1}^{n} \zeta_i = Z$, which is often tractable, though it is crucial to control the increased total sensitivity bound $S$. Our bound $S$ will depend on $p, d,$ and $\mu $ as detailed in Theorem \ref{thm:mainoverview}. The meta theorem (Proposition \ref{thm:sensitivity}) then ensures that if we sample $k=O\left( \frac{S}{\varepsilon^2}\left( \Delta \ln S + \ln \left(\frac{1}{\delta}\right) \right) \right)$ points each with probability $p_i = \frac{s_i}{S}$, then the resulting subsample, reweighted by $w_i=\frac{1}{k p_i}$, is an $\varepsilon$-coreset with probability at least $1-\delta$. More formal details are deferred to Appendix \ref{app:sensitivity}.

\paragraph{Properties of the Loss Function}
The first step for bounding the parameters of the sensitivity framework is to analyze the tails of the generalized normal distribution and relate its negative logarithm to $g(r)\approx {r^p}/{p}$ for $r\geq 0$, and to $\exp(-{|r|^p}/{p})$ for $r<0$. This generalizes classic results on the tails of the standard normal distribution \citep{Gordon41}. The following lemma makes this more precise and follows by combining the technical derivations of Lemma \ref{lem:h-prop} in the appendix.

\begin{lem}\label{lem:g-prop}
The function $g$ (see Eqn. (\ref{eq:g_function})) is convex and strictly increasing.
Further for any $r\geq 0$ we have
\begin{align*}
g'(r) &\geq r^{p-1},
\intertext{for any $r\geq 1$ we have}
g'(r) &\leq r^{p-1}+\frac{p-1}{r}.
\intertext{and there exists a constant $c_1>0$ such that}
g(r) &\geq c_1 e^{-|r|^p/p}
\end{align*}
for any $r<0$.
\end{lem}

\paragraph*{Bounding the VC Dimension}
It is well-known \citep{HugginsCB16,MunteanuSSW18} that for uniform weights, the VC dimension of any invertible function of $X\beta$ can be bounded by relating to the class of affine separators, whose VC dimension equals $d+1$, see \citep{KearnsV94}. \citet{MunteanuSSW18} showed that the number of distinct weights for logistic regression can be bounded by $|W|=O(\log n)$ by rounding the sensitivities to the next powers of $2$, implying the VC dimension is bounded by $O(d\log n)$. The same arguments would apply to $p$-probit loss as well. In our fine-grained analysis we start by \emph{additionally} rounding all sensitivities below $\frac{S}{n}$ to $\frac{S}{n}$ while increasing the total sensitivity to at most $2S+n\frac{S}{n}=3S$. Then we inspect large and small sensitivities separately, depending on their relation to a threshold $s_0$. For the large sensitivities $s_i>s_0$ we can show using Lemmas \ref{lem:g-prop} and \ref{lem:gbound} that their contribution can be approximated well when we replace $g(r)$ by a function $G^+(r)$ that equals zero for $r < 0$, and $\frac{r^p}{p}$ for $r\geq 0$, and whose VC dimension is bounded by $O(d)$.
\begin{lem}\label{lem:senssplit}
Let $I_1$ be the index set of all data points with $s_i>s_0:=\frac{\mu S c \ln(p\varepsilon^{-1})}{ \varepsilon  n }$ for some constant $c \in \mathbb{R}_{>0}$.
Then for all $\beta \in \mathbb{R}^d$ it holds that
\begin{align*}
\sum_{i \in I_1}G^+(x_i\beta) &\leq \sum_{i \in I_1}g(x_i \beta) 
\leq (1+\varepsilon)\sum_{i \in I_1}G^+(x_i\beta) ~+~ \varepsilon \cdot \frac{n}{\mu}.
\end{align*}
\end{lem}
The additive error can be charged using a lower bound on the $p$-probit loss function.
\begin{lem}\label{lem:f-bound}
Assume $X\in\mathbb{R}^{n\times d}$ is $\mu$-complex. Then we have for any $\beta \in \mathbb{R}^d$ that
\begin{align*}
f(X\beta)=\Omega\left(\frac{n}{\mu}\left(1+\ln(\mu)\right)\right).
\end{align*}
\end{lem}
For the remaining small sensitivities, they are bounded between narrow thresholds $\frac{S}{n}\leq s_i \leq s_0$ and by the previous rounding to the next power of $2$ we can argue that there are only $|W|=O(\log\frac{\mu}{\varepsilon})$ different weights. Overall we can bound the VC dimension by $O(d\log\frac{\mu}{\varepsilon})$ using the following argument.

By the technical Corollary \ref{cor:senssplit} in the appendix, our goal of obtaining a coreset for $f$ reduces to obtaining a coreset for the substitute function $$\tilde{f}(X\beta)=\sum_{i\in [n] \setminus I_1}g(x_i\beta) + \sum_{i\in I_1}G^+(x_i\beta).$$
To this end we set $\mathcal{F}_1=\{ w_i G^+_i ~|~ i \in I_1 \}$ where $G^+_i(\beta)=G^+(x_i\beta)$ and $\mathcal{F}_2=\{ w_i g_i ~|~ i \in I_2:=[n]\setminus I_1  \}$ where $g_i(\beta)=g(x_i\beta)$.
Further we set $\mathcal{F}=\mathcal{F}_1 \cup \mathcal{F}_2$ and show that the VC dimension of $\mathcal{F}$ can be bounded as desired:

\begin{lem}\label{lem:VCdim}
For the VC dimension $\Delta$ of $ \mathfrak{R}_{\mathcal{F}}$ we have
\begin{align*}
    \Delta &\leq (d+1) \left(\log_2\left( \mu  c \varepsilon^{-2}\right)+2 \right) ={O}(d\log({\mu}/{\varepsilon})).
\end{align*}
\end{lem}

\paragraph*{Bounding the Sensitivities}
Using the analytic bounds and the assumption of $\mu$-complex data we can relate the sensitivities for the negative domain to roughly $\frac{\mu}{n}$, and for the positive part to the $\ell_p$ leverage scores, defined by $u_i = \sup_{\beta\neq 0} \frac{|x_i\beta|^p}{\|X\beta\|^p_p}$, cf. \citep{DasguptaDHKM09}. This gives us the following lemma. 
\begin{lem}\label{lem:sensbound}
There is a constant $c_s$ such that the sensitivity $\zeta_i$ of $x_i, i\in[n]$ for $\tilde{f}$ is bounded by
\begin{align*}
\zeta_i \leq c_s \mu\left(\frac{1}{n}+u_i\right)
\end{align*}
\end{lem}
Summing over the combined upper bound yields a total sensitivity of $S\leq\sum_{i\in [n]} c_s\mu (1/n + u_i) = O(\mu \sum_{i\in [n]} u_i)$, which dominates the size of the coreset; see Theorem \ref{thm:mainoverview}. The sum of leverage scores can further be bounded by roughly $d^{O(p)}$ by means of $\ell_p$-well-conditioned bases for the column space of $X$.

\paragraph*{Well-conditioned Bases via Sketching Techniques}
One can approximate the $\ell_p$ leverage scores using an orthonormal basis for the column space of $X$. Unfortunately this gives only an $n^c$-approximation for $p\neq 2$. We thus work with a generalization to so called \emph{well-conditioned bases}.
An $(\alpha,\beta,p)$-well-conditioned basis $V$ is a basis that preserves the norm of each vector well, as detailed in the following definition.

\begin{mydef}[\citet{DasguptaDHKM09}]
\label{def:good_basis}
Let $X$ be an $n\times m$ matrix of rank $d$, let $p \in [1,\infty)$, and let $q$ be its dual norm, i.e., $q \in (1, \infty]$ satisfying $\frac{1}{p}+\frac{1}{q}=1$.
Then an $n \times d$ matrix $V$ is an \emph{$(\alpha,\beta,p)$-well-conditioned basis} for the column space of $X$ if\\
(1) $\Vert V \Vert_p:=\left( \sum_{i \leq n, j \leq d}|V_{ij}|^p\right)^{1/p}\leq \alpha$, and\\
(2) for all $z\in\mathbb{R}^d$, $\Vert z \Vert_q \leq \beta \Vert V z\Vert_p$.

We say that $V$ is a \emph{$p$-well-conditioned basis} for the column
space of $X$ if $\alpha$ and $\beta$ are $d^{O(1)}$,
independent of $m$ and $n$.
\end{mydef}
We can bound the leverage scores in terms of the row-wise $p$-norms of such a basis.% such that $\sum u_i\leq \beta^p \|q_i\|_p^p=\beta^p \|Q\|_p^p = (\alpha\beta)^p = d^{O(p)}$.
\begin{lem}\label{lem:levscorebound}
Let $V$ be an \emph{$(\alpha,\beta,p)$-well-conditioned basis} for the column space of $X$.
Then it holds for all $i\in[n]$ that $u_i \leq \beta^p \Vert v_i\Vert_p^p$. % and thus $$\sum_{i=1}^{n} u_i\leq \beta^p \|q_i\|_p^p=\beta^p \|Q\|_p^p = (\alpha\beta)^p.$$
As a direct consequence we have $\sum_{i=1}^{n} u_i\leq \beta^p \|V\|_p^p \leq (\alpha\beta)^p = d^{O(p)}$.
\end{lem}

A prominent example of a well-conditioned basis is the aforementioned orthonormal basis for $\ell_2$, which can be obtained by QR-decomposition (or SVD) in $O(nd^2)$. Such a basis $Q$ is $(\sqrt{d},1,2)$-well-conditioned, since $\|Q\|_F = \sqrt{d}$ and $\|Qz\|_2 = \|z\|_2$ due to rotational invariance of the $\ell_2$-norm. For general $p$ there exist so called Auerbach bases \citep{Auerbach1930} with $\alpha=d$ and $\beta=1$, and approximations thereof can be computed in time $O(nd^5\log n)$ via L\"owner–John ellipsoids \citep{Clarkson05,DasguptaDHKM09}.

To get around this computational bottleneck, we first apply a sketching matrix to our data that yields a so called subspace embedding for approximating the $\ell_p$ norms of all vectors in the column space of $X$ with a distortion of roughly $O((d\log d)^{1/p})$. To this end we apply an embedding into $\ell_\infty$ using $1/p$-powers of inverse exponential random variables \citep{Andoni17,WoodruffZ13} followed by a dimensionality reduction via sparse random embeddings \citep{ClarksonW17,Cohen16} to obtain a small sketch $\Pi X$.
\begin{lem}[\citet{WoodruffZ13,ClarksonW17}] \label{lem:wcb}
There exists a random embedding matrix $\Pi\in\mathbb{R}^{n'\times n}$ and $\gamma = {O}(d \log(d))$ such that
\begin{align*}%\label{eq:ranmatrix}
\forall \beta \in \mathbb{R}^d: ~ \frac{1}{\gamma^{1/p}} \Vert X \beta \Vert_p
\leq  \Vert \Pi X \beta \Vert_q \leq \gamma^{1/p}\Vert X \beta \Vert_p
\end{align*}
holds with constant probability, where $q=2, n'=O(d^2)$ if $p\in[1,2]$ and $q=\infty, n'=O(n^{1-\frac{2}{p}}\log n(d \log d)^{1+\frac{2}{p}} + d^{5+4p})$ if $p\in(2,\infty)$. For $p=2$ we have $\gamma = 2$. Further $\Pi X$ can be computed in ${O}(\nnz(X))$ time.
\end{lem}
For $p\in [1,2]$, we can obtain $n'=O(d\log d)$ in exchange for an increased running time ${O}(\nnz(X)\log d)$ by using the $\ell_2$ subspace embedding of \citet{Cohen16} as a replacement for \citep{ClarksonW17}. We further note that $\Omega(n^{1-\frac{2}{p}}\log n)$ is necessary for $p>2$ due to tight lower bounds for sketching the $p$th frequency moments \citep{AndoniNPW13}.

Now we decompose the sketch $\Pi X = QR$, which is fast due to its reduced size and argue that $V = XR^{-1}$ is $p$-well-conditioned, leading to the desired $d^{O(p)}$ bound.
\begin{lem}\label{lem:wellcondition}
If $\Pi$ satisfies Lemma \ref{lem:wcb} and $\Pi X=QR$ is the QR-decomposition of $\Pi X$ then $V=XR^{-1}$ is an \emph{$(\alpha,\beta,p)$-well-conditioned basis} for the columnspace of $X$, where for $\gamma = {O}(d \log(d))$ we have
\begin{align*}
(\alpha, \beta)= \begin{cases}
(\sqrt{2 d},\sqrt{2}), & \text{for }p=2\\
(d\gamma^{1/p},\gamma^{1/p}), & \text{for }p \in [1, 2)\\
(d\gamma^{1/p},d \gamma^{1/p}), & \text{for }p \in (2, \infty).
\end{cases}
\end{align*}

\end{lem}
We note that similar methods have been used before in the context of \emph{linear} $\ell_p$ regression \citep{SohlerW11,WoodruffZ13} and for approximating the $\ell_2$ leverage scores \citep{DrineasMMW12}.

\subsection{Proof of Theorem \ref{thm:mainoverview}}
The previous argumentation yields a coreset whose size does not depend on $n$ and that can be calculated efficiently on large data. We prove our main results.
\begin{proof}[Proof of Theorem \ref{thm:mainoverview}]
We first describe the algorithm (Algorithm \ref{alg:main}) and its running time:
First we apply our sketching matrix $\Pi$ from Lemma \ref{lem:wcb} to compute $\Pi X$ in time ${O}(\nnz(X))$ in one pass over the data. The number of rows is $n'={O}(d^2)$ for $p\in[1,2]$ and $n'={O}(n^{1-\frac{2}{p}}\log n\,\poly(d))$ for $p>2$.
Then we calculate the $QR$-decomposition of $\Pi X= QR$ in time ${O}(d^4)$ respectively in time $O(n^{1-\frac{2}{p}}\log n\,\poly(d))$ depending on the value of $p$, which is faster than $O(nd^2)$ without sketching.
In a second pass over the data we compute the row norms $\|v_i\|_p^p=\|x_iR^{-1}\|_p^p$ used in our sampling probabilities.
We set $s_i = \frac{1}{n}+ \beta^p \Vert v_i \Vert_p^p$.
By Lemma \ref{lem:levscorebound}, $S_0=1+(\alpha \beta)^p  $ is an upper bound for $\sum_{i=1}^n s_i$.
Next we set $s_i'=\max\{2^{\lceil \log_2(s_i) \rceil}, \frac{S_0}{n} \}$, i.e. we round $s_i$ to the next power of $2$ such that $S'=\sum_{i=1}^n s_i'\leq 2S_0+n\cdot\frac{S_0}{n}=3S_0$.
As we calculate those values, we can feed the point $x_i$ augmented with the corresponding sampling weight $s_i'$ directly to $k$ independent copies of a weighted reservoir sampler \citep{chao82}. The latter is an online algorithm and updates its sample in constant time.
The second pass takes ${O}(\nnz(X)d+\poly(d))$ time for $p\in[1,2]$, respectively ${O}(\nnz(X)d+\poly(d)\,n^{1-\frac{2}{p}}\log n)$ for $p>2$.
Lemma \ref{lem:sensbound} yields $S=\sum_{i\in [n]} c_s\mu (1/n + u_i) = O(\mu \sum_{i\in [n]} u_i)$ and by Lemma \ref{lem:levscorebound} we have that $\sum_{i\in [n]} u_i \leq (\alpha\beta)^p = d^{O(p)}$ where the values of $\alpha, \beta$ are detailed in Lemma \ref{lem:wellcondition}.
Using Lemmas \ref{lem:VCdim}, \ref{lem:sensbound}, and \ref{lem:levscorebound} to bound the parameters of Proposition \ref{thm:sensitivity} we get for the substitute function $\tilde f$ that 
   $\forall \beta \colon| \tilde{f}_w(X'\beta)- \tilde{f}(X\beta) |\leq \varepsilon \tilde{f}(X\beta)$
with probability at least $1-\delta$. By Corollary \ref{cor:senssplit} and Lemma \ref{lem:I1ass} this implies with high probability that $(X', w)$ is a $7\varepsilon$-coreset for $f$. Folding the constant into $\varepsilon$ completes the proof.
\end{proof}
\begin{proof}[Proof of Corollary \ref{cor:p2}]
If $p=2$ and $d= \omega(\ln n)$, we can use a Johnson–Lindenstrauss transform, i.e., a matrix $G \in \mathbb{R}^{d \times m}$ where $m={O}(\ln(n))$ and whose entries are i.i.d. $G_{ij}\sim N(0, \frac{1}{m})$ \citep{JohnsonL84} to compute a $\frac{1}{2} $-approximation to the row norms: We have $\|v_i'\|_2^2:=\Vert x_i(R^{-1}G) \Vert_2^2\geq {\Vert x_iR^{-1} \Vert_2^2}/{2}$ for all $i \in [n]$ simultanously with constant probability. The running time reduces to ${O}(\nnz(X)\ln(n)+\poly(d))$.
The online algorithm (Algorithm \ref{alg:online}) is obtained by running the online $\ell_2$ leverage score algorithm of \citet{ChhayaC0S20} that recently extended the previous work of \cite{CohenMP20}. Each row update takes $O(d^2)$ time except for at most $O(d)$ updates that take $O(d^3)$ time, implying $O(nd^2 + \poly(d))$ total running time. %TODO $(\sigma(X))$ 
The slightly increased coreset size results from an increase of the total sensitivity by at most $\log(\|X\|_2)$ due to the online procedure \citep{ChhayaC0S20}.
\end{proof}

\section{EMPIRICAL EVALUATION}

Our intention is to corroborate our theoretical results by investigating the following questions empirically:

\begin{itemize}
    \item[\textbf{(Q1)}]How does the $1$-probit model compare to logistic regression?
    \item[\textbf{(Q2)}]How do the $p$-probit models for different values of $p$ compare to one another?
    \item[\textbf{(Q3)}] How accurate are the maximum likelihood (ML) estimators obtained from $p$-probit coresets?
    \item[\textbf{(Q4)}] How fast can we obtain an accurate ML estimator from $p$-probit coresets?
\end{itemize}

All experiments were conducted on an AMD Ryzen 7 2700x processor (8 cores, 3.7GHz, 16GB RAM). 

Our results can be reproduced with our 
open Python implementation available at \url{https://github.com/cxan96/efficient-probit-regression}.
More information on data sets, pseudo code and several plots are in Appendix \ref{app:experiments}.

\textbf{I. Statistical Modeling Aspects:} For \textbf{(Q1)} \& \textbf{(Q2)} no data reduction is applied. Specifically, only statistical modeling aspects are studied here.

\textbf{(Q1)} We investigate how closely the $1$-probit model estimate equals logistic regression.
This seems to be a natural question since both functions have similar tail behaviors as we have argued before in the introduction regarding their link functions. The similarities naturally extend to their loss functions. We observe that the logistic and $1$-probit loss functions both converge to a linear function for positive arguments and to the exponential function for negative arguments. Near zero, however, the two loss functions differ more significantly. 
We compare the coefficients obtained for the two models on full data sets from public repositories, see Figure \ref{fig:2d-example}. % and \ref{fig:2d-example:appendix}.
The coefficients $\beta_i$ are very close to each other for the Covertype and Kddcup data. This is what we expected since the tails of the 1-generalized normal (Laplace) distribution and the logistic distribution have a very similar exponential decay.
For Webspam, however, there are few single values of $\beta_i$ that differ more significantly for the two models. Still, most coefficients are very similar and all less similar values have the same sign, so they point in the same direction regarding the influence of single variables.
The deviating cases indicate that many points are located close to the separating hyperplane, where the two distributions deviate most from each other. This impression is affirmed in Figure \ref{fig:2d-example}, where we see the resulting linear separators and (misclassified) residuals for a 2D data example. The logit regression line passes through a whole bunch of points, and as a result the logit model is more strongly attracted to the blue outliers, being closer to the cases $p=1.5$ or $p=2$ than to $p=1$. The logistic model thus seems to adapt to different situations, while the $1$-probit model always yields a robust estimator.

\begin{figure}[t!]
    \centering
    \includegraphics[width=0.49\linewidth]{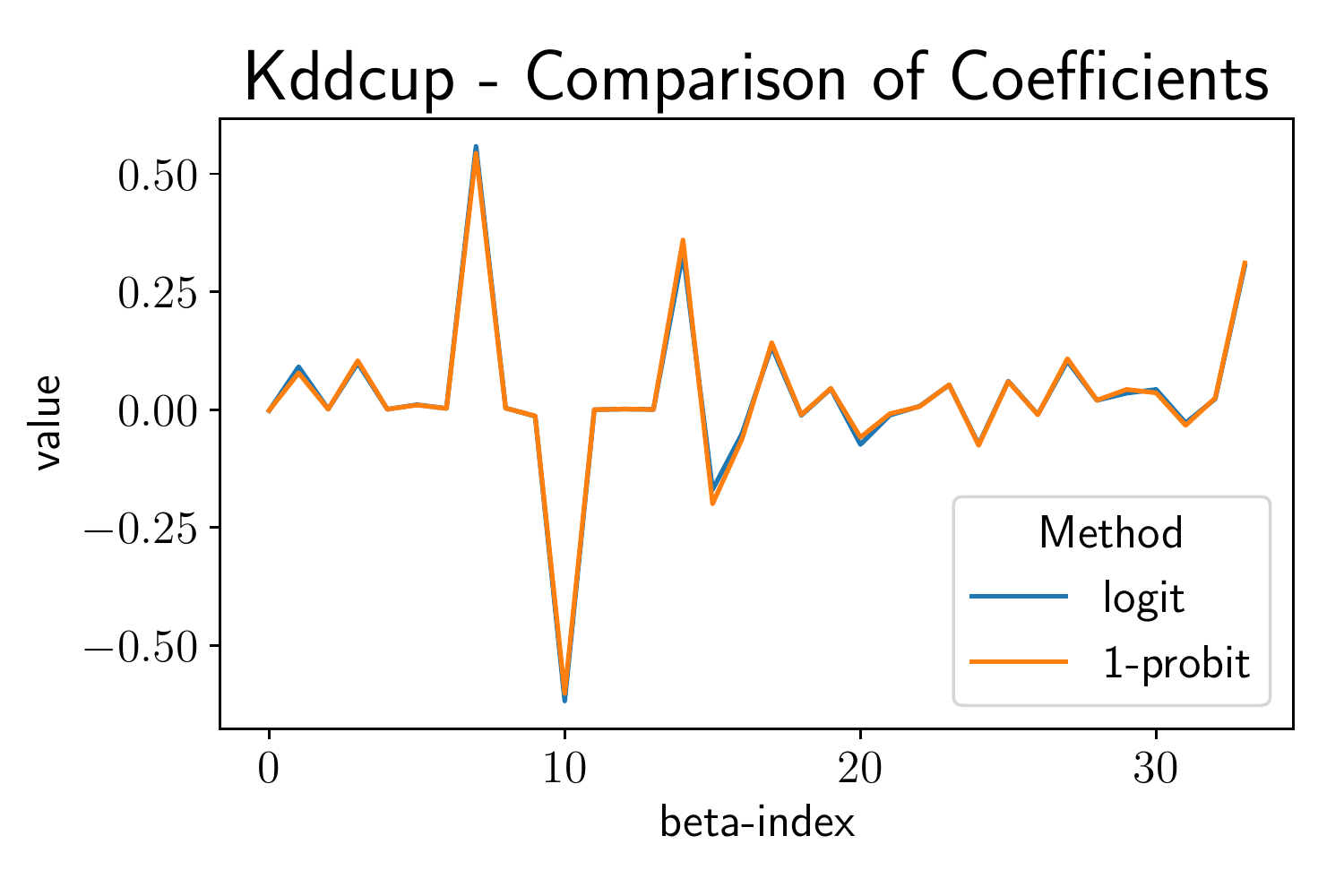}
    \includegraphics[width=0.49\linewidth]{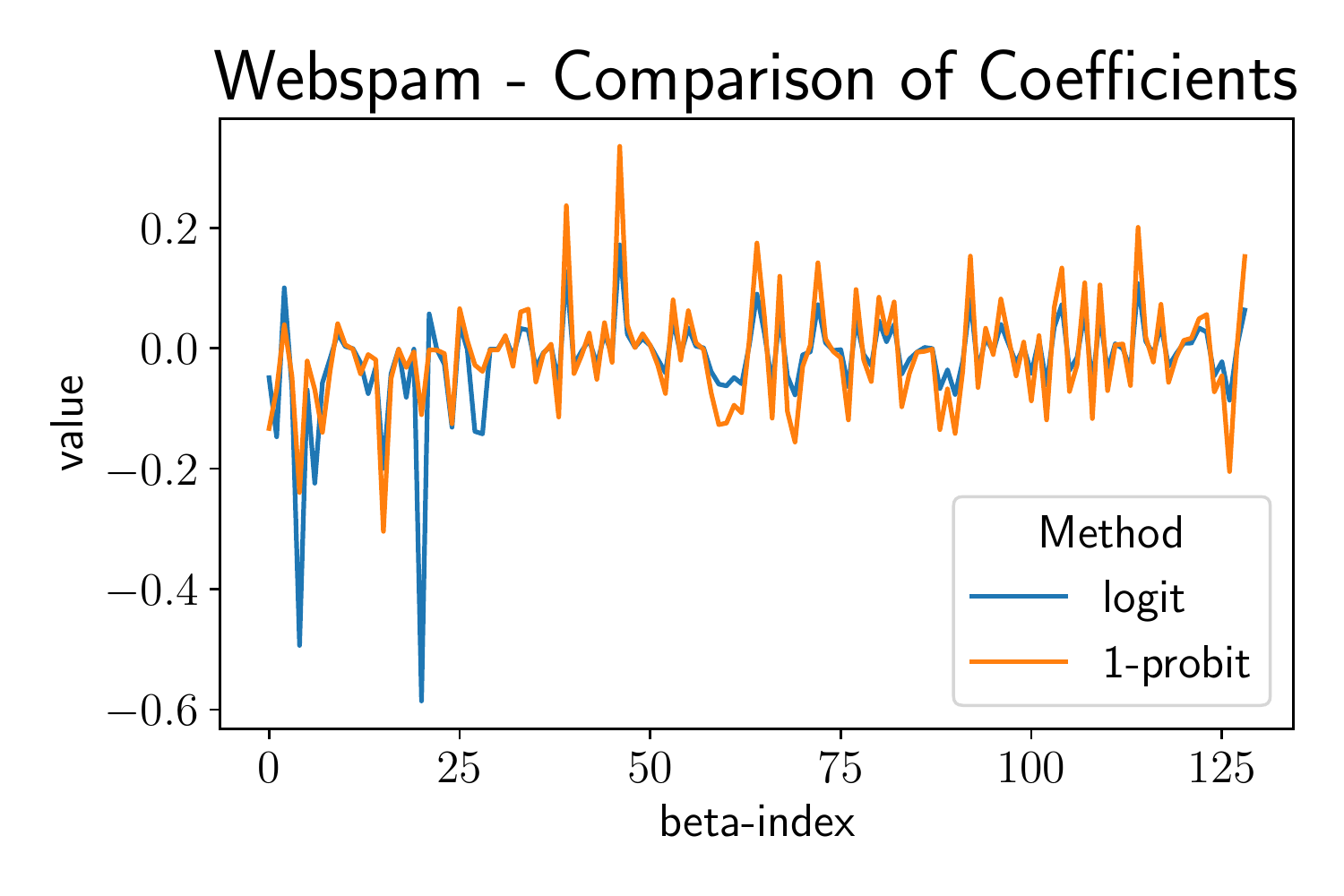}\\
    \includegraphics[width=.49\linewidth]{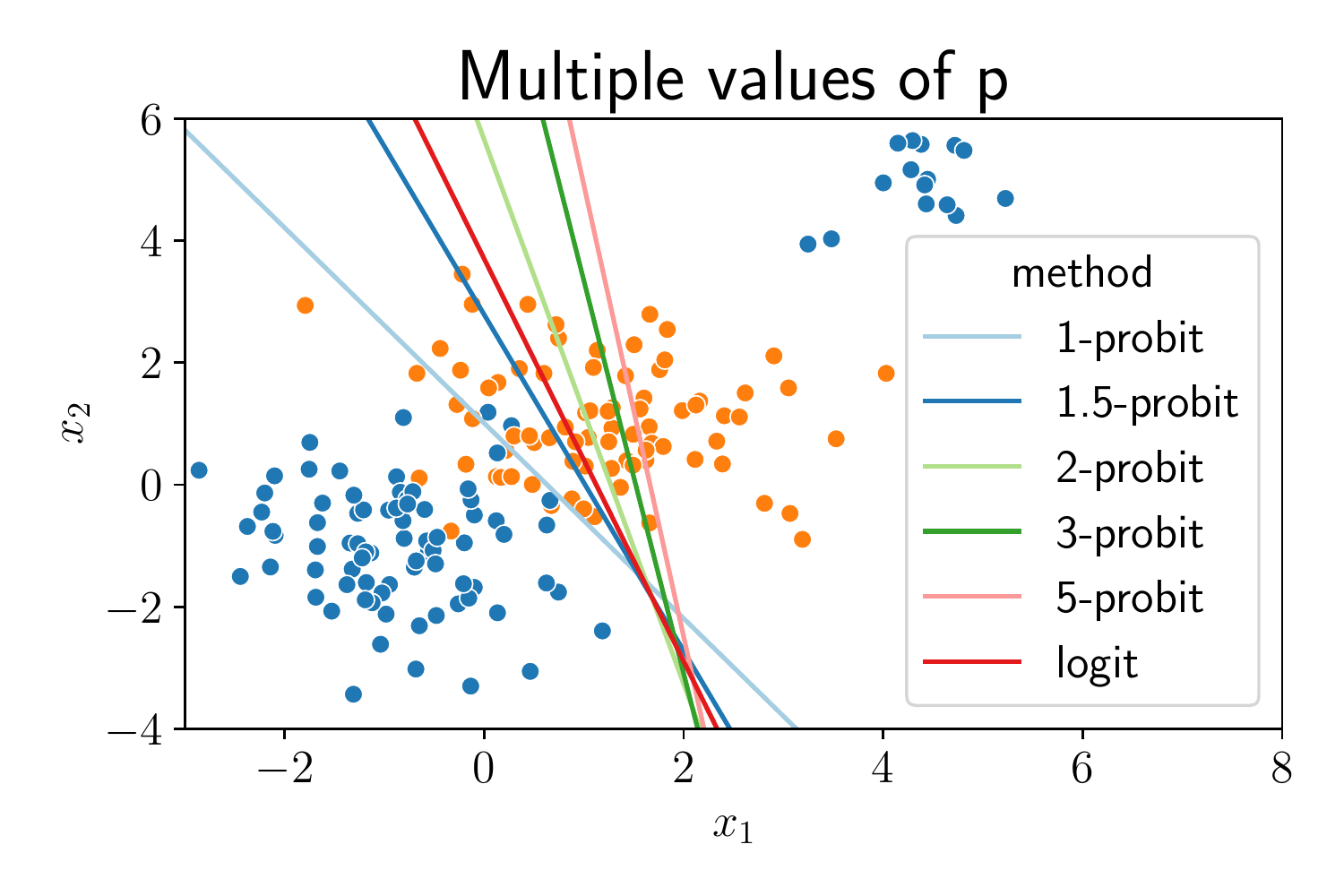}
    \includegraphics[width=.49\linewidth]{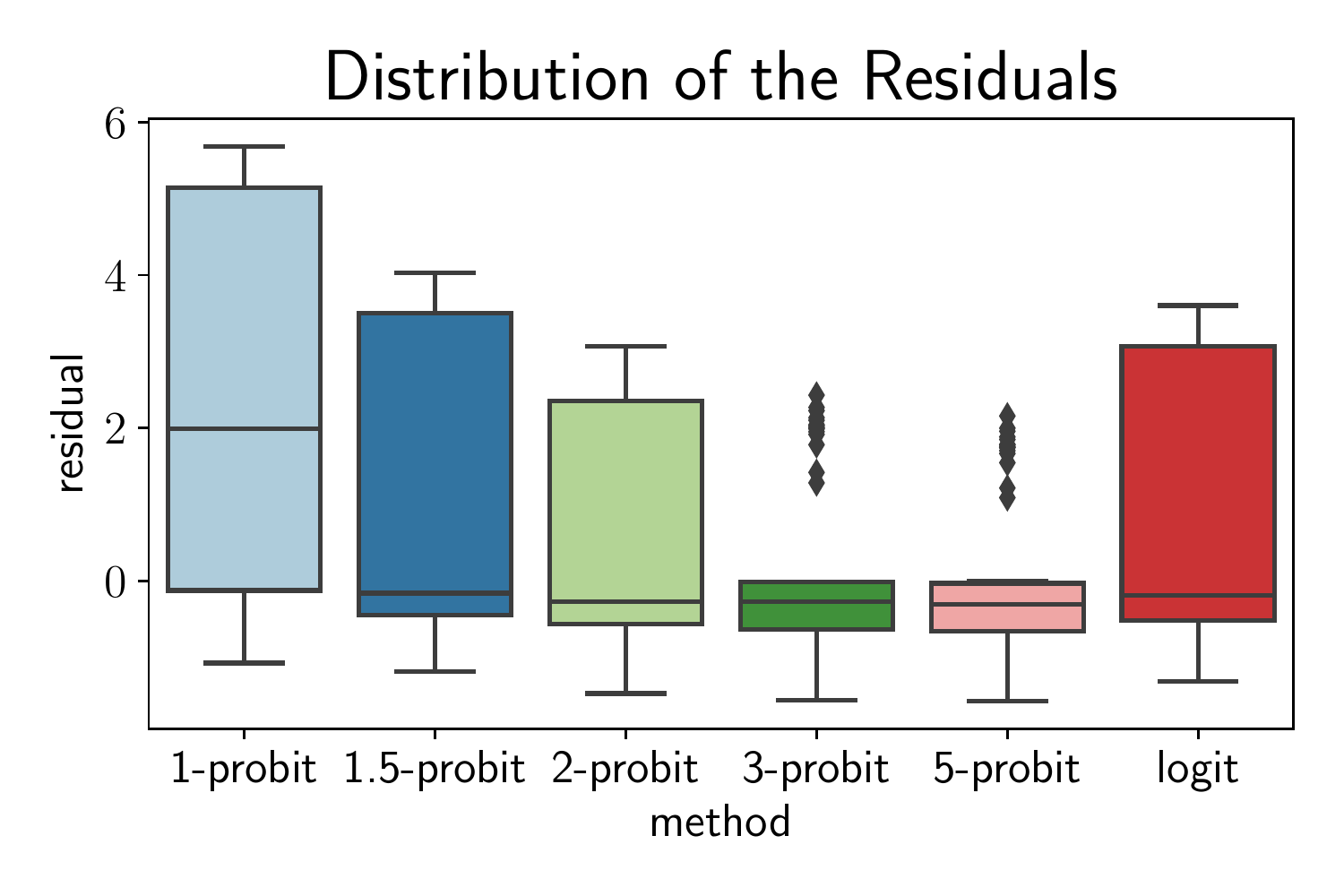}
    \caption{\textbf{(Q1)} \& \textbf{(Q2)}: (top) Comparison of normed coefficients for the logit vs. the 1-probit model on different data sets. (bottom) A 2D example data set that demonstrates how different values of $p$ affect the linear separator, as well as the distribution of the residuals $X{\beta}/{\|\beta\|_2}$ of misclassifications.}
    \label{fig:2d-example}
\end{figure}

\textbf{(Q2)} We assess the difference between $p$-probit models for different values of $p\in\{1,1.5,2,3,5\}$ on a 2D data set (plus intercept). In Figure \ref{fig:2d-example} we see that the case $p=1$ is most robust, fitting the majority of the points, and ignoring the outliers. With an increasing value of $p$, the model becomes more sensitive and attracted to the outliers. The boxplots show the residuals of the misclassified points. For small values of $p\in\{1,1.5,2\}$ the model decreasingly ignores the blue outliers, which lead to large positive residuals. There are no orange outliers, so the negative residuals are relatively small. With increasing $p\in\{3,5\}$ the model gradually tends to minimize the distance to the most outlying misclassifications, independent of what happens between those extremes. We finally note that we have tried larger values but the effect does not change significantly beyond $p=5$, which indicates that relatively small values are already close to the limiting case $p\rightarrow\infty$.

\begin{figure*}[t!]
\begin{center}
\begin{tabular}{ccc}
\includegraphics[width=0.309\linewidth]{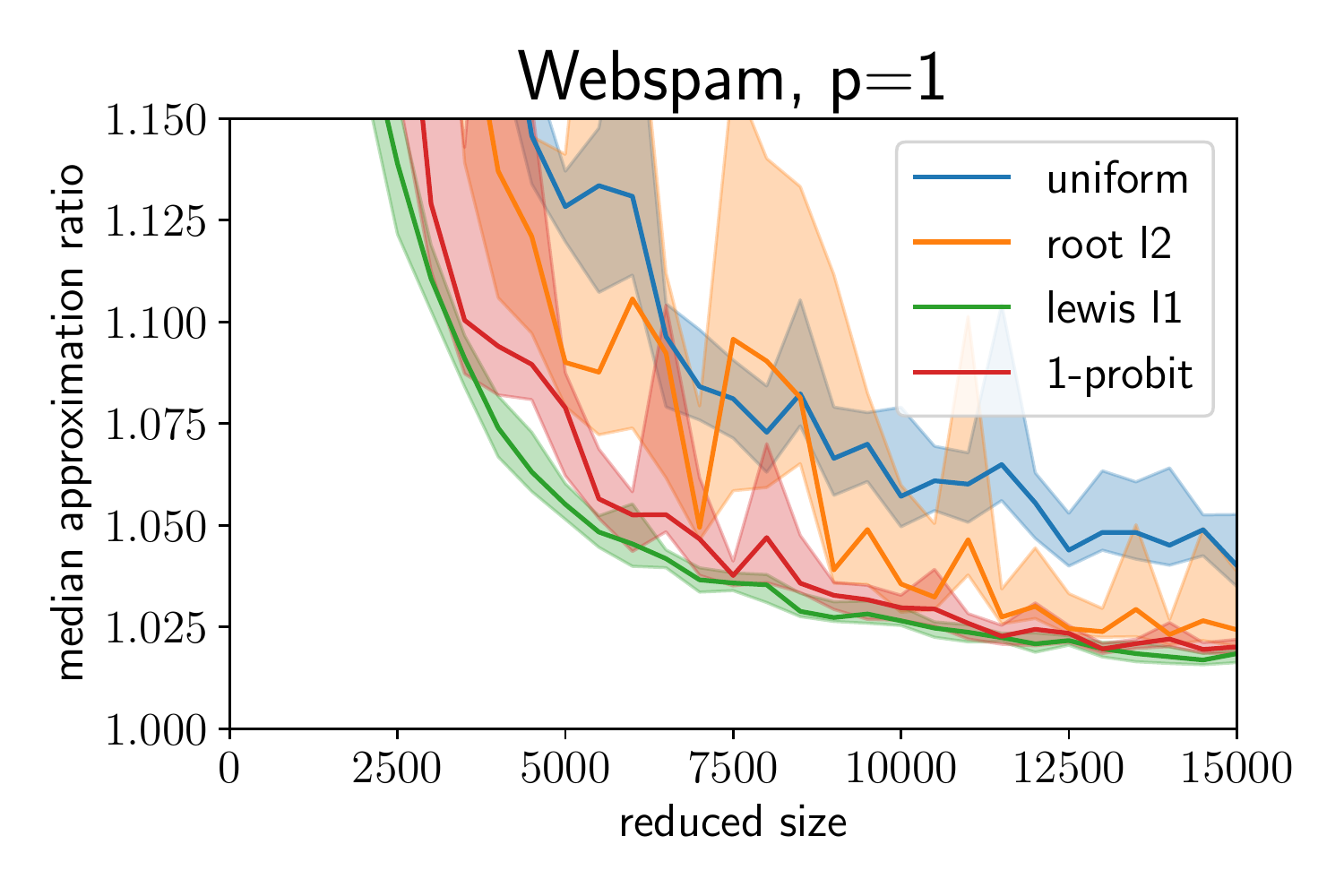}&
\includegraphics[width=0.309\linewidth]{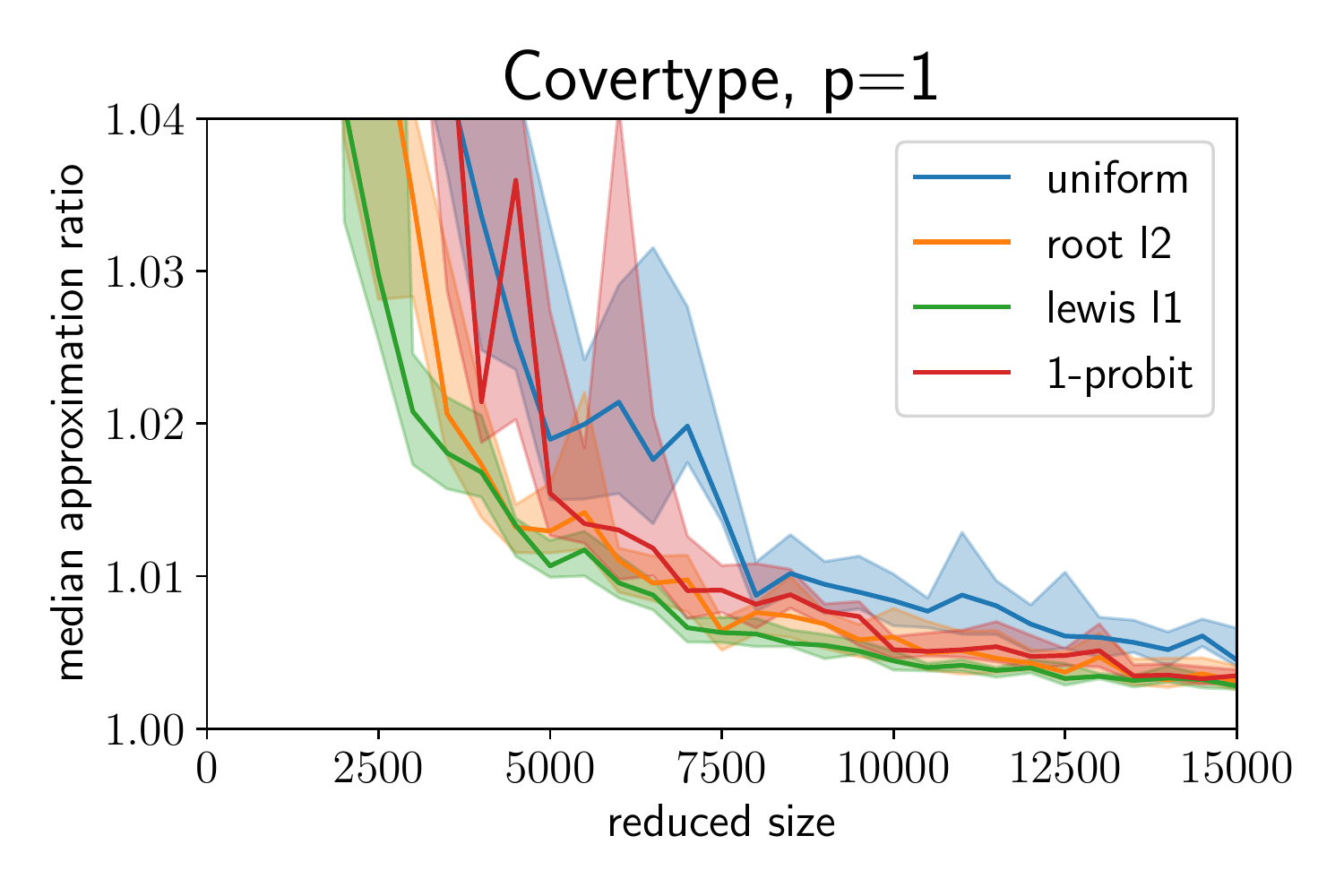}&
\includegraphics[width=0.309\linewidth]{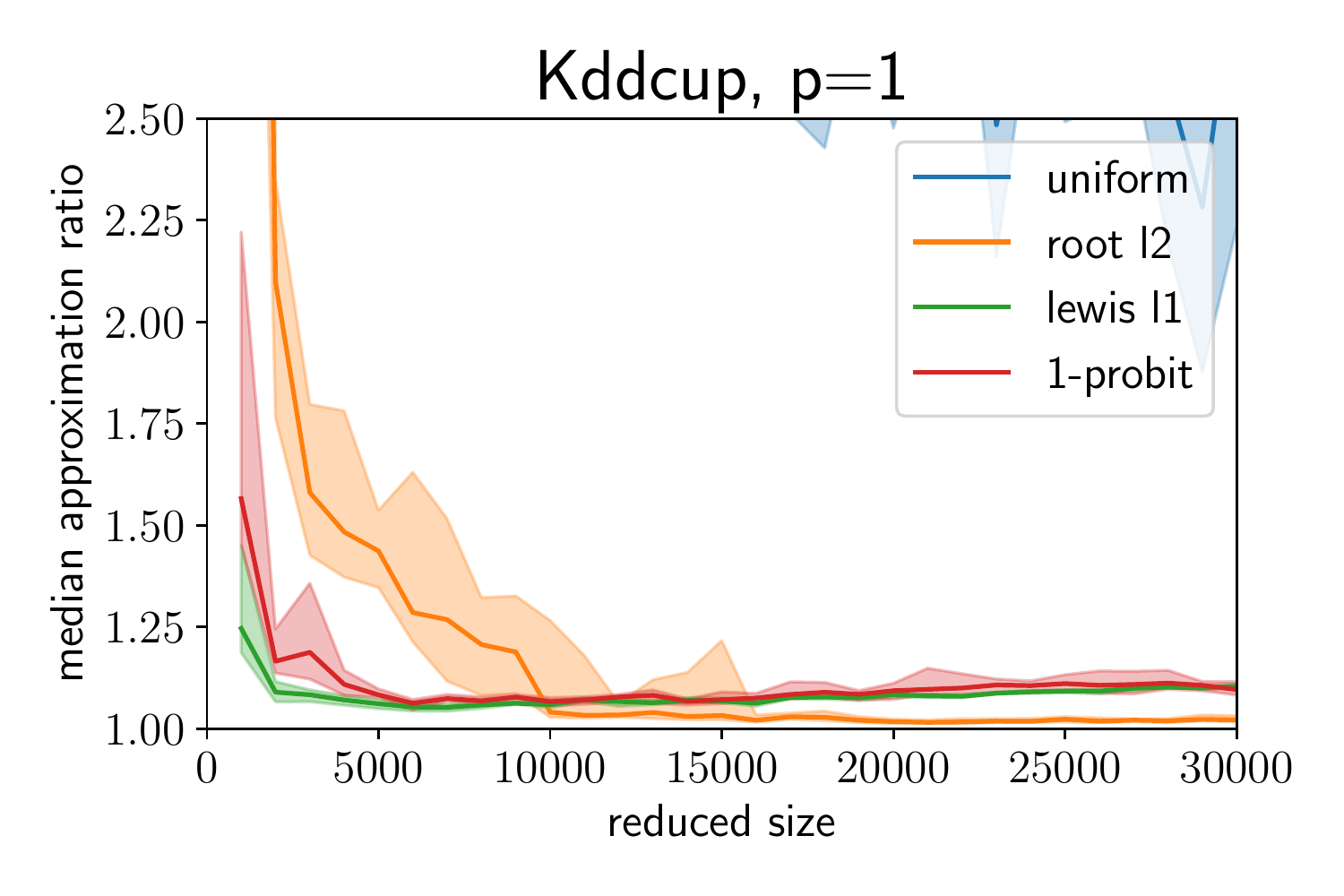}\\
\includegraphics[width=0.309\linewidth]{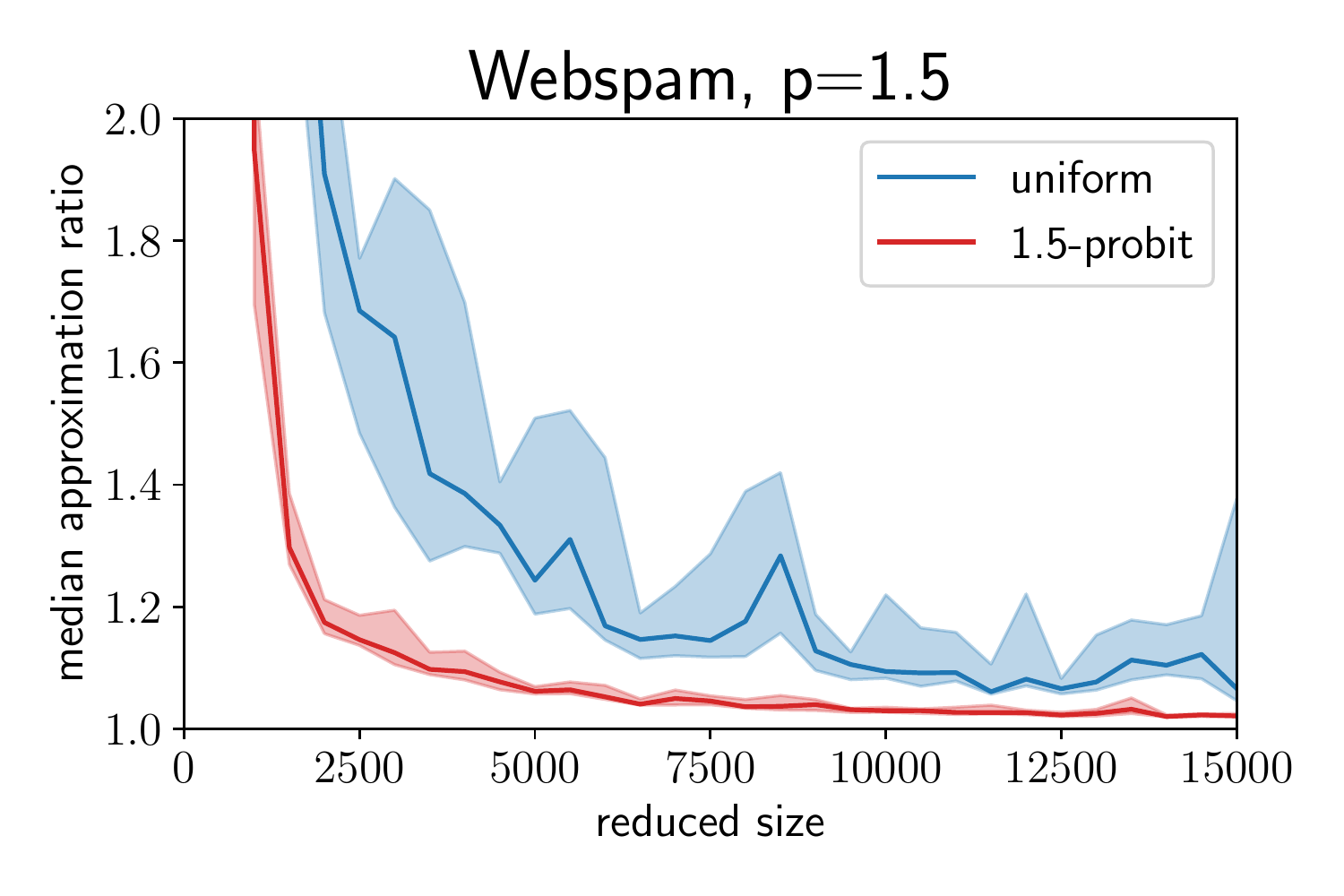}&
\includegraphics[width=0.309\linewidth]{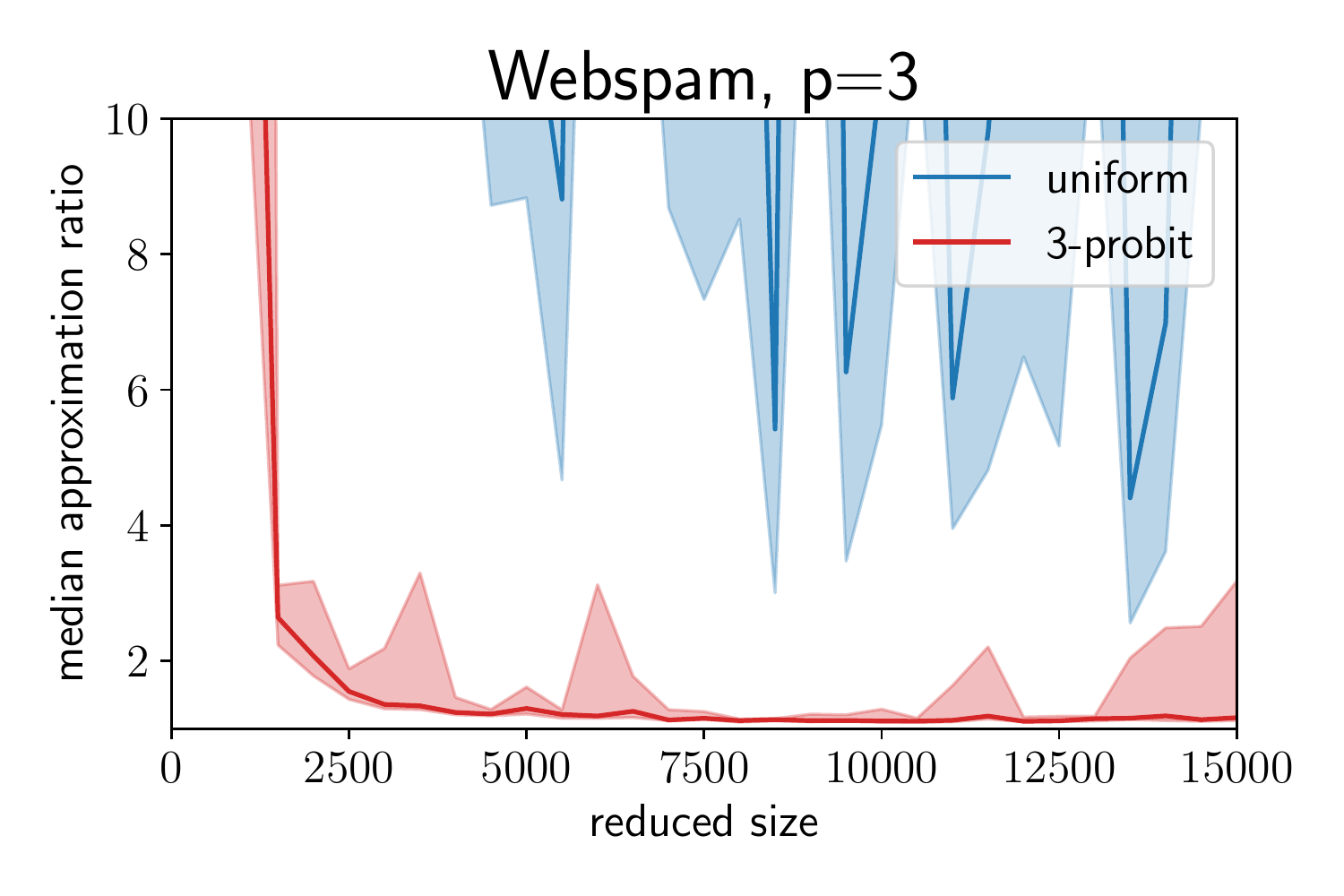}&
\includegraphics[width=0.309\linewidth]{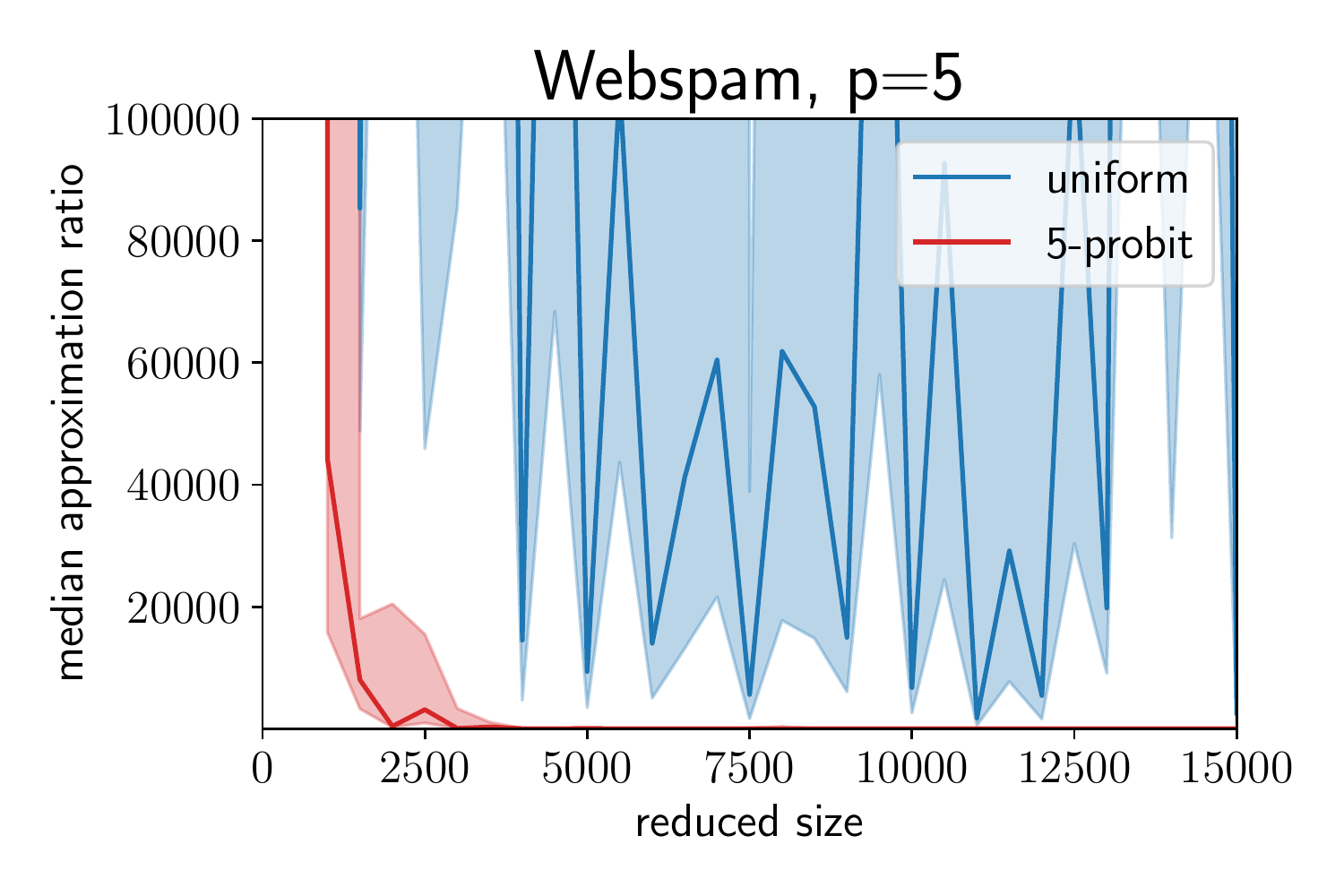}\\
\includegraphics[width=0.309\linewidth]{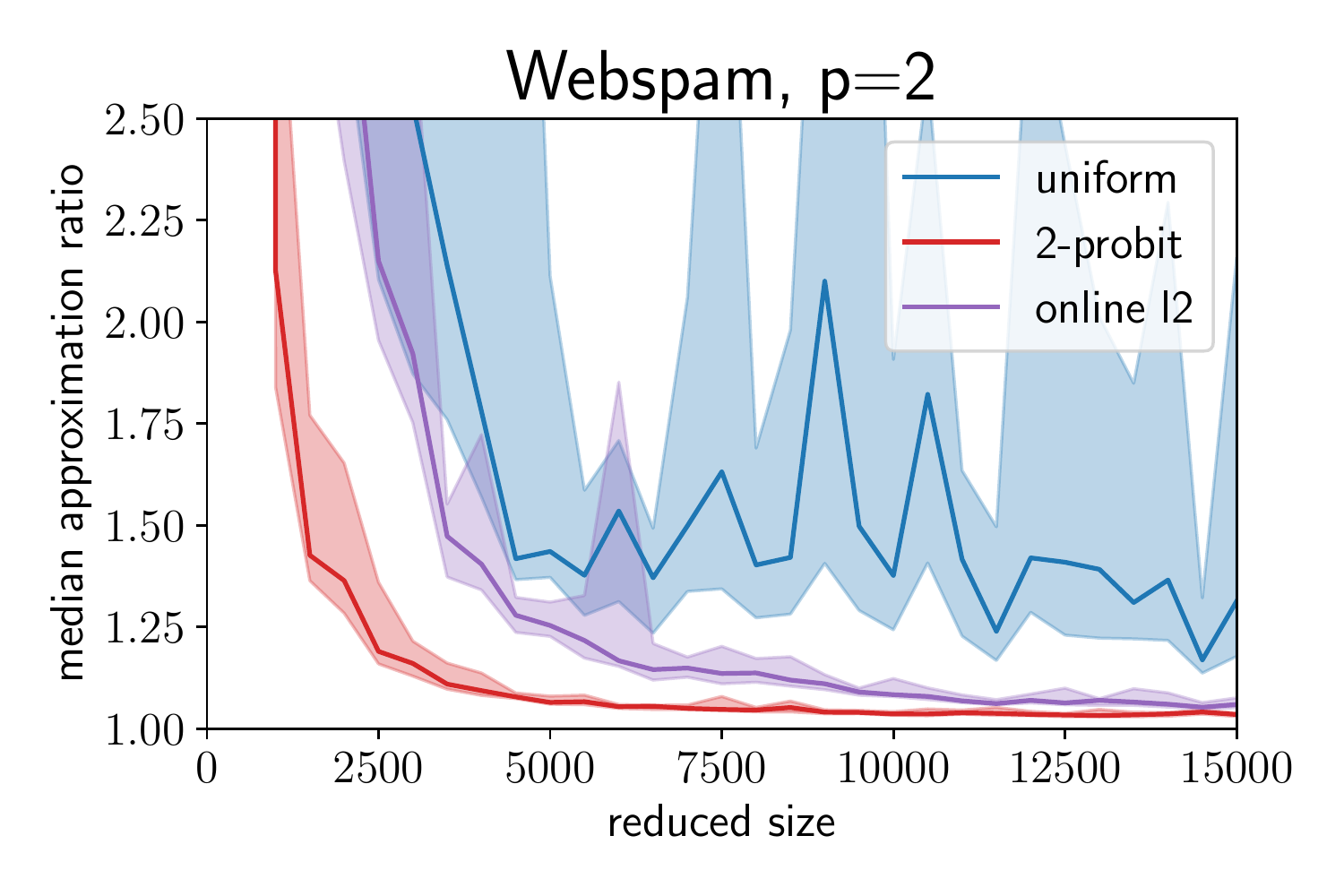}&
\includegraphics[width=0.309\linewidth]{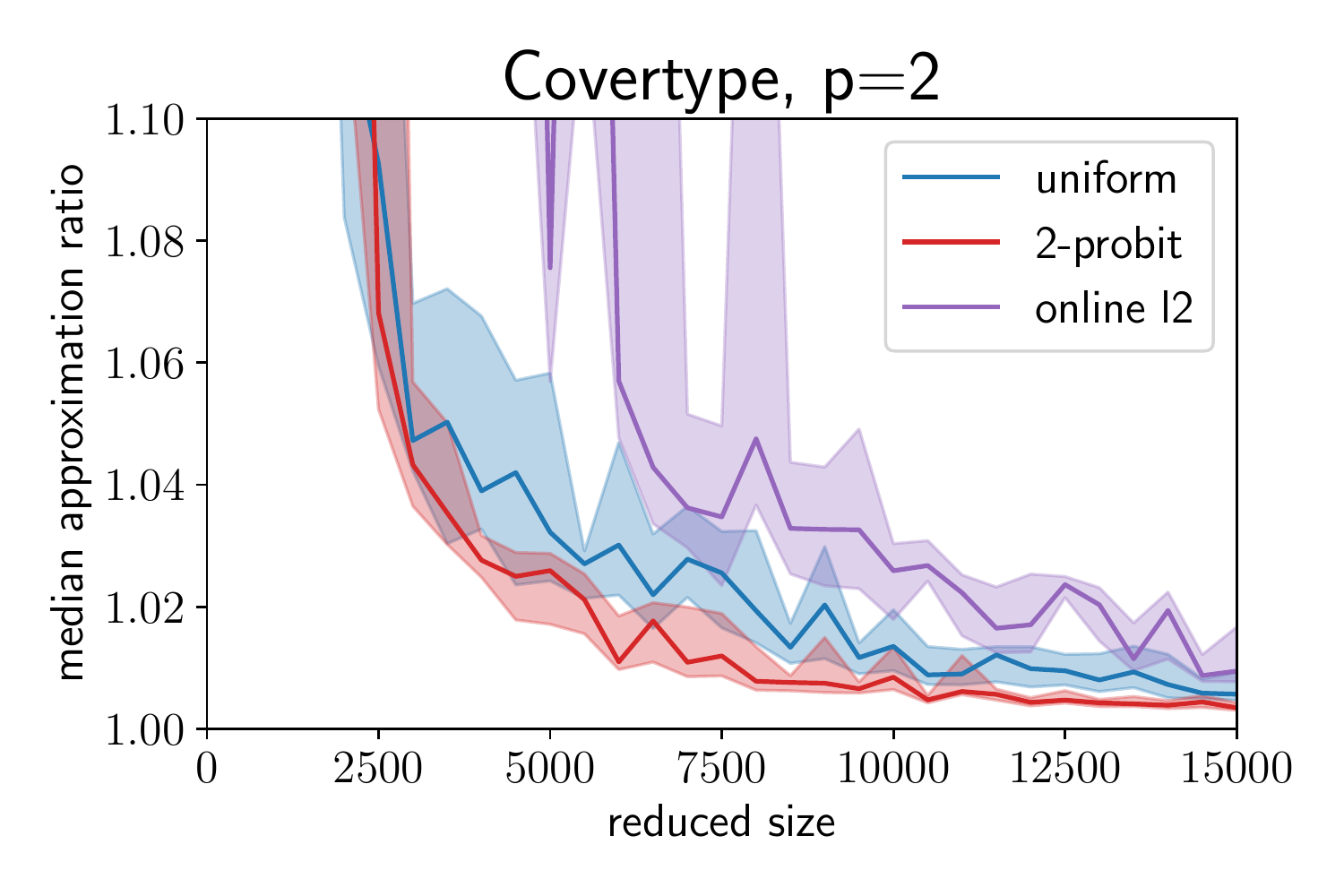}&
\includegraphics[width=0.309\linewidth]{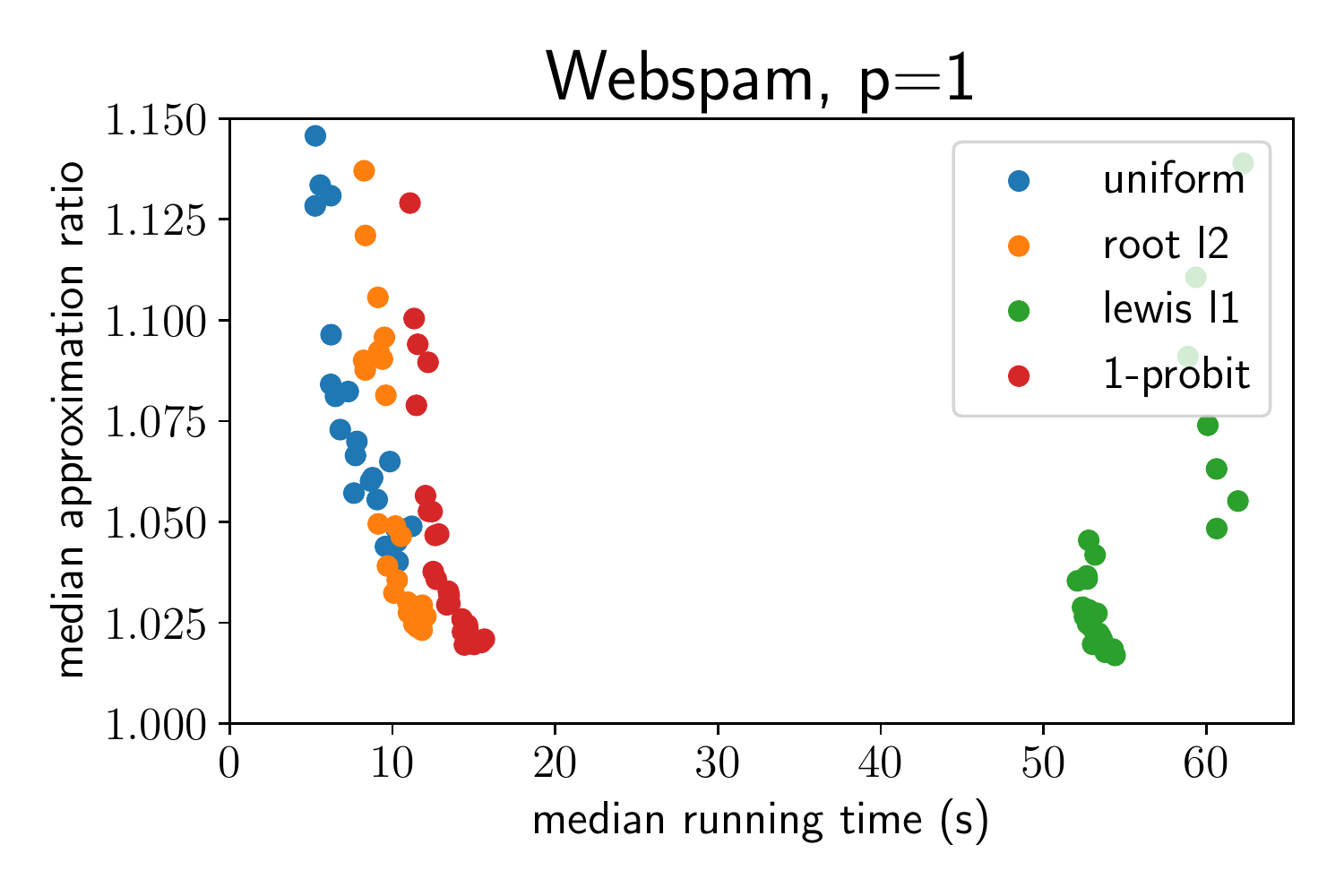}%\\
\end{tabular}
\caption{\textbf{(Q3)} \& \textbf{(Q4)}: (top) Accuracy of coreset constructions for $p=1$ on three data sets.
(middle) Accuracy on the Webspam data for increasing values of $p$.
(bottom left) Accuracy of coreset constructions for $p=2$ on two data sets. (bottom right) Time vs. accuracy plot for different coreset constructions for $p=1$. The median approximation ratio denotes for each sample size the median of the approximation ratios over all repetitions.
}
\label{fig:strct_exp}
\end{center}
\end{figure*}

\textbf{II. Accuracy and Algorithmic Efficiency Aspects:}
We compare our coresets to uniform sampling and for $p=1$ additionally to the square root of $\ell_2$ leverage scores and to $\ell_1$-Lewis weights, that were developed for logistic regression \citep{MunteanuSSW18,MaiRM21}. 
For $p=2$ we additionally compare to the online algorithm for online $\ell_2$ leverage scores. See Figure \ref{fig:strct_exp}.

\textbf{(Q3)} For $p=1$ the Lewis weights have the best accuracy, closely followed by the $1$-probit coreset. Square root $\ell_2$ sampling is third on Webspam and Covertype but seems superior and more stable on Kddcup. Uniform sampling is good only on Covertype (which seems to be a very uniform data set) but otherwise uniform sampling is considerably worse than all competitors.
For $p=2$ we see that the online leverage score approximation performs worse than uniform sampling on some data sets but is close to the probit model on others. This might depend on the condition of the data as indicated by the theoretical analyses of online leverage scores \citep{ChhayaC0S20,CohenMP20}.
Uniform sampling generally detoriorates with increasing $p$ since it misses the small number of outliers which become increasingly important for obtaining a good approximation. Only for Covertype it seems to be a good sampling probability independent of $p$ due to the uniformity of the data set. For more complicated Webspam and especially for Kddcup, uniform sampling fails. Our coresets perform well for all values of $p$, although for Covertype we note that for $p>2$ the convergence seems to start very late. We believe this is due to the considerably higher complexity of sketching and consequently weaker estimates of the norms for $p>2$.

\textbf{(Q4)} We fist note that the online algorithm is extremely slow, running in $\Theta(nd^2+ \poly(d))$; %TODO
its advantage is that it needs only one single row-order pass as opposed to the other two-pass algorithms. Although Lewis weights won the competition with respect to accuracy, we note that they require repeated calculations of leverage scores and are thus much slower, which is also reflected when considering the time vs. accuracy trade-off. Uniform sampling seems to be the fastest choice when we only need a weak approximation, but our coresets are the best choice to obtain both, fast and accurate estimates at the same time.

In summary, the answers to questions \textbf{(Q1)}-\textbf{(Q4)} show the benefits and limitations of our methods.

\section{CONCLUSION}
We introduce the $p$-generalized probit model as a flexible framework for modeling binary data, generalizing the standard probit model to control the tail behavior (kurtosis) of the link, and to model different levels of sensitivity to outliers. To facilitate an efficient and scalable maximum likelihood estimation for the parameters of the $p$-generalized probit model, we develop coreset constructions and combine with sketching techniques for the associated $p$-probit loss function. Hereby, we advance the analytical analysis on the tails of the $p$-generalized normal distribution, which may be of independent interest. Our algorithms run in nearly input sparsity time in two passes over the data which can be improved for $p=2$, in which case we also have a single pass online algorithm. Our experiments demonstrate the usefulness of our $p$-generalization to model the sensitivity to outliers, as well as the efficiency and accuracy of our coreset and sketching algorithms. It would be interesting to continue this work by extending modeling flexibility to allow for skewed distributions.

\subsubsection*{Acknowledgements}
We thank the anonymous reviewers of AISTATS 2022 for their helpful comments. We thank Prof. Dr.~Katja~Ickstadt for pointing us to the probit model and for valuable discussions on that topic.
We thank the authors of \citep{MaiRM21} for providing their implementation of Lewis weights integrated into our codebase.
The authors were supported by the German Science Foundation (DFG), Collaborative Research Center SFB 876, project C4 and by the Dortmund Data Science Center (DoDSc).

\renewcommand\bibname{References}
\bibliographystyle{plainnat}
\bibliography{references}
\clearpage
\onecolumn
\allowdisplaybreaks
\appendix

\section{GRADIENT AND HESSIAN MATRIX}\label{app:grad_hess}

The optimization of $f$ can be done by gradient descent or by applying the Newton-Raphson method \citep{Bubeck15}, an iterative procedure that starts at an initial guess $\beta^{(0)}$ and successively applies the following update rule:
\begin{align*}
    \beta^{(t)} = \beta^{(t-1)} - \left(\frac{\partial^2f(\beta^{(t-1)})}{\partial\beta\partial\beta^T}\right)^{-1}
    \cdot \frac{\partial f(\beta^{(t-1)})}{\partial\beta},
\end{align*}
where $\left(\frac{\partial^2f(\beta^{(t-1)})}{\partial\beta\partial\beta^T}\right)^{-1}$ refers to the inverse of the Hessian matrix of $f$, evaluated at $\beta^{(t-1)}$, and $\frac{\partial f(\beta^{(t-1)})}{\partial\beta}$ refers to the gradient of $f$, evaluated at $\beta^{(t-1)}$. The idea behind this procedure is, broadly speaking, to approximate $f$ locally around $\beta^{(t)}$ by its second degree Taylor-polynomial and then analytically find the minimum of this polynomial. The minimum of this local polynomial approximation of $f$ is then used iteratively as a basis for the next step of the Newton-Raphson algorithm.

It remains to derive the gradient and the Hessian matrix of $f$. Since $f$ is a sum of the function $g$ evaluated at different points, it makes sense to first determine the derivative of $g$. 
To this end $\varphi_p(r)$ is the density function of the (standardized) $p$-generalized normal distribution function, cf. \citep{Dytso18,KalkeR13}:
\begin{align*}
    \varphi_p(r) = \frac{p^{1-{1}/{p}}}{2\Gamma({1}/{p}) } \exp(-|r|^p/p).
\end{align*}
We proceed by using the chain rule as follows:
\begin{align*}
    %\begin{split}
        \frac{d}{dr}g(r) = \frac{d}{dr} - \ln \left( \Phi_p(-r)\right)
        &= \frac{d}{dr} \ln \left(\frac{1}{1 - \Phi_p(r)}\right)                 \\
        & = (1 - \Phi_p(r)) \cdot \frac{d}{dr} \left(\frac{1}{1 - \Phi_p(r)}\right) \\
        & = (1 - \Phi_p(r)) \cdot \frac{(-1)}{(1 - \Phi_p(r))^2} \cdot \frac{d}{dr} (1 - \Phi_p(r)) \\
        & = \frac{(-1)}{1 - \Phi_p(r)} \cdot (-1) \cdot \varphi_p(r) \\
        & = \frac{\varphi_p(r)}{1 - \Phi_p(r)},
    %\end{split}
\end{align*}
We can use this result to calculate the gradient of $f$:
\begin{align*}
        \frac{\partial}{\partial \beta} f_w(X\beta)
        & = \frac{\partial}{\partial \beta} \sum_{i=1}^n w_i g(x_i \beta) \\
        & = \sum_{i=1}^n w_i x_i g'(x_i \beta) \\
        & = \sum_{i=1}^n w_i x_i \frac{\varphi_p(x_i \beta)}{1 - \Phi_p(x_i \beta)}
\end{align*}

Next, we need to determine the Hessian matrix of $f$. To this end, we again start by finding the second derivative of $g$, this time using the quotient rule.
\begin{align*}
        \frac{d^2}{dr^2}g(r)
        & = \frac{d}{dr} \frac{\varphi_p(r)}{1 - \Phi_p(r)} \\
        & = \frac{\varphi_p'(r)(1 - \Phi_p(r)) - \varphi_p(r) \cdot (-1) \cdot \varphi_p(r)}
        {(1 - \Phi_p(r))^2} \\
        & = \frac{(-1) \cdot \sgn(r) \cdot |r|^{p-1} \cdot \varphi_p(r)(1 - \Phi_p(r)) - \varphi_p(r) \cdot (-1) \cdot \varphi_p(r)}
        {(1 - \Phi_p(r))^2} \\
        & = \frac{[\varphi_p(r)]^2 -\sgn(r) |r|^{p-1} \cdot \varphi_p(r) \cdot (1 - \Phi_p(r))}{(1 - \Phi_p(r))^2} \\
        & = \left(\frac{\varphi_p(r)}{1 - \Phi_p(r)}\right)^2 - \sgn(r) |r|^{p-1} \cdot \frac{\varphi_p(r)}{1 - \Phi_p(r)} \\
        & = \frac{\varphi_p(r)}{1 - \Phi_p(r)} \left( \frac{\varphi_p(r)}{1 - \Phi_p(r)} - \sgn(r) |r|^{p-1} \right)  \\
        & = g'(r) \cdot (g'(r) - \sgn(r) |r|^{p-1})
\end{align*}
We can now use this result to find the Hessian matrix of $f$.
Recall that $x_i$ are \emph{row vectors} in our paper and thus each $x_i^T x_i $ is a $d \times d$-matrix.
\begin{align*}
        \frac{\partial^2}{\partial \beta \partial \beta^T} f_w(X\beta)
        & = \sum_{i=1}^n
        \frac{\partial^2}{\partial \beta \partial \beta^T} w_i g(x_i \beta)\\
        & = \sum_{i=1}^n w_i x_i^T x_i g'(x_i \beta)(g'(x_i \beta) - \sgn(x_i \beta)|x_i \beta|^{p-1})\\
        & = \sum_{i=1}^n w_i x_i^T x_i
        \frac{\varphi_p(x_i \beta)}{1 - \Phi_p(x_i \beta)} \left( \frac{\varphi_p(x_i \beta)}{1 - \Phi_p(x_i \beta)} - \sgn(x_i \beta)|x_i \beta|^{p-1} \right).
\end{align*}

It can be shown, that $f_w(X\beta)$ is a convex function of $\beta$, and that the Newton-Raphson algorithm converges to a global optimum when applied to a convex function \citep{Bubeck15}. The optimization procedure thus converges to the maximum likelihood estimate $\hat{\beta} \in \operatorname{argmin}_{\beta\in\mathbb{R}^d} f_w(X\beta)$ provided it exists.%under the condition that the data is not linearly separable.

\section{SENSITIVITY FRAMEWORK}\label{app:sensitivity}
Our approach is based on the so called sensitivity framework \citep{LangbergS10}. The sensitivity of a point is its worst case contribution to the entire loss function $\zeta_i=\sup_{\beta \in \mathbb{R}^d} \frac{g(x_i\beta)}{f(X\beta)}$, cf. Definition \ref{def:sensitivity}.
Computing the sensitivity of a point can be difficult. Fortunately it suffices to get a reasonably tight upper bound on the sensitivity of a point. In Section \ref{sec:sen} we will show that such a bound can be derived in our setting.

Since sensitivities are not sufficient to get a good bound for all solutions $\beta$, we also need the terminology of the VC-dimension:

\begin{mydef}
	The range space for a set $\mathcal{F}$ is a pair $\mathfrak{R}=(\mathcal{F},\ranges)$ where $\ranges$ is a family of subsets of $\mathcal{F}$. The VC-dimension $\Delta(\mathfrak{R})$ of $\mathfrak{R}$ is the size $|G|$ of the largest subset $G\subseteq \mathcal{F}$ such that $G$ is shattered by $\ranges$, i.e., $\left| \{G\cap R\mid R\in \ranges \} \right| = 2^{|G|}.$
\end{mydef}
\begin{mydef}
	Let $\mathcal{F}$ be a finite set of functions mapping from $\mathbb{R}^d$ to $\mathbb{R}_{\geq 0}$. For every $x\in\mathbb{R}^d$ and $r\in \mathbb{R}_{\geq 0}$, let $\rng{\mathcal{F}}(x,r) = \{ f\in \mathcal{F}\mid f(x)\geq r\}$, and $\ranges(\mathcal{F})=\{\rng{\mathcal{F}}(x,r)\mid x\in\mathbb{R}^d, r\in \mathbb{R}_{\geq 0} \}$, and $\mathfrak{R}_{\mathcal{F}}=(\mathcal{F},\ranges(\mathcal{F}))$ be the range space induced by $\mathcal{F}$.
\end{mydef}

The VC-dimension can be thought of something similar to the dimension of our problem.
For example the VC-dimension of the set of hyperplane classifiers in $\mathbb{R}^d$ is $d+1$ \citep{KearnsV94}. We analyze the VC-dimension of our problem in Section \ref{sec:VC}.
The sensitivity scores were combined with a theory on the VC-dimension of range spaces in \citep{FeldmanL11,BravermanFL16}. We use a more recent version of \citet{FeldmanSS20}.

\begin{pro}{\citep{FeldmanSS20}}
	\label{thm:sensitivity}
	Consider a family of functions $\mathcal{F}=\{f_1,\ldots,f_n\}$ mapping from $\mathbb{R}^d$ to $[0,\infty)$ and a vector of weights $u\in\mathbb{R}_{> 0}^n$. Let $\varepsilon,\delta\in(0,1/2)$.
	Let $s_i\geq \zeta_i$.
	Let $S=\sum\nolimits_{i=1}^{n} s_i \geq \sum\nolimits_{i=1}^{n} \zeta_i = Z$. Given $s_i$ one can compute in time $O(|\mathcal{F}|)$ a set $R\subset \mathcal{F}$ of $$O\left( \frac{S}{\varepsilon^2}\left( \Delta \ln S + \ln \left(\frac{1}{\delta}\right) \right) \right)$$ weighted functions such that with probability $1-\delta$, we have for all $x\in \mathbb{R}^d$ simultaneously $$\left| \sum_{f_i\in \mathcal{F}} u_i f_i(x) - \sum_{f_i\in R} w_i f_i(x) \right| \leq \varepsilon \sum_{f_i\in \mathcal{F}} u_i f_i(x),$$
	where each element of $R$ is sampled i.i.d. with probability $p_j=\frac{s_j}{S}$ from $\mathcal{F}$, $w_i = \frac{Su_j}{s_j|R|}$ denotes the weight of a function $f_i\in R$ that corresponds to $f_j\in\mathcal{F}$, and where $\Delta$ is an upper bound on the VC-dimension of the range space $\mathfrak{R}_{\mathcal{F}^*}$ induced by $\mathcal{F}^*$ obtained by defining $\mathcal{F}^*$ to be the set of functions $f_j\in\mathcal{F}$, where each function is scaled by $\frac{Su_j}{s_j|R|}$.
\end{pro}

In the following section we will show how to compute upper bounds for the sensitivities and VC-dimension for our loss function.
In order to derive those bounds we first need to analyze the loss function $g$ of an individual point.
This will enable us to compute the desired upper bounds for both, the VC-dimension in Section \ref{sec:VC} and the sensitivities in Section \ref{sec:sen}.
Finally we will be able to apply Proposition \ref{thm:sensitivity} to obtain a coreset, and to prove Theorem \ref{thm:mainoverview}.

\section{PROPERTIES OF \texorpdfstring{$g$}{g}}\label{app:prop_g}
In this section we will determine useful properties of $g$.
Recall that for any function $f: \mathbb{R} \rightarrow \mathbb{R}$ the derivative $\frac{d}{dr} \int_{r}^\infty f(t) \,dt$ equals $\lim_{t \rightarrow \infty}f(t)-f(r)$ if the integral is finite.
First we note that for $r \in \mathbb{R}$ with $r\geq 0$
\begin{align*}
g'(r)&= \frac{\varphi_p(r)}{1-\Phi(r)} = \frac{\exp(- |r|^p/p)}{ \int_{r}^\infty  \exp(-|t|^p/p)\,dt}\\
&=\frac{1}{ \exp(|r|^p/p) \int_{r}^\infty  \exp(-|t|^p/p)\,dt}>0 .
\end{align*}
We set $h(r):=\frac{1}{g'(r)}=\exp(|r|^p/p)\int_{r}^\infty  \exp(-|t|^p/p)\,dt$.
Our aim is to characterize the tail behavior of the $p$-generalized normal distribution.
To this end we will first analyze $h$ using similar methods as \citet{Gordon41} who considered the case $p=2$, i.e., the standard normal distribution.

\begin{lem}\label{lem:h-prop}
The following holds for any $r>0$:
\begin{align}
h'(r)&=r^{p-1}h(r)-1; \label{eq:h'}\\
h''(r)&=(p-1)r^{p-2}h(r)+r^{p-1}h'(r);\label{eq:h''}\\
h''(r)&=\frac{r^{p}+p-1}{r}h'(r)+\frac{p-1}{r};\label{eq:h''2}\\
h'''(r)&=\left( 1+ \frac{p}{r^p+p-1}+\frac{p-2}{r^p} \right)r^{p-1}h''(r)
 -\frac{(p-1)pr^{p-2}}{r^p+p-1};\label{eq:h'''}\\
h(r)&>0; \label{eq:h}\\
h'(r)&<0;\label{eq:h'2}\\
h(r)&< \frac{1}{r^{p-1}}\label{eq:hub}\\
\intertext{Further if $r\geq 1$ then it holds that}
h''(r)&\geq 0;\label{eq:h''3}\\
h(r)&\geq \frac{r}{r^p+p-1}.\label{eq:hlb}
\end{align}
\end{lem}

\begin{proof}[Proof of Lemma \ref{lem:h-prop}]
Equations (\ref{eq:h'}) and (\ref{eq:h''}) can be derived by a direct calculation of the derivatives.
Note that (\ref{eq:h'}) is equivalent to
\begin{align}
h(r)=\frac{h'(r)+1}{r^{p-1}}\label{eq:h'r}
\end{align}
Equation (\ref{eq:h''2}) follows by substitution of (\ref{eq:h'r}) in (\ref{eq:h''}).
Equation (\ref{eq:h''}) is equivalent to
\begin{align}
h(r)=\frac{h''(r)}{(p-1)r^{p-2}}-\frac{r}{p-1}h'(r) \label{eq:h''r}
\end{align}
To get (\ref{eq:h'''}) we first note that by (\ref{eq:h'}) and then (\ref{eq:h''r}) it holds
\begin{align}
\frac{p-1}{r^2}h'(r)&=\frac{p-1}{r^2}(r^{p-1}h(r)-1)\notag\\
&=\frac{p-1}{r^2}r^{p-1}h(r)-\frac{p-1}{r^2}\notag\\
&=\frac{h''(r)}{r}-r^{p-2}h'(r)-\frac{p-1}{r^2}. \label{eq:h'help}
\end{align}
Further note that (\ref{eq:h''2}) is equivalent to
\begin{align}
h'(r)=\frac{rh''(r)}{r^p+p-1}-\frac{p-1}{r^p+p-1} \label{eq:h''2r}
\end{align}
Taking the derivative of (\ref{eq:h''2}) and using the equations (\ref{eq:h'help}) and (\ref{eq:h''2r}) we get
\begin{align*}
h'''(r)&=(p-1)r^{p-2}h'(r)-\frac{p-1}{r^2}h'(r)+r^{p-1}h''(r)\\
&\quad +\frac{p-1}{r}h''(r)-\frac{p-1}{r^2}\\
&\stackrel{(\ref{eq:h'help})}{=}p r^{p-2}h'(r)+r^{p-1}h''(r)+\frac{p-2}{r}h''(r)\\
&\stackrel{(\ref{eq:h''2r})}{=}pr^{p-2}\cdot\left( \frac{r}{r^p+p-1}h''(r)-\frac{p-1}{r^p+p-1}\right)\\
&\quad +r^{p-1}h''(r)+\frac{p-2}{r}h''(r)\\
&=\left( 1+ \frac{p}{r^p+p-1}+\frac{p-2}{r^p} \right)r^{p-1}h''(r)\\
&\quad -\frac{(p-1)pr^{p-2}}{r^p+p-1}.
\end{align*}
Equation (\ref{eq:h}) follows since all terms appearing in $h(r)$ are positive.

For (\ref{eq:h'2}) we note that
\begin{align}
r^{p-1} h(r)&= \exp(r^p/p)\int_{r}^\infty  r^{p-1}\exp(-|t|^p/p)\,dt\notag\\
&< \exp(r^p/p)\int_{r}^\infty \frac{p}{p} t^{p-1}\exp(-|t|^p/p)\,dt\notag\\
&=\exp(r^p/p) \cdot \exp(-r^p/p) = 1 \label{eq:h'help2}
\end{align}
and thus (\ref{eq:hub}) follows from dividing by $r^{p-1}$ and (\ref{eq:h'2}) also follows from (\ref{eq:h'help2}) using Equation (\ref{eq:h'}).% and (\ref{eq:h'help2}).

Next we prove (\ref{eq:h''3}):
For $r\geq 1$ it holds that $\left( 1+ \frac{p}{r^p+p-1}+\frac{p-2}{r^p} \right)r^{p-1}>0$.
Now let $r_0\geq 1$. 
Assume for the sake of contradiction that $h''(r_0)< 0$.
Then using (\ref{eq:h'''}) we also get $h'''(r_0)\leq \left( 1+ \frac{p}{r^p+p-1}+\frac{p-2}{r^p} \right)r^{p-1}h''(r_0)< 0$.
Thus we have $h''(r)<h''(r_0)$ for all $r>r_0$.
Consequently $h'$ is also strictly decreasing by a rate of at least $h''(r_0)$ starting at $r_0$.
This implies that there exists $r'>r_0$ with $h(r')<0$, which contradicts (\ref{eq:h}) and thus (\ref{eq:h''3}) follows.
Lastly (\ref{eq:hlb}) follows by substitution of (\ref{eq:h'}) in (\ref{eq:h''}) and using (\ref{eq:h''3}).
\end{proof}

\begin{lem}\label{lem:g-propcopy}[Copy of Lemma \ref{lem:g-prop}]
The function $g$ is convex and strictly increasing.
Further for any $r\geq 0$ we have
\begin{align*}
g'(r) &\geq r^{p-1},
\intertext{for any $r\geq 1$ we have}
g'(r) &\leq r^{p-1}+\frac{p-1}{r}.
\intertext{and there exists a constant $c_1>0$ such that}
g(r) &\geq c_1 e^{-|r|^p/p}
\end{align*}
for any $r<0$.
\end{lem}

\begin{proof}[Proof of Lemma \ref{lem:g-prop}/\ref{lem:g-propcopy}]
First note that $g=-\ln(\Phi_p(-r))$ is strictly increasing since $ \Phi(-r) \in (0,1)$ is strictly decreasing for increasing $r$ and $-\ln(t) $ is strictly increasing for decreasing $t$.
Next consider $r\geq 0$.
Then $g''(r)=\left( \frac{1}{h(r)} \right)'=-\frac{h'(r)}{h(r)^2}>0$ by (\ref{eq:h'2}) thus $g$ is convex on $ [0, \infty)$.
For $r<0$ we have derived in Section \ref{app:grad_hess} that
\begin{align*}
g''(r)=g'(r)(g'(r)-\sgn(r)|r|^{p-1}).
\end{align*}
For $r<0$ all terms are positive. Thus $g(r)$ is convex for all $r\in\mathbb{R}$.
The bounds for $g'$ follow immediately by the bounds for $h$ from Lemma \ref{lem:h-prop} (\ref{eq:hub}) and (\ref{eq:hlb}).

Now, for $ r<-1$ using the Taylor series of $-\ln(t)$ at $t=1$, the normalizing constant $C_p = \int_{-\infty}^\infty  \exp(-|t|^p/p)\,dt = \frac{2\Gamma(1/p)}{p^{1-1/p}}$ and Equation (\ref{eq:hlb}) we have
\begin{align*}
g(r)&=-\ln\left(1-C_p^{-1}\int_{-r}^\infty \exp(-|t|^p/p) \,dt\right)\\
&\geq C_p^{-1}\int_{-r}^\infty \exp(-|t|^p/p) \,dt\\
&\geq  C_p^{-1}\exp(-(-r)^p/p)\cdot \frac{r}{r^p+p-1}\geq \frac{\exp(-(-r)^p/p)}{p C_p}.
\end{align*}
For any $r \in [-1, 0]$ we have $ g(r) \geq g(-1)\geq g(-1)\exp(-(-r)^p/p)$.
Thus for $c_1=\min \{ g(-1), 1/(pC_p) \}$ we have  $ g(r) \geq c_1\exp(-(-r)^p/p)$.
\end{proof}

These properties can be used to prove the following lemma:

\begin{lem}\label{lem:gbound}
Set $G^+(r)=\frac{r^p}{p}$ if $r \geq 0$ and $G^+(r)=0$ if $r < 0$.
There exists $c_2>0$ depending only on $p$ such that for any $\varepsilon\in (0, e^{-1})$ and any $r \in \mathbb{R}$ it holds that
\begin{align}
G^+(r) \leq g(r) \leq (1+\varepsilon)G^+(r)+c_2 \ln\left(\frac{p}{\varepsilon}\right). \label{eq:gbounds}
\end{align}
\end{lem}

\begin{proof}[Proof of Lemma \ref{lem:gbound}]
For $r<0$ we have $g(r)>0=G^+(r)$.
For $r \geq 0$ by using Lemma \ref{lem:g-prop} we get
\begin{align*}
g(r)\geq  g(0)+\int_0^r g'(t)\,dt & \geq  g(0)+\int_0^r t^{p-1} \,dt\\
& = g(0)+G^+(r)\geq G^+(r).
\end{align*}
For the second inequality we split the domain of $g$ into three parts:
First since $g$ is monotonically increasing for any $r$, we have $g(r)\leq g(1)$ for $r \in ( -\infty, 1]$.
For $r\geq 1$, by using Lemma \ref{lem:g-prop}, it holds that
\begin{align}
g(r) & \leq g(1)+\int_1^{r}t^{p-1}+\frac{p-1}{t}\,dt\notag\\
 & = g(1)+G^+(r)-\frac{1}{p}+(p-1)\ln(r). \label{eq:gint}
\end{align}
Now consider $r \in (1, r_0]$ where $r_0=\frac{p^3}{\varepsilon^3}$.
Then we have
\begin{align*}
g(r) & \leq g(1)+G^+(r)+(p-1)\ln(r_0)=g(1)+G^+(r)+3(p-1)\ln\left(\frac{p}{\varepsilon}\right)
\end{align*}
Our last step is to show that for $r>r_0$ it holds that $(p-1)\ln(r) \leq \varepsilon G^+(r) $.
We assume without loss of generality that $ \varepsilon^{-1}\geq 2$.
Now the equation
\begin{align*}
\varepsilon G^+(r)=\varepsilon\frac{r^p }{p} \geq (p-1)\ln(r) 
\end{align*}
is equivalent to
\begin{align*}
\exp\left(\frac{\varepsilon r^p}{p^2-p} \right)\geq r.
\end{align*}
Note that $r^p \geq r$ holds since $r\geq r_0>1$ and thus we get for any $r=a r_0 $ with $a\geq 1$ that
\begin{align*}
\exp\left(\frac{\varepsilon r^p}{p^2-p} \right)&\geq \exp\left(\frac{\varepsilon r}{p^2} \right)
\geq \exp\left(\frac{\varepsilon a r_0}{p^2} \right)
\geq \exp\left(\frac{a p}{\varepsilon^2} \right)
\geq \exp\left(2a \cdot\frac{p}{\varepsilon} \right)\geq a r_0 = r.
\end{align*}
The last inequality follows from the fact that $e^{2 a z}\geq a z^3$ always holds in our case where $z\geq 2$ and $a\geq 1$.
Consequently it holds for any $r \in [r_0, \infty)$ that
\begin{align*}
g(r) & \leq g(1)+G^+(r)+(p-1)\ln(r)\leq g(1)+(1+\varepsilon)G^+(r).
\end{align*}
Combining all three inequalities we note that for any $r \in \mathbb{R}$ it holds that
\begin{align*}
g(r) & \leq g(1)+(1+\varepsilon)G^+(r)+(p-1)\ln\left(\frac{p^3}{\varepsilon^3}\right)\\
& = (1+\varepsilon)G^+(r)+\left(\frac{g(1)}{\ln(p/\varepsilon)}+3(p-1)\right)\ln\left(\frac{p}{\varepsilon}\right) \\
& \leq (1+\varepsilon)G^+(r)+c_2\ln\left(\frac{p}{\varepsilon}\right) 
\end{align*}
where $c_2:= (g(1)+3(p-1)) \geq (\frac{g(1)}{\ln(p/\varepsilon)}+3(p-1))$ holds, since $\varepsilon^{-1} \geq e$ and $p\geq 1$.
\end{proof}

\begin{lem}\label{lem:f-boundcopy}[Copy of Lemma \ref{lem:f-bound}]
Assume $X\in\mathbb{R}^{n\times d}$ is $\mu$-complex. Then we have for any $\beta \in \mathbb{R}^d$ that
\begin{align*}
f(X\beta)=\Omega\left(\frac{n}{\mu}\left(1+\ln(\mu)\right)\right).
\end{align*}
\end{lem}

\begin{proof}[Proof of Lemma \ref{lem:f-bound}/\ref{lem:f-boundcopy}]
Let $z=X\beta$.
For $r\leq 0$ we have $g(r) \geq c_1 e^{-r^p/p}$ by Lemma \ref{lem:g-prop}.
For $r\geq 0$ we have $g(r) = g(0)+\int_0^{r}g'(t)\,dt$.
Recall that
\[ g'(t)=\frac{1}{h(t)}\geq t^{p-1} \]
and thus 
\begin{align}
g(r) \geq g(0)+\int_0^{r}t^{p-1}\,dt=g(0)+\frac{r^{p}}{p}. \label{eq:gpos}
\end{align}
Set $z_-=\frac{1}{n}\sum_{z_i\leq 0}|z_i|^p $ and $z_+=\frac{1}{n}\sum_{z_i\geq 0}|z_i|^p \geq \frac{z_-}{\mu}$.
We set $z^- \in \mathbb{R}^n$ to be the vector with $z_i^-=z_i$ if $z_i<0$ and $z_i^-=0$ else.
Using convexity of $e^{-r}$ we can apply Jensens inequality to conclude that
\begin{align*}
f(X\beta) &= \sum_{i=1}^n g(z_i)\\
&=\sum_{i=1}^n \min\{g(z_i), g(0) \} ~+~ \sum_{z_i\geq 0}\int_0^{z_i}t^{p-1}\,dt\\
&\geq \sum_{i=1}^n c e^{-|z_i^-|^p/p} + \frac{1}{p}\sum_{z_i\geq 0}z_i^p\\
&\geq n c_1 e^{-(z_-)/p} + \frac{n z_+}{p}\\
&\geq n c_1 e^{-(z_-)/p} + \frac{n z_-}{\mu p}.
\end{align*}
Taking the derivative of $\ell(r)=n c_1 e^{-(r)/p} + \frac{n r}{\mu p}$, i.e. $\ell'(r)= \frac{n}{p}(-c_1 e^{-(r)/p} + \frac{1}{\mu})$ which is $0$ if $\frac{r}{p}= \ln(c_1\mu)$.
Thus it holds that
\begin{align*}
f(X\beta) \geq \ell(z_-) \geq \frac{n}{\mu}  (1+ \ln(c_1 \mu))
\end{align*}
which is exactly what we needed to show.
\end{proof}

\subsection{Bounding the VC-Dimension}\label{sec:VC}

In order to bound the VC-dimension of the range space induced by the weighted set of functions we need to reduce the number of distinct weights considered.
We first round all sensitivities to their closest power of $2$. The new total sensitivity $S'$ is at most twice the old sensitivity $S$.
Next we increase all sensitivities smaller than $\frac{S}{n}$ to $\frac{S}{n}$.
The new sensitivity is at most $S'+n\cdot S/n = 3S$.
The next step is to split the data into \emph{high sensitivity} points and \emph{low sensitivity} points.

\begin{lem}\label{lem:senssplitcopy}[Copy of Lemma \ref{lem:senssplit}]
Let $I_1$ be the index set of all data points with $s_i>s_0:=\frac{\mu S c \ln(p\varepsilon^{-1})}{ \varepsilon  n }$ for some constant $c \in \mathbb{R}_{>0}$.
Then for all $\beta \in \mathbb{R}^d$ it holds that
\begin{align*}
\sum_{i \in I_1}G^+(x_i\beta) &\leq \sum_{i \in I_1}g(x_i \beta) \leq (1+\varepsilon)\sum_{i \in I_1}G^+(x_i\beta) ~+~ \varepsilon \cdot \frac{n}{\mu}.
\end{align*}
\end{lem}

\begin{proof}[Proof of Lemma \ref{lem:senssplit}/\ref{lem:senssplitcopy}]
We set $c=c_2$ as in Lemma \ref{lem:gbound}.
Note that there are at most $\frac{S}{s_0}=\frac{\varepsilon  n}{c \ln(p\varepsilon^{-1})\mu}$ points in $I_1$.
Thus the lemma follows by applying Lemma \ref{lem:gbound} to each point in $I_1$.
\end{proof}

As a consequence we get the following corollary:

\begin{cor}\label{cor:senssplit}
Let $I_2=[n]\setminus I_1$. Further let $(X', w)\in \mathbb{R}^{n'\times d}\times \mathbb{R}^{n'} $ with rows $x_i'=x_{\pi(i)}$ for some mapping $\pi: [n'] \rightarrow [n]$.
We set $I_1'=\{i \in [n'] ~|~ \pi(i)\in I_1 \}$ and similarly $I_2'=\{i \in [n'] ~|~ \pi(i)\in I_2 \}$.
Further define $\tilde{f}_w(X'\beta)=\sum_{i \in I_2'}w_i g(x_i'\beta) ~+~ \sum_{i \in I_1'}w_i G^+(x_i'\beta) $ and by $\tilde{f}(X\beta)=\sum_{i \in I_2}g(x_i\beta) ~+~ \sum_{i \in I_1} G^+(x_i\beta) $.
Assume that for all $\beta\in \mathbb{R}^d$ it holds
\begin{align}\label{ass}
    | \tilde{f}_w(X'\beta)- \tilde{f}(X\beta) |\leq \varepsilon \tilde{f}(X\beta)
\end{align}
and $\sum_{i \in I_1'}w_i\leq \frac{2S}{s_0}$.
Further assume that $\varepsilon\leq \frac{1}{4}$.
Then $(X', w)$ is a $7\varepsilon$-coreset for the original $f$.
\end{cor}

\begin{proof}[Proof of Corollary \ref{cor:senssplit}]
Observe that by triangle inequality
\begin{align}\label{eqn:triangle}
    | f_w(X'\beta)- f(X\beta) |
    &\leq |f_w(X'\beta)-\tilde{f}_w(X'\beta)| + | \tilde{f}_w(X'\beta)- \tilde{f}(X\beta)| + |\tilde{f}(X\beta)-f(X\beta)|
\end{align}
By Lemma \ref{lem:senssplit} it holds that
\begin{align*}
    \tilde{f}(X\beta) \leq f(X\beta) &\leq \tilde{f}(X\beta)+\varepsilon \sum_{i \in I_1}G^+(x_i\beta)+ \varepsilon \cdot \frac{n}{\mu} \\
    &\leq \tilde{f}(X\beta)+2\varepsilon f(X\beta)
\end{align*}
We thus have that
\begin{align*}
    |\tilde{f}(X\beta)-f(X\beta)|\leq 2\varepsilon f(X\beta).
\end{align*}
Analogously to Lemma \ref{lem:senssplit}, using the bounded size of $\sum_{i \in I_1'}w_i$ and the assumption (\ref{ass}) one can show that
\begin{align*}
    \tilde{f}_w(X'\beta) \leq f_w(X'\beta) &\leq \tilde{f}_w(X'\beta)+\varepsilon \sum_{i \in I_1'}w_iG^+(x_i'\beta)+ \sum_{i \in I_1'}w_i \cdot c_2\ln\left(\frac{p}{\varepsilon}\right) \\
    &\leq \tilde{f}_w(X'\beta) + \varepsilon \tilde{f}_w(X'\beta)+\frac{2S}{s_0}\cdot c_2\ln\left(\frac{p}{\varepsilon}\right)\\
    & \stackrel{(\ref{ass})}{\leq} \tilde{f}_w(X'\beta) + \varepsilon (1+\varepsilon) \tilde{f}(X\beta)+2\varepsilon \frac{n}{\mu}\\
    &\leq \tilde{f}_w(X'\beta) + 2\varepsilon \tilde{f}(X\beta)+2\varepsilon f(X\beta)\\
    &\leq \tilde{f}_w(X'\beta) + 4\varepsilon f(X\beta)
\end{align*}    
and thus we have
\begin{align*}
    |f_w(X'\beta)- \tilde{f}_w(X'\beta)|\leq 4\varepsilon f(X\beta).
\end{align*}

Now combining everything into Equation (\ref{eqn:triangle}) yields
\begin{align*}
 | f_w(X'\beta)- f(X\beta) | &\leq |f_w(X'\beta)-\tilde{f}_w(X'\beta)| + | \tilde{f}_w(X'\beta)- \tilde{f}(X\beta)| + |\tilde{f}(X\beta)-f(X\beta)| \\
 &\leq 4\varepsilon f(X\beta) + \varepsilon \tilde{f}(X\beta) + 2\varepsilon f(X\beta)\\
 &\leq 7\varepsilon f(X\beta)
\end{align*}
and thus $(X', w)$ is a $7\varepsilon$-coreset.
\end{proof}

Before we continue showing that for the set of functions that we consider, the VC-dimension is not too large, we show that the assumption made in Corollary \ref{cor:senssplit} that $\sum_{i \in I_1'}w_i\leq \frac{2S}{s_0}$ is reasonable, i.e., that it holds with high probability in our context:

\begin{lem}\label{lem:I1ass}
Assume, as in the context of Proposition \ref{thm:sensitivity}, that for $R$ with $|R|=k$ where each element of $R$ is sampled i.i.d. with probability $p_j=\frac{s_j}{S}$ from $\mathcal{F}$ and $w_i = \frac{S}{s_j|R|}=\frac{1}{k p_j}$ denotes the weight of a function $f_i\in R$ that corresponds to $f_j\in\mathcal{F}$.
Then with probability at least $1-\frac{1}{k}$ it holds that $\sum_{i \in I_1'}w_i\leq  \frac{2S}{s_0}$.
\end{lem}

\begin{proof}[Proof of Lemma \ref{lem:I1ass}]
Let $x_{\pi(i)}$ be the $i$th element of $R$.
We set $Z_i=w_{\pi(i)}$ if $\pi(i)\in I_1$ and $Z_i=0$ otherwise. 
Then $\lambda=\mathbb{E}(Z_i)=\sum_{j \in I_1}p_j w_j=\sum_{j \in I_1}p_j\frac{1}{kp_j}=\frac{|I_1|}{k}$.
Recall from Lemma \ref{lem:senssplit} that $|I_1|\leq \frac{S}{s_0}$ and for $j\in I_1$ we have $s_j>s_0$.
For the variance it follows
\begin{align*}
\mathbb{E}[(Z_i-\lambda)^2] &= \mathbb{E}[Z_i^2] - \mathbb{E}[Z_i]^2 \leq \mathbb{E}[Z_i^2] = \sum_{j \in I_1}p_j w_j^2 =\sum_{j \in I_1}\frac{1}{k^2 p_j}\leq \sum_{j \in I_1}\frac{S}{k^2 s_0} = |I_1| \cdot \frac{S}{k^2 s_0} \leq \frac{S^2}{s_0^2k^2}\,.
\end{align*}
Thus by independence of the $Z_i$ the variance of $Z=\sum_{i=1}^k Z_i$ is bounded by $\frac{S^2}{s_0^2 k}$.
Now applying Chebyshev's inequality yields
\begin{align*}
P\left(Z\geq 2\cdot \frac{S}{s_0}\right) \leq P\left(Z-\mathbb{E}(Z)\geq \frac{S}{s_0}\right)\leq \frac{\mathrm{Var}(Z)}{S/s_0} \leq \frac{S^2/(s_0^2k)}{S^2/s_0^2}=\frac{1}{k}.
\end{align*}
\end{proof}

By the technical Corollary \ref{cor:senssplit} our goal of obtaining a coreset for $f$ reduces to obtaining a coreset for the substitute function $$\tilde{f}(X\beta)=\sum_{i\in [n] \setminus I_1}g(x_i\beta) + \sum_{i\in I_1}G^+(x_i\beta).$$
To this end we set $\mathcal{F}_1=\{ w_i G^+_i ~|~ i \in I_1 \}$ where $G^+_i(\beta)=G^+(x_i\beta)$ and $\mathcal{F}_2=\{ w_i g_i ~|~ i \in I_2=[n]\setminus I_1  \}$ where $g_i(\beta)=g(x_i\beta)$.
Further we set $\mathcal{F}=\mathcal{F}_1 \cup \mathcal{F}_2$ and show that the VC-dimension of $\mathcal{F}$ can be bounded as desired:

\begin{lem}\label{lem:VCdimcopy}[Copy of Lemma \ref{lem:VCdim}]
For the VC-dimension $\Delta$ of $ \mathfrak{R}_{\mathcal{F}}$ we have
\begin{align*}
    \Delta &\leq (d+1) \left(\log_2\left( \mu  c \varepsilon^{-2}\right)+2 \right)={O}(d\log({\mu}/{\varepsilon})).
\end{align*}
\end{lem}

\begin{proof}[Proof of Lemma \ref{lem:VCdim}/\ref{lem:VCdimcopy}]
First note that for any $G \subseteq \mathcal{F}_1 $, $\beta \in \mathbb{R}^d$ and $ r \in \mathbb{R}$ it holds that $\rng{\mathcal{F}_1}(\beta ,r) \cap G=\rng{G}(\beta ,r)$.
We show that the VC-dimension of $\mathfrak{R}_{\mathcal{F}_1}$ is at most $d+1$. Indeed, it holds that $\rng{\mathcal{F}_1}(\beta ,r) =\mathcal{F}_1$ if $r\leq 0$ since all weights are positive and $G^+$ is also positive. Otherwise we have that
\begin{align*}
\rng{\mathcal{F}_1}(\beta ,r)&=\{ w_iG^+_i\in \mathcal{F}_1 ~|~ w_i G^+_i(\beta)\geq r \}\\
&= \{ w_iG^+_i\in \mathcal{F}_1 ~|~ w_i( x_i\beta)^p/p\geq r \wedge x_i\beta>0 \}\\
%&= \{ w_iG^+_i\in \mathcal{F}_1 ~|~ \frac{w_i^{1/p}x_i}{p^{1/p}}\beta \geq r^{1/p} \}\\
&= \left\{ w_iG^+_i\in \mathcal{F}_1 ~|~ x_i\beta \geq \left(\frac{p r}{w_i}\right)^{1/p}  \right\}.
\end{align*}
We conclude that for $G \subseteq \mathcal{F}_1$ it holds that
\begin{align*}
    |\{G\cap R\mid R\in \ranges(\mathcal{F}_1) \}| &=|\{ \rng{G}(\beta ,r) ~|~ \beta \in \mathbb{R}^d, r \in \mathbb{R}_{> 0}\} \cup \{ \rng{G}(\beta ,r) ~|~ \beta \in \mathbb{R}^d, r \in \mathbb{R}_{\leq 0}\}|\\
    &=\left|\left\{\left\{ w_iG^+_i\in G \mid x_i \beta\geq\left({p r}/{w_i}\right)^{1/p} \right\} \mid \beta \in \mathbb{R}^d, r \in \mathbb{R}_{\geq 0}\right\}\cup \{ G \}\right|\\
    &\leq |\{\{ w_iG^+_i\in G \mid x_i \beta - s \geq 0 \} ~|~ \beta \in \mathbb{R}^d, s \in \mathbb{R}\}|
\end{align*}
which corresponds to a set of affine hyperplane classifiers $\beta \mapsto \mathbf{1}_{x_i \beta - s \geq 0}$, which have VC-dimension $d+1$ \citep{KearnsV94}. Thus, the induced range space $\mathfrak{R}_{\mathcal{F}_1}$ has VC-dimension at most $d+1$.\\
Next consider $\mathcal{F}_2$.
Note that $g$ is a strictly monotonic and thus also invertible function.
First fix a weight $v\in \mathbb{R}_{>0 }$ and let $ \mathcal{F}_v=\{ w_i g_i ~|~ w_i=v \}$.
We have
\begin{align*}
\rng{\mathcal{F}_v}(\beta ,r)&=\{ w_ig_i\in \mathcal{F}_v ~|~ w_ig_i(\beta)\geq r \}\\
&= \left\{ w_ig_i\in \mathcal{F}_v ~|~ x_i\beta \geq g^{-1}\left(\frac{r}{v}\right) \right\}
\end{align*}
which corresponds to a set of points shattered by the affine hyperplane classifier $\beta \mapsto \mathbf{1}_{x_i\beta-g^{-1}\left(\frac{r}{v}\right)\geq 0}$ and thus the VC-dimension of the induced range space $\mathfrak{R}_{\mathcal{F}_v}$ is at most $d+1$.
Let $W$ be the set of all weights for functions in $\mathcal{F}_2$.
Since all weights are powers of $2$, and we have $\frac{S}{n} \leq v \leq \frac{\mu S c \ln(p\varepsilon^{-1})}{ \varepsilon  n }$ it holds that$|W|\leq \log_2(\frac{\mu  c \ln(p \varepsilon^{-1})}{ \varepsilon  })\leq \left(\log_2\left( \mu c \varepsilon^{-2}\right)+2 \right)$.
Now we claim that the VC-dimension of $\mathfrak{R}_{\mathcal{F}}$ is at most $ (|W|+1)(d+1)$ as $\mathcal{F}=\mathcal{F}_1\cup \bigcup_{v \in W}\mathcal{F}_v $.
Assume for the sake of contradiction that there exists $G \subset \mathcal{F}$ such that $|G|>(|W|+1)(d+1)$ and $G$ is shattered by the ranges of $\mathcal{F}$.
Then by the pigeonhole principle $G'=G \cap \mathcal{F}' > d+1$ for some $\mathcal{F}' \in \{ {F}_1 \} \cup \bigcup_{v \in W}\{\mathcal{F}_v\}$. But due to the pairwise disjointness of all of $\mathcal{F}_1$ and $\mathcal{F}_v$, $G'$ must be shattered by the ranges of $\mathcal{F}'$, which contradicts that their VC-dimension is bounded by $d+1$, cf. Lemma 11 in \citep{MunteanuSSW18}.
\end{proof}

\subsection{Bounding the Sensitivities}\label{sec:sen}

We define the $\ell_p$-leverage scores of $X$ by $u_j=\sup_{\beta \in \mathbb{R}^d\setminus \{0\}}\frac{\vert x_j\beta \vert^p}{\sum_{i=1}^n \vert x_i \beta \vert^p}$, cf. \citep{DasguptaDHKM09}.
We note that the supremum is attained by some $\beta \in \mathbb{R}^d$ since
\begin{align*} 
\sup_{\beta \in \mathbb{R}^d\setminus \{0\}}\frac{\vert x_j\beta \vert^p}{\sum_{i=1}^n \vert x_i \beta \vert^p}
&=\sup_{\beta \in \mathbb{R}^d\setminus \{0\}}\frac{\|\beta\|_2^{p} \cdot \vert x_j\beta/\|\beta\|_2 \vert^p }{\|\beta\|_2^{p} \cdot \sum_{i=1}^n \vert x_i \beta/\|\beta\|_2 \vert^p }
=\sup_{\beta \in \mathbb{R}^d, \Vert \beta \Vert_2=1}\frac{\vert x_j\beta \vert^p}{\sum_{i=1}^n \vert x_i \beta \vert^p}
\end{align*}
and $\{ \beta \in \mathbb{R}^d \mid \Vert \beta \Vert_2=1\} $ is a compact set.
We also note that $u_p \leq 1$ always holds.

\begin{lem}\label{lem:sensboundcopy}[Copy of Lemma \ref{lem:sensbound}]
There is a constant $c_s$ such that the sensitivity $\zeta_i$ of $x_i, i\in[n]$ for $\tilde{f}$ is bounded by
\begin{align*}
\zeta_i \leq c_s \mu\left(\frac{1}{n}+u_i\right)
\end{align*}
\end{lem}

\begin{proof}[Proof of Lemma \ref{lem:sensbound}/\ref{lem:sensboundcopy}]
First note that by Lemma \ref{lem:f-bound} it holds that $ f(X\beta) \geq \frac{n }{\mu}$ and thus by Lemma \ref{lem:senssplit} $ \tilde{f}(X\beta) \geq  f(X\beta)-\varepsilon \sum_{i \in I_1}G^+(x_i\beta) - \varepsilon \cdot \frac{n}{\mu} \geq \frac{ f(X\beta)}{2} \geq \frac{n }{2\mu}$ holds for small enough $\varepsilon\leq 1/4$.
Thus for any $\beta$ with $x_i\beta\leq 1$ we have $\frac{G^+(x_i\beta)}{\tilde{f}(X\beta)} \leq \frac{g(x_i\beta)}{\tilde{f}(X\beta)}\leq \frac{g(1)}{n  /2\mu}= 2g(1)\frac{\mu}{n}$.

Further for $\beta$ with $x_i \beta > 1$ it holds that $g(x_i \beta )\leq c_3 (x_i\beta)^p$ for some constant $c_3 \leq 2g(1)+1$ since by Lemma \ref{lem:g-prop} and using $ \frac{p-1}{t}\leq p-1$ for $t\geq 1$ it holds that
\begin{align*}
   g( x_i \beta) &= g(1)+\int_{1}^{x_i \beta }g'(t) \,dt\\
   &\leq  g(1)+\int_{1}^{ x_i \beta } t^{p-1}+\frac{p-1}{t}\,dt\\
   &\leq  g(1)+\int_{1}^{ x_i \beta } t^{p-1}+ p-1\,dt\\
   &\leq  g(1)+\int_{1}^{ x_i \beta } t^{p-1}+(p-1)t^{p-1}\,dt\\
   &\leq  g(1)+\int_{1}^{ x_i \beta } p t^{p-1}\,dt\\
   &=g(1)+(x_i\beta)^p-1\leq (g(1)+1)(x_i\beta)^p. 
\end{align*}
Also note that by definition of $\mu$ it holds that
\begin{align*}
    &\frac{1}{\sum_{x_j\beta>0} |x_j\beta|^p(1+\mu)} \leq \frac{1}{\sum_{x_j\beta>0} |x_j\beta|^p+\sum_{x_j\beta<0} |x_j\beta|^p} = \frac{1}{\sum_{j=1}^n |x_j\beta|^p}
\end{align*}
and thus $$\frac{1}{\sum_{x_j\beta>0}|x_j\beta|^p} \leq  \frac{1+\mu}{\sum_{j=1}^n |x_j\beta|^p}.$$
Now setting $c_3=2g(1)+1$ and using $\tilde{f}(X\beta)\geq \sum_{x_j\beta>0} \frac{|x_j\beta|^p}{p}=\sum_{j=1}^n G^+(x_j\beta)$ we get
\begin{align*}
\frac{G^+(x_i\beta)}{\tilde{f}(X\beta)} \leq \frac{g(x_i\beta)}{\tilde{f}(X\beta)}\leq \frac{c_3}{1/p}\cdot \frac{|x_i\beta|^p}{\sum_{x_j\beta>0} |x_j\beta|^p} \leq p c_3(1+\mu) u_i \leq 2pc_3\mu u_i := c_s \mu u_i.
\end{align*}
Combining both bounds gives us the bound for $\zeta_i$.
\end{proof}

\subsection{Well Conditioned Bases and Approximate Leverage Scores}

In order to approximate the leverage scores we will need well conditioned bases:

An $(\alpha,\beta,p)$-well-conditioned basis $V$ is a basis that preserves the norm of each vector well, in the sense that its entry-wise $p$ norm $\|V\|_p\leq\alpha$ and for all $z\in \mathbb{R}^d \colon \|z\|_q\leq \beta \|Vz\|_p$, where $q$ denotes the dual norm to $p$, i.e., $\frac{1}{p}+\frac{1}{q}=1$, see Definition \ref{def:good_basis}.
We will first state the properties of the \emph{$(\alpha,\beta,p)$-well-conditioned basis} and then we describe how to compute the basis.

\begin{lem}\label{lem:levscoreboundcopy}[Copy of Lemma \ref{lem:levscorebound}]
Let $V$ be an \emph{$(\alpha,\beta,p)$-well-conditioned basis} for the column space of $X$.
Then it holds for all $i\in[n]$ that $u_i \leq \beta^p \Vert v_i\Vert_p^p$. % and thus $$\sum_{i=1}^{n} u_i\leq \beta^p \|q_i\|_p^p=\beta^p \|Q\|_p^p = (\alpha\beta)^p.$$
As a direct consequence we have $\sum_{i=1}^{n} u_i\leq \beta^p \|V\|_p^p \leq (\alpha\beta)^p = d^{O(p)}$.
\end{lem}

\begin{proof}[Proof of Lemma \ref{lem:levscorebound}/\ref{lem:levscoreboundcopy}]
We have by a change of basis
\begin{align*}
u_i&=\sup_{z \in \mathbb{R}^d\setminus\{0\}} \frac{|(Xz)_i|^p}{\Vert Xz\Vert_p^p} =\sup_{z \in \mathbb{R}^d\setminus\{0\}} \frac{|(Vz)_i|^p}{\Vert Vz\Vert_p^p} .
\end{align*}
Now assume that $z$ attains the value $\sup_{z \in \mathbb{R}^d\setminus\{0\}}\frac{|(Vz)_i|^p}{\Vert Vz\Vert_p^p}$.
Then we get by using Hölder's inequality and the properties of $V$ that
\begin{align*}
u_i&=\frac{|(Vz)_i|^p}{\Vert Vz\Vert_p^p}
\leq \frac{\beta^p|(Vz)_i|^p}{\Vert z\Vert_q^p}
\leq \frac{\beta^p\Vert v_i\Vert_p^p \Vert z\Vert_q^p}{\Vert z\Vert_q^p}=\beta^p\Vert v_i\Vert_p^p.
\end{align*}
\end{proof}

An \emph{$(\alpha,\beta,p)$-well-conditioned basis} can be computed using sketching techniques.

\begin{lem}{\citep{WoodruffZ13,ClarksonW17}}[Copy of Lemma \ref{lem:wcb}]\label{lem:wcb_copy}
There exists a random embedding matrix $\Pi\in\mathbb{R}^{n'\times n}$ and $\gamma = {O}(d \log(d))$ such that
\begin{align*}%\label{eq:ranmatrix}
\forall \beta \in \mathbb{R}^d: ~ \frac{1}{\gamma^{1/p}} \Vert X \beta \Vert_p
\leq  \Vert \Pi X \beta \Vert_q \leq \gamma^{1/p}\Vert X \beta \Vert_p
\end{align*}
holds with constant probability, where 
    \[
	(q,n') = \begin{cases} (2,O(d^2)) &\mbox{if } p\in[1,2] \\ 
	(\infty,O(n^{1-\frac{2}{p}}\log n(d \log d)^{1+\frac{2}{p}} + d^{5+4p})) & \mbox{if } p\in(2,\infty). \end{cases}
	\]
%$q=2, n'=O(d^2)$ if $p\in[1,2]$ and $q=\infty, n'=O(n^{1-\frac{2}{p}}\log n(d \log d)^{1+\frac{2}{p}} + d^{5+4p})$ if $p\in(2,\infty)$.
For $p=2$ we have $\gamma = 2$. Further $\Pi X$ can be computed in ${O}(\nnz(X))$ time.
\end{lem}

The sketching matrix $\Pi$ can be constructed as follows:
First let $D\in \mathbb{R}^{n \times n}$ be the diagonal matrix with $D_{ii}=1$ or $D_{ii}=-1$ each with probability $1/2$.
Further let $h:[n] \rightarrow [n']$ be a random map where $h$ hashes each entry of $[n]$ to one of $n'$ buckets uniformly at random.
Set $\Psi\in \mathbb{R}^{n' \times n}$ to be the matrix where $\Psi_{h(i)i}=1$ and $\Psi_{ji}=0$ if $j \neq h(i)$.
For $p=2$ it suffices to take $\Pi=\Psi D$ \citep{ClarksonW17}.
%TODO: phi änder, D->E
Otherwise if $p\neq 2$ let $E$ be a diagonal matrix with $E_{ii}={1}/{\lambda_i^{1/p}}$ where $ \lambda_i \sim \exp(1)$ is drawn from a standard exponential distribution and set $\Pi=\Psi D E$ \citep{WoodruffZ13}.

\begin{lem}\label{lem:wellconditioncopy}[Copy of Lemma \ref{lem:wellcondition}]
If $\Pi$ satisfies Lemma \ref{lem:wcb} and $\Pi X=QR$ is the QR-decomposition of $\Pi X$ then $V=XR^{-1}$ is an \emph{$(\alpha,\beta,p)$-well-conditioned basis} for the columnspace of $X$, where for $\gamma = {O}(d \log(d))$ we have
\begin{align*}
(\alpha, \beta)= \begin{cases}
(\sqrt{2 d},\sqrt{2}), & \text{for }p=2\\
(d\gamma^{1/p},\gamma^{1/p}), & \text{for }p \in [1, 2)\\
(d\gamma^{1/p},d \gamma^{1/p}), & \text{for }p \in (2, \infty).
\end{cases}
\end{align*}

\end{lem}

\begin{proof}[Proof of Lemma \ref{lem:wellcondition}/\ref{lem:wellconditioncopy}]
	We are going to use the fact that $Q=\Pi X\R$ is an orthonormal basis. Let $e_i$ for $i\in[d]$ denote the $i$th standard basis vector. We define $(\R)^{(i)}$ to be the $i$th column of $\R$. We have
	\begin{align}
	\label{eqn:begin}
	\normp{V} &= \normp{X\R} = \normp{X\sum_{i=1}^d(\R)^{(i)}e_i^T} = \normp{\sum_{i=1}^d X(\R)^{(i)}e_i^T}\notag\\
	&\leq \sum_{i=1}^d \normp{X(\R)^{(i)}e_i^T} = \sum_{i=1}^d \normp{X(\R)^{(i)}}
	\end{align}
	Now suppose $p \in (2,\infty)$.
	\begin{align*}
	(\ref{eqn:begin}) &\leq \gamma^{1/p} \sum_{i=1}^d \norminf{\Pi X(\R)^{(i)}} \leq \gamma^{1/p}\sqrt{d} \left( \sum_{i=1}^d \norminf{\Pi X(\R)^{(i)}}^2 \right)^\frac{1}{2} \\
	&\leq \gamma^{1/p} \sqrt{d} \left( \sum_{i=1}^d \norm{\Pi X(\R)^{(i)}}^2 \right)^\frac{1}{2} \leq \gamma^{1/p} \sqrt{d} \left( \sum_{i=1}^d \underbrace{\norm{Q^{(i)}}^2}_{=1} \right)^\frac{1}{2} = \gamma^{1/p} d %\qedhere
	\end{align*}
	For arbitrary $z\in \RL^d$ it holds that
	\begin{align*}
	\normq{z} \leq \sqrt{d}\norm{z} = \sqrt{d}\norm{Q z} = \sqrt{d}\norm{\Pi X \R z} \leq d\norminf{\Pi X \R z} \leq d \gamma^{1/p} \normp{V z}.
	\end{align*}
	Consequently $V$ is $(\gamma^{1/p} d,\gamma^{1/p} d, p)$-well-conditioned.
	
	Next suppose $p \in [1,2)$.	Again we bound
	\begin{align*}
	(\ref{eqn:begin}) &\leq \gamma^{1/p} \sum_{i=1}^d \norm{\Pi X(\R)^{(i)}} \leq \gamma^{1/p} \sqrt{d} \left( \sum_{i=1}^d \norm{\Pi X(\R)^{(i)}}^2 \right)^\frac{1}{2} \\ 
	&\leq \gamma^{1/p} \sqrt{d} \left( \sum_{i=1}^d \underbrace{\norm{Q^{(i)}}^2}_{=1} \right)^\frac{1}{2} = \gamma^{1/p} d %\qedhere
	\end{align*}
	Also, since $p\leq 2$, the dual norm satisfies $q\geq 2$. Fix an arbitrary $z\in \RL^d$. It follows that
	\begin{align*}
	\normq{z} \leq \norm{z} = \norm{Q z} = \norm{\Pi X \R z} \leq \gamma^{1/p} \normp{V z}.
	\end{align*}
	It follows that $V$ is even $(\gamma^{1/p} d,\gamma^{1/p}, p)$-well-conditioned in this case.
	
	Finally suppose $p =2$, where the entry-wise matrix norm is the Frobenius norm $\|\cdot\|_F$. We have
	\begin{align*}
	\Vert V \Vert_F^2=\sum_{i=1}^d \norm{ V^{(i)}}^2=\sum_{i=1}^d \norm{ (X R^{-1})^{(i)}}^2&\leq \sum_{i=1}^d 2\norm{ (\Pi X R^{-1})^{(i)}}^2\\
	&= 2\sum_{i=1}^d \norm{ Q^{(i)}}^2=2d.
	\end{align*}
	Thus we have $\Vert V \Vert_F=\sqrt{2d} $.
	Since $p=q=2$ we have for any $\beta\in \RL^d$ that
	\begin{align*}
	\normq{z} = \norm{z} = \norm{Q z} = \norm{\Pi X \R z} \leq \sqrt{2} \normp{V z}.
	\end{align*}
	Consequently $V$ is a $(\sqrt{2d},\sqrt{2}, p)$-well-conditioned in this case.
\end{proof}

\subsection{Proof of Main Results}

The main results, Theorem \ref{thm:mainoverview} and the improvement for $p=2$, Corollary \ref{cor:p2}, are completely contained in the main body of the paper. It remains to prove Corollary \ref{cor:minimization} that handles the approximation factor of the solution obtained from the coreset with respect to the original loss function on the full data.

\begin{cor}\label{cor:minimizationcopy}[Copy of Corollary \ref{cor:minimization}]
Let $(X',w)$ be a weighted $\varepsilon$-coreset for $f$. Let $\tilde{\beta}\in \operatorname{argmin}_{\beta\in\mathbb{R}^d} f_w(X'\beta)$. Then it holds that $f(X\tilde\beta )\leq (1+3\varepsilon)\min_{\beta\in\mathbb{R}^d} f(X\beta)$.
\end{cor}

\begin{proof}[Proof of Corollary \ref{cor:minimization}/\ref{cor:minimizationcopy}]
Let ${\beta^*}\in \operatorname{argmin}_{\beta\in\mathbb{R}^d} f(X\beta)$. Since $(X',w)$ is a coreset we have by Definition \ref{def:coreset} and using the optimality of $\tilde\beta$ for the coreset that
\begin{align*}
    f(X\tilde\beta) &\leq f_w(X'\tilde\beta)/(1-\varepsilon) \leq f_w(X'\beta^*)/(1-\varepsilon)\\
     & \leq f(X\beta^*)(1+\varepsilon)/(1-\varepsilon) \leq f(X\beta^*)(1+3\varepsilon). 
\end{align*}
\end{proof}

\clearpage
\section{EXPERIMENTS, PLOTS AND PSEUDO CODE}\label{app:experiments}
Here we provide further material deferred from the experimental section in the main body.

We briefly introduce the data sets that we used. Table \ref{tab:datasetvalues} provides a summary.
The {Webspam}\footnote{https://www.csie.ntu.edu.tw/$\sim$cjlin/libsvmtools/datasets/binary.html\#webspam} ~data consists of $350\,000$ unigrams with $128$ features from web pages which have to be classified as spam or normal pages.
The {Covertype}\footnote{https://archive.ics.uci.edu/ml/datasets/Covertype} ~data consists of $581\,012$ cartographic observations of different forests with $54$ features. The task is to predict the type of trees at each location.
The {Kddcup}\footnote{https://kdd.ics.uci.edu/databases/kddcup99/kddcup99.html} ~data consists of $494\,021$ network connections with $33$ features and the task is to detect network intrusions.
The {Example-2D} data consists of $175$ synthetic data points ($80$ per class plus $15$ outliers in one class) in $2$ dimensions to visualize the results of $p$-probit regressions for different values of $p$.

\begin{table*}[ht!]
    \label{tab:datasetvalues}
    \caption{Summary of the used data sets and their dimensions: The given values of $d$ do not include the intercept. The data sets are downloaded or generated automatically by our open Python implementation that is available at \url{https://github.com/cxan96/efficient-probit-regression}. 
	}
	\centering
	\begin{tabular}{ | l| r| r|}
		\hline
		{\bf data set} & {$\mathbf{n}$} & {$\mathbf{d}$} \\ \hline
		Webspam & $350\,000$ & $128$ \\ \hline 
		Covertype & $581\,012$ & $54$ \\ \hline 
		Kddcup & $494\,021$ & $33$ \\ \hline 
		Example-2D & $175$ & $2$ \\\hline 
	\end{tabular}
\end{table*}

\begin{algorithm}[ht!]
  \caption{Coreset algorithm for $p$-generalized probit regression.}\label{alg:main}%\label{alg:code}
  \begin{algorithmic}[1]
    \Statex \textbf{Input:} data $X \in \mathbb{R}^{n \times d}$, number of rows $k$.;
    \Statex \textbf{Output:} coreset $C=(X', w) \in \mathbb{R}^{k \times d}$ with $k$ rows.; 
    \State  Initialize sketch $X''=\mathbf{0} \in \mathbb{R}^{n' \times d}$, (where $n'={O}(d^2)$ for $p\leq 2$ or $n'={O}(n^{1-\frac{2}{p}}\log n \cdot \poly(d))$ for $p>2$);
 	\For{$i=1\ldots n$}
 		\State Draw a random number $B_i \in [n']$; \Comment{hash to bucket $B_i$}
 		\State Draw a random number $\sigma_i \in \{-1, 1\}$; \Comment{random sign}
 		\If{$p\neq 2$} 
 		    \State Draw a random number $\lambda_i \sim \exp(1)$; \Comment{$\ell_p$ embedding}
 			\State $\sigma_i={\sigma_i}/{\lambda_i^{1/p}}$.
 		\EndIf
 		\State $ X_{B_i}''=X_{B_i}''+\sigma_i \cdot x_i$. \Comment{sketch}
 	\EndFor
	\State Compute the QR-decomposition of $X''=QR$.; \Comment{well-conditioned basis}
	\State  Initialize coreset $X'=\mathbf{0} \in \mathbb{R}^{k \times d}$  \Comment{coreset points} 
	\State  Initialize weights $w=0 \in \mathbb{R}^{k}$;  \Comment{coreset weights}
	\State  Initialize $k$ independent weighted reservoir samplers $S_j$, sampling row $X'_j$, for each $j\in [k]$;
	\State  Initialize $G=I\in \mathbb{R}^{d\times d}$;  \Comment{Identity matrix}
 	\If{$p=2$ and $\ln n < d$}
 	    \State Draw $G\in\mathbb{R}^{d\times \ln n}$ with $G_{ij}\sim N(0,\frac{1}{\ln n})$; \Comment{JL-embedding}
    \EndIf
	\For{$i=1\ldots n$}
	    \State Compute $q_i=\|x_i (R^{-1} G)\|_p^p$; \Comment{$\ell_p$-leverage score approximation}
		\For{$j=1\ldots k$}
		    \State Feed $s_i=q_i+1/n$ to $S_j$; \Comment{unnormalized sampling probabilities}
		    \If{$S_j$ samples $x_i$}
		        \State $w_j={1}/{(k\cdot s_i)}$; \Comment{unnormalized weights}
		        \State $X'_j=x_i$; \Comment{save row identity in the coreset}
		    \EndIf
	    \EndFor	    
	\EndFor
	\State $w = w\cdot \sum_{i=1}^n s_i$; \Comment{normalize weights}
	\State \textbf{return} $C = (X', w)$;
  \end{algorithmic}
\end{algorithm}

\begin{algorithm}[ht!]
  \caption{Online coreset algorithm for the standard probit model $(p=2)$.}\label{alg:online}%\label{alg:code}
  \begin{algorithmic}[1]
    \Statex \textbf{Input:} data $X \in \mathbb{R}^{n \times d}$, number of rows $k$.;
    \Statex \textbf{Output:} coreset $C=(X', w) \in \mathbb{R}^{k \times d}$ with $k$ rows.; 
    \State  Initialize $M = M_{inv} = Q = 0 \in \mathbb{R}^{d \times d}$;
    \State  Initialize $k$ independent weighted reservoir samplers $S_j$, sampling row $X'_j$, for each $j\in [k]$;
 	\For{$i=1\ldots n$}
 	    \State $x_i \in \mathbb{R}^d := $ $i$th column of $X^T$;
 	    \State $M = M + x_i x_i^T$;
 		\If{$\lVert Q x_i \rVert_2 = \lVert x_i \rVert_2$} 
 		    \Comment{is $x_i$ in column-space of Q?}
 		    \State $M_{inv} = M_{inv} - \frac{M_{inv}x_ix_i^TM_{inv}}{1 + x_i^T M_{inv} x_i}$;
 		    \Comment{adapted Sherman-Morrison formula}
        \Else
            \State $M_{inv} = M^\dagger$; \Comment{Moore-Penrose pseudoinverse}
            \State $QR = M$; \Comment{$QR$ decomposition of $M$}
 		\EndIf
 		\State $\ell_i = \min\{x_i^T M_{inv} x_i, 1\}$; \Comment{approximate $\ell_2$-leverage scores}
 		\For{$j=1\ldots k$}
 		    \State Feed $s_i = \ell_i + \frac{1}{n}$ to $S_j$; \Comment{unnormalized sampling probabilities}
 		    \If{$S_j$ samples $x_i$}
 		        \State $w_j = {1}/{(k \cdot s_i)}$; \Comment{unnormalized weights}
 		        \State $X_j' =  x_i$; \Comment{save row in coreset}
 		    \EndIf
 		\EndFor
 	\EndFor
 	\State $w = w\cdot \sum_{i=1}^n s_i$; \Comment{normalize weights}
	\State \textbf{return} $C = (X', w)$;
  \end{algorithmic}
\end{algorithm}

\begin{figure}[ht!]
    \centering
    \includegraphics[width=.49\linewidth]{plots/2d-example.pdf}
    \includegraphics[width=.49\linewidth]{plots/2d-example-residual-plot.pdf}\\
    \includegraphics[width=0.49\linewidth]{plots/kddcup-compare-coefficients.pdf}
    \includegraphics[width=0.49\linewidth]{plots/webspam-compare-coefficients.pdf}\\
    \includegraphics[width=0.47\linewidth]{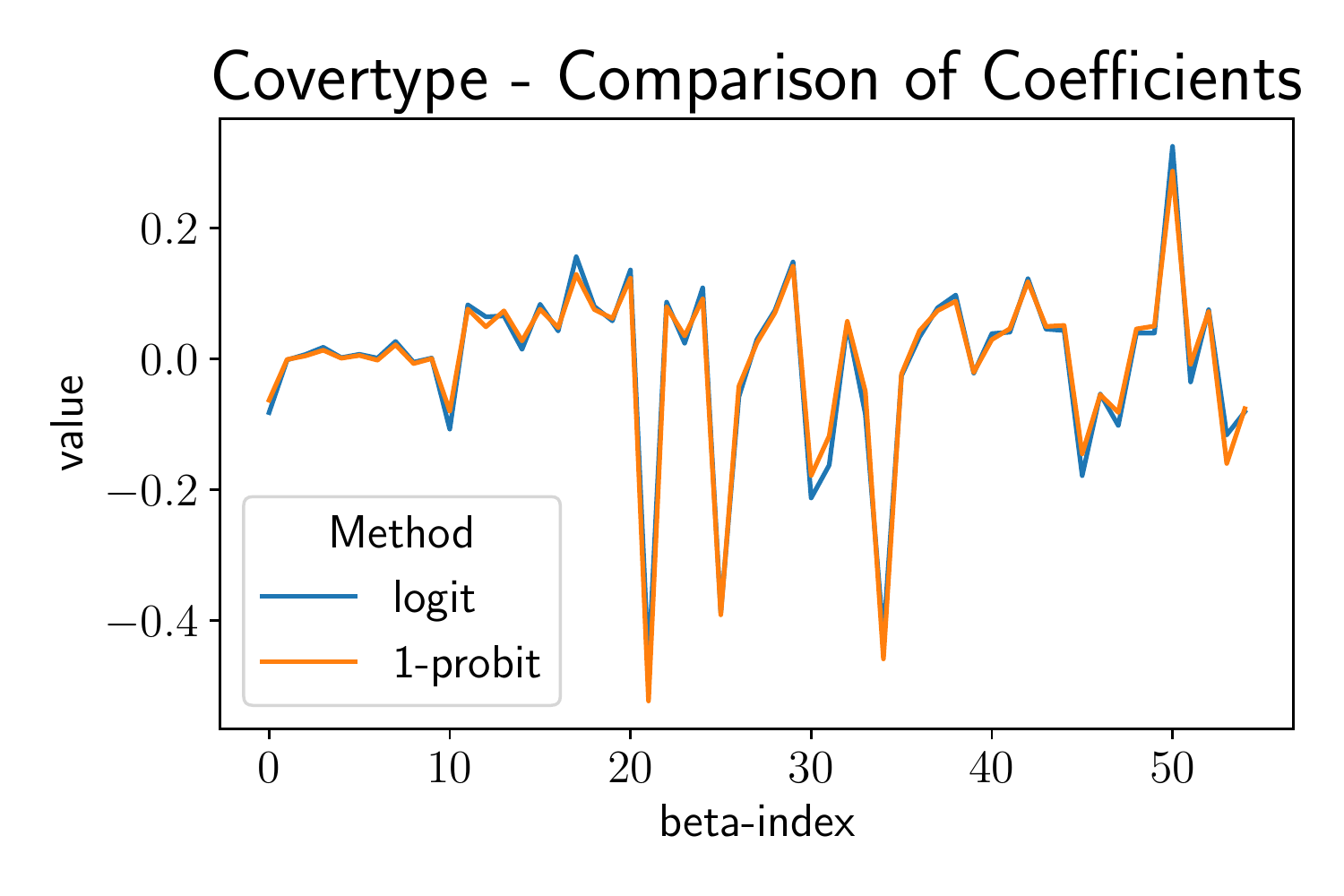}
    \includegraphics[width=0.47\linewidth]{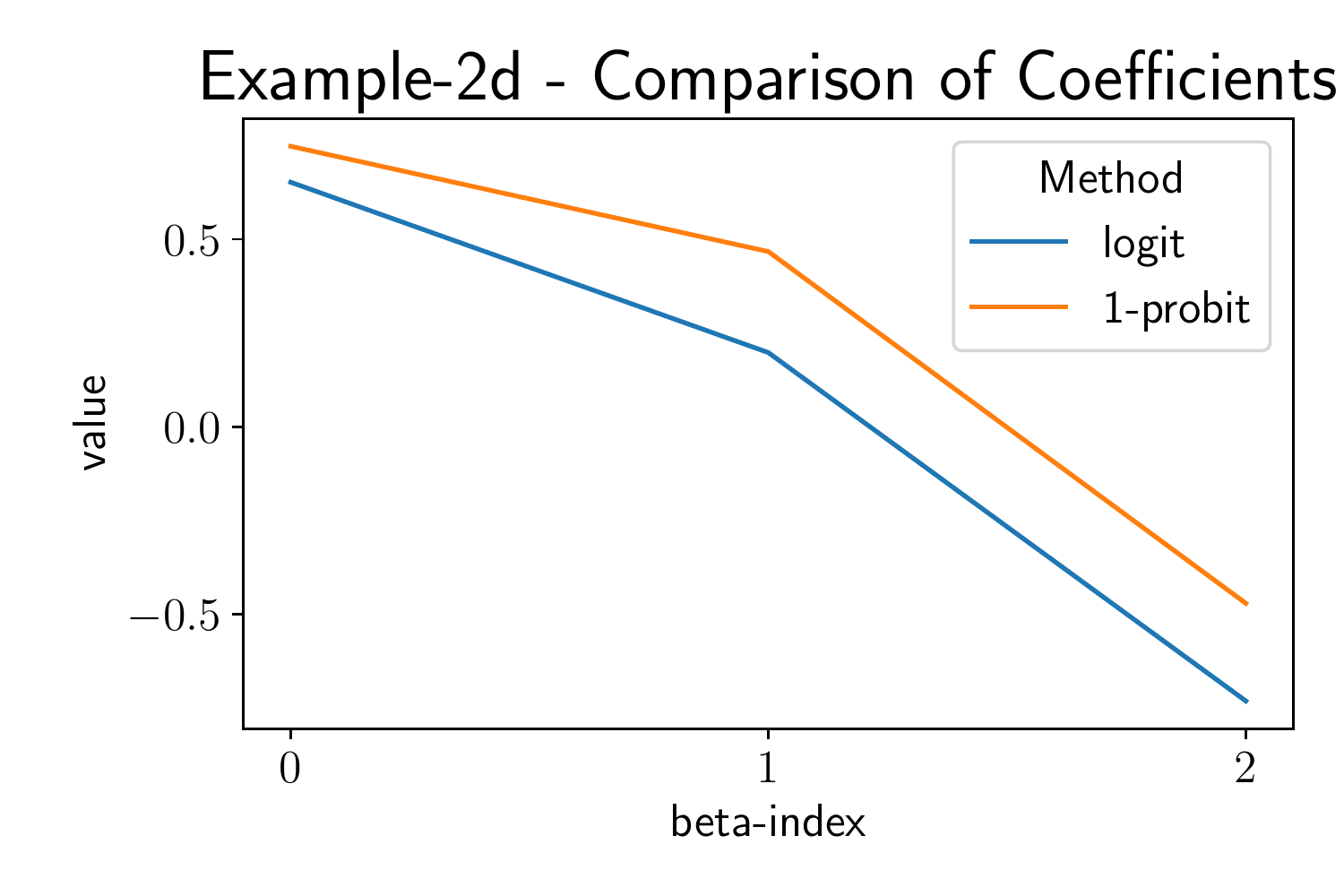}
    \caption{\textbf{(Q1)} \& \textbf{(Q2)}: (top) A 2D example data set that demonstrates how different values of $p$ affect the linear separator, as well as the distribution of the residuals $X{\beta}/{\|\beta\|_2}$ of misclassifications. (middle \& bottom) Comparison of normed coefficients for logit vs. 1-probit model for different data sets.}
    \label{fig:2d-example:appendix}
\end{figure}

\begin{figure*}[ht!]
\begin{center}
\begin{tabular}{ccc}
\includegraphics[width=0.309\linewidth]{plots/covertype_ratio_plot_p_1.pdf}&
\includegraphics[width=0.309\linewidth]{plots/webspam_ratio_plot_p_1.pdf}&
\includegraphics[width=0.309\linewidth]{plots/kddcup_ratio_plot_p_1.pdf}\\

\includegraphics[width=0.309\linewidth]{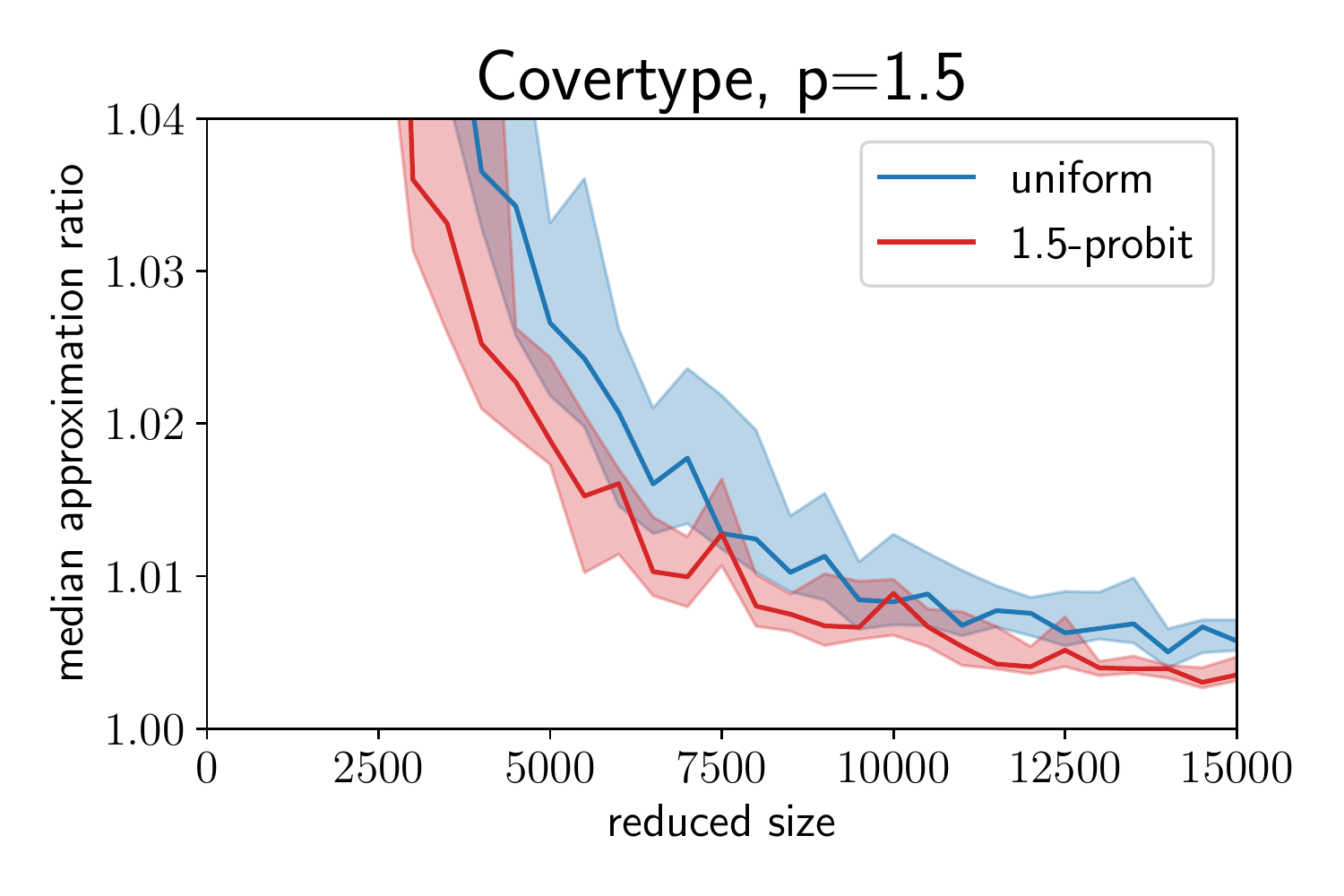}&
\includegraphics[width=0.309\linewidth]{plots/webspam_ratio_plot_p_1.5.pdf}&
\includegraphics[width=0.309\linewidth]{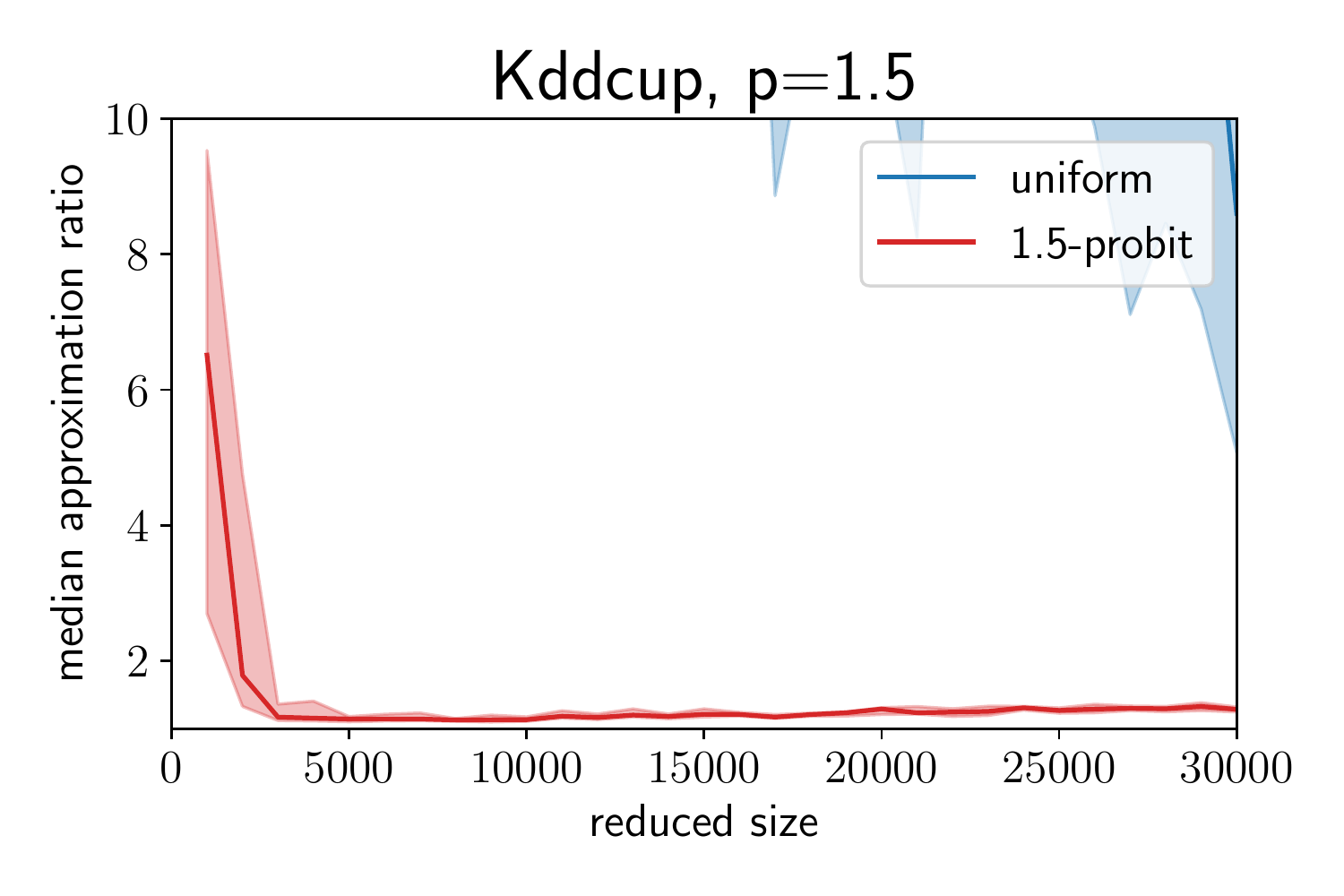}\\

\includegraphics[width=0.309\linewidth]{plots/covertype_ratio_plot_p_2.pdf}&
\includegraphics[width=0.309\linewidth]{plots/webspam_ratio_plot_p_2.pdf}&
\includegraphics[width=0.309\linewidth]{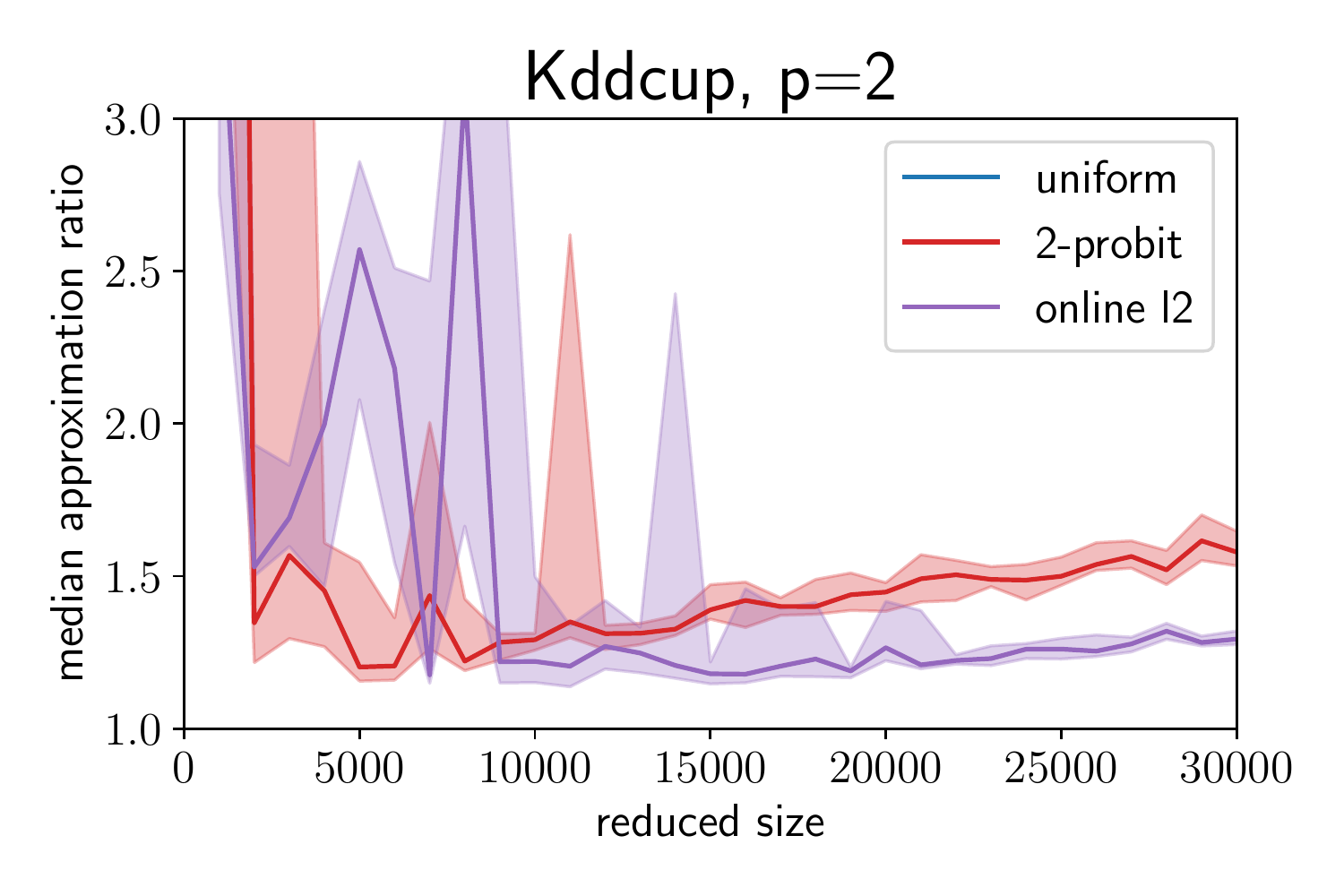}\\

\includegraphics[width=0.309\linewidth]{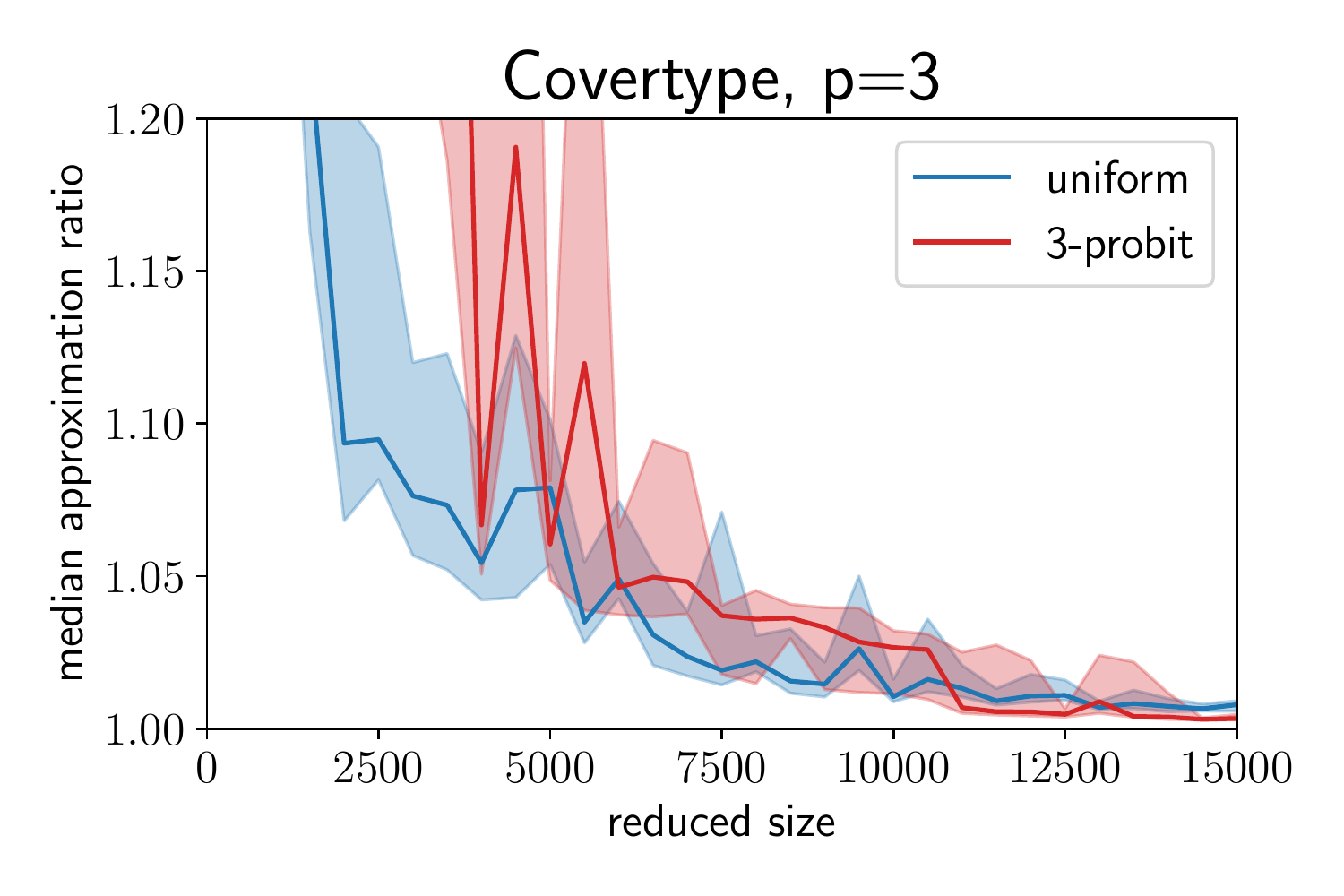}&
\includegraphics[width=0.309\linewidth]{plots/webspam_ratio_plot_p_3.pdf}&
\includegraphics[width=0.309\linewidth]{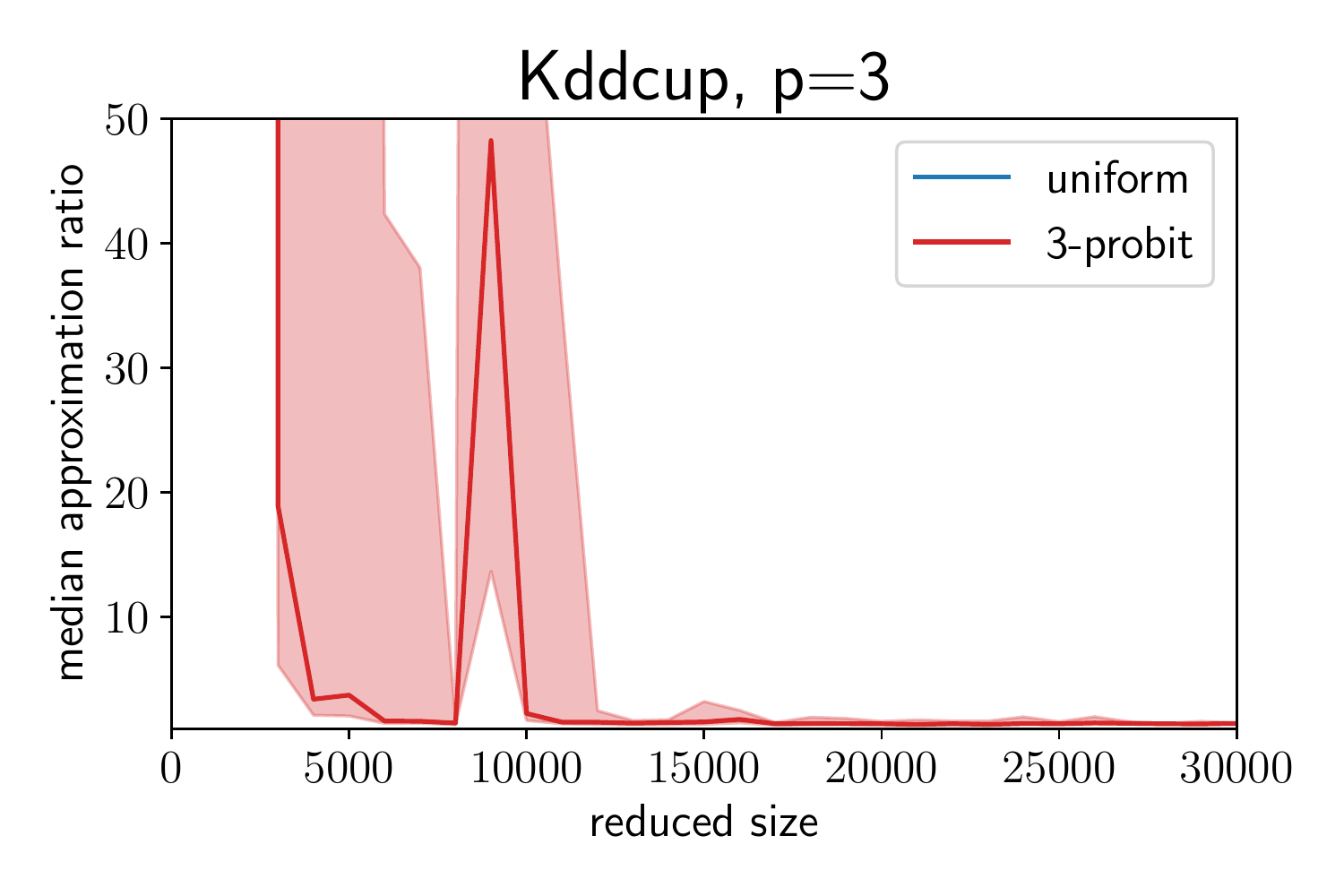}\\

\includegraphics[width=0.309\linewidth]{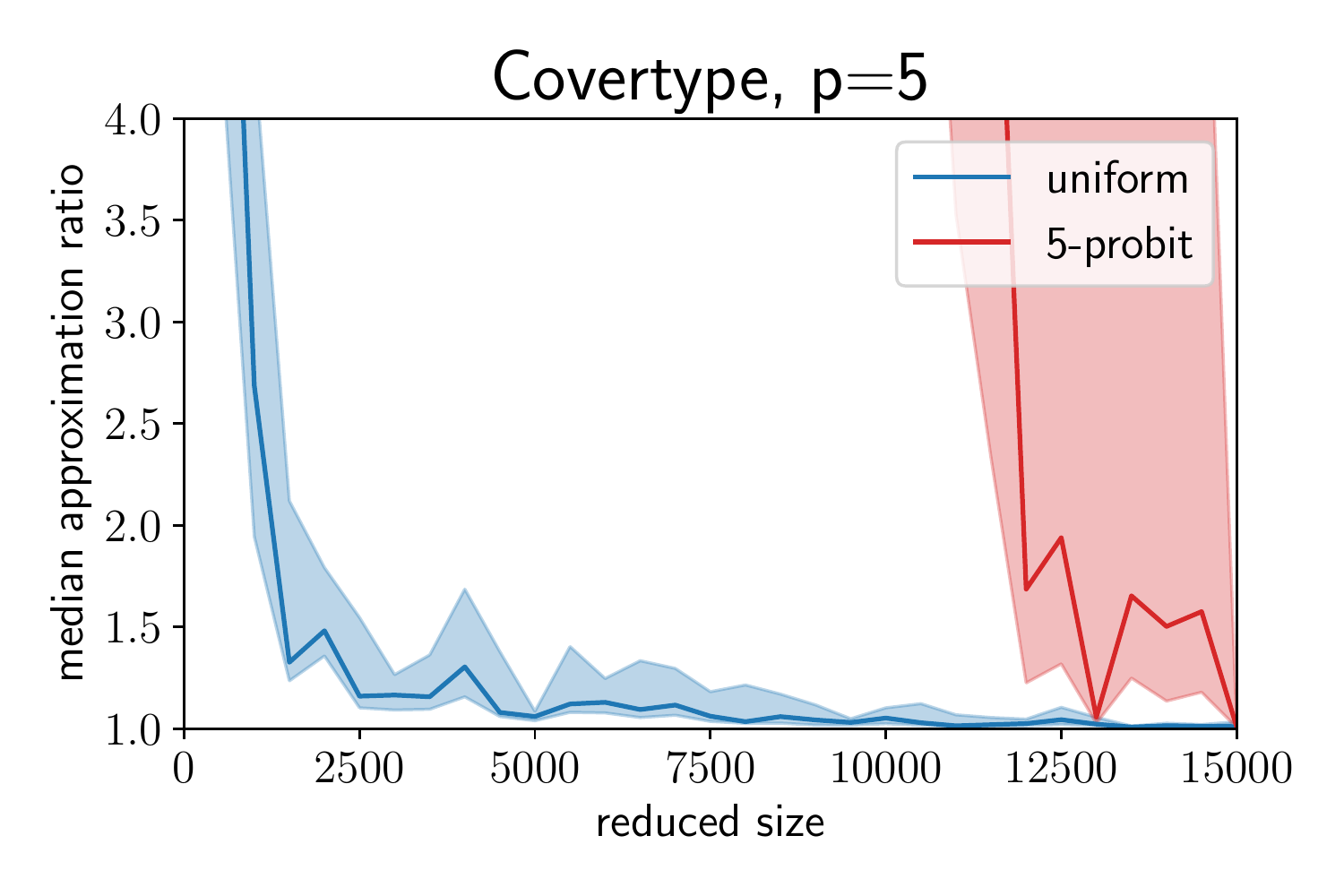}&
\includegraphics[width=0.309\linewidth]{plots/webspam_ratio_plot_p_5.pdf}&
\includegraphics[width=0.309\linewidth]{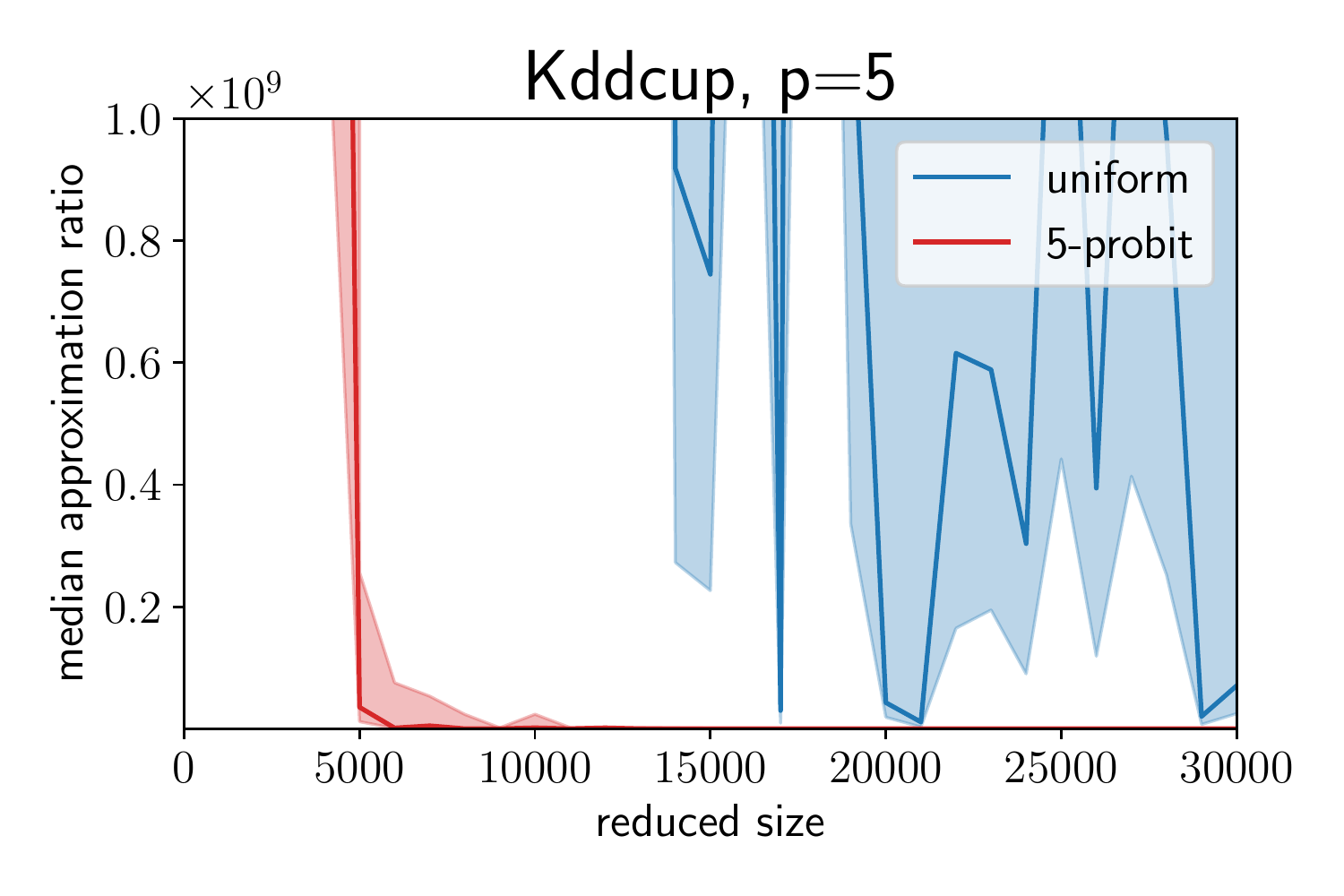}\\
\end{tabular}
\caption{Comparison of approximation ratios for $p\in \{1, 1.5, 2, 3, 5\}$ on the 
data sets Covertype, Webspam and Kddcup. The median approximation ratio denotes for each sample size the median of the approximation ratios over all repetitions. The solid line indicates the median, and the shaded area indicates the normalized interquartile range.}
\label{fig:all-ratio-plots}
\end{center}
\end{figure*}

\begin{figure*}[ht!]
\begin{center}
\begin{tabular}{ccc}
\includegraphics[width=0.309\linewidth]{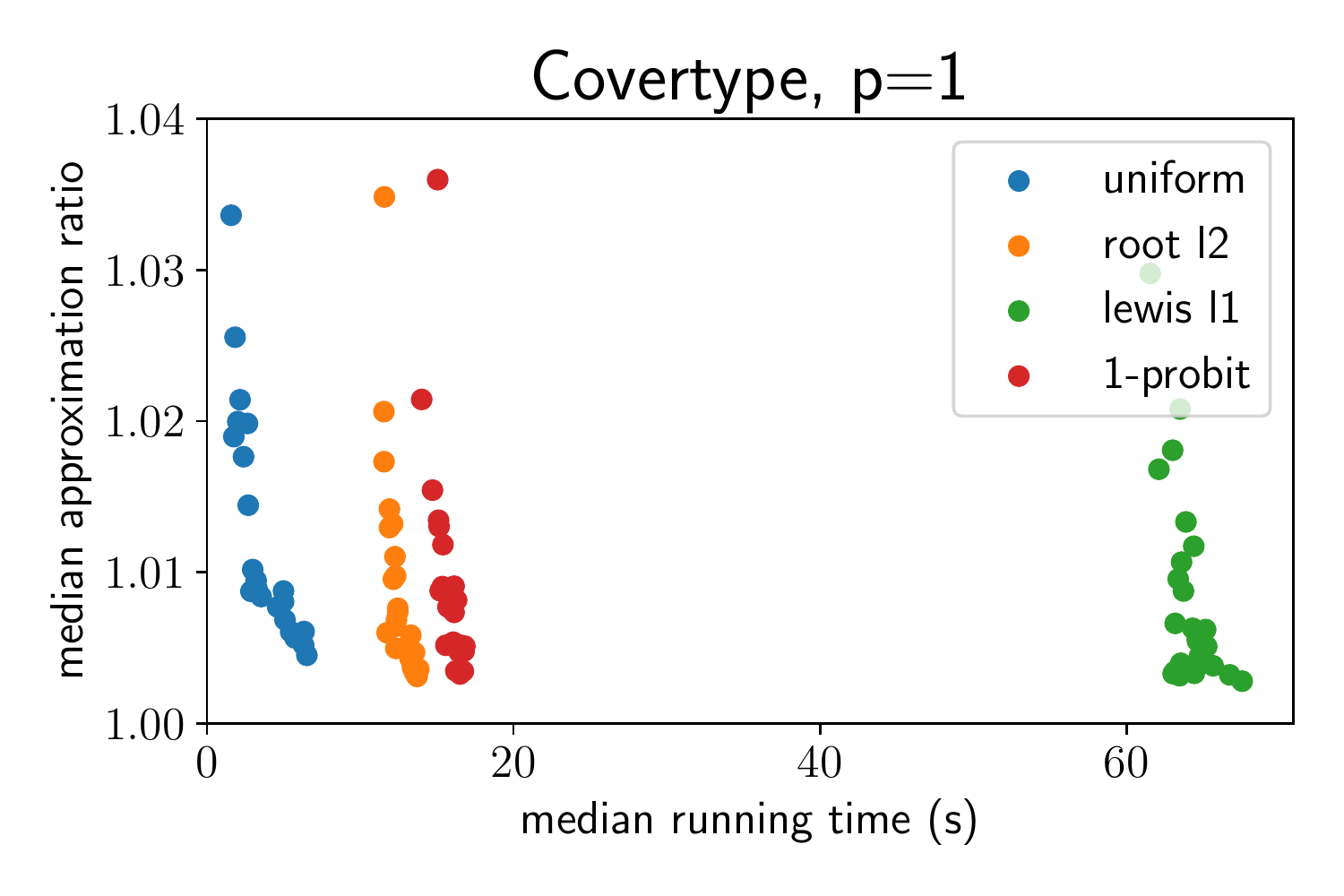}&
\includegraphics[width=0.309\linewidth]{plots/webspam_runtime_plot_p_1.pdf}&
\includegraphics[width=0.309\linewidth]{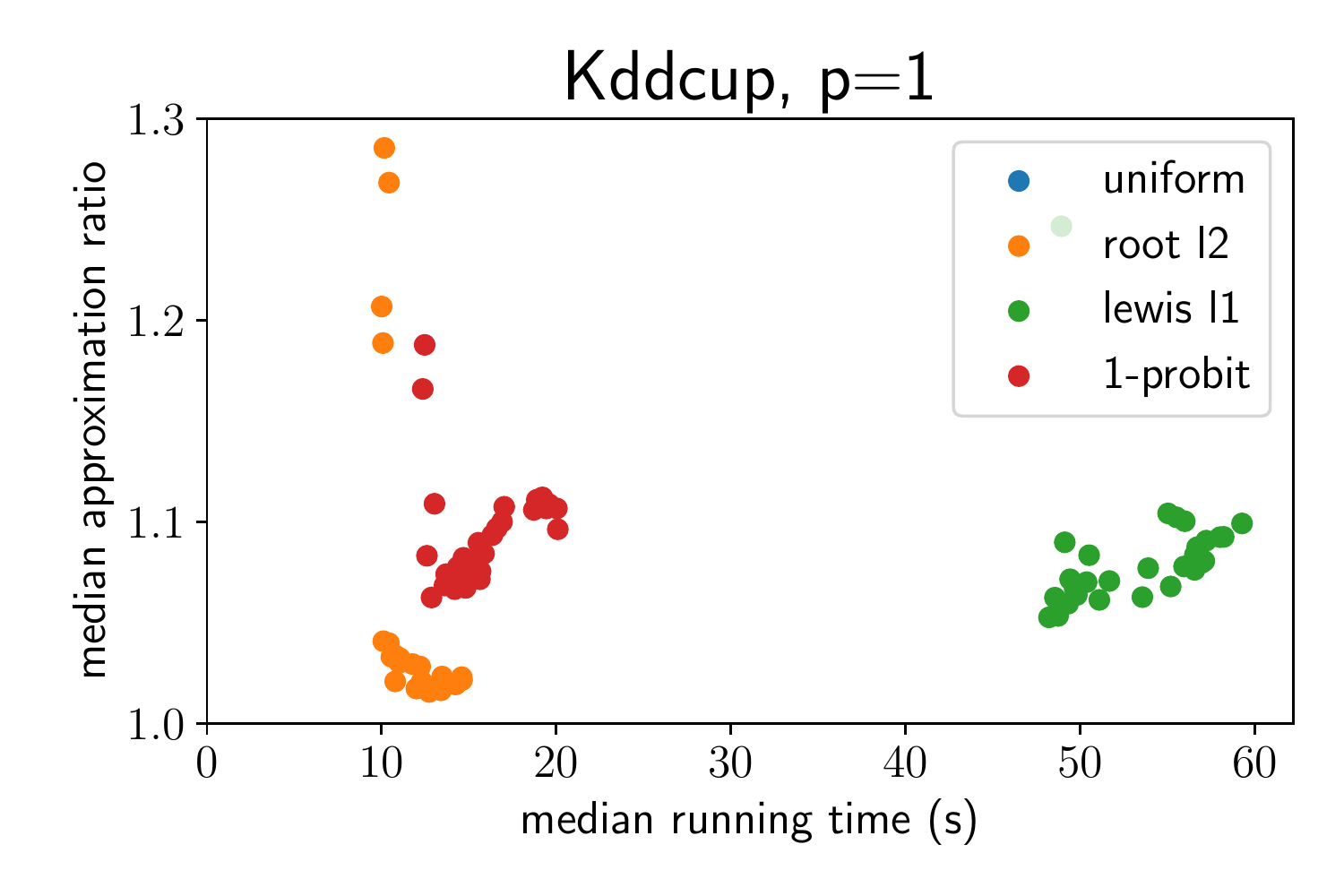}\\

\includegraphics[width=0.309\linewidth]{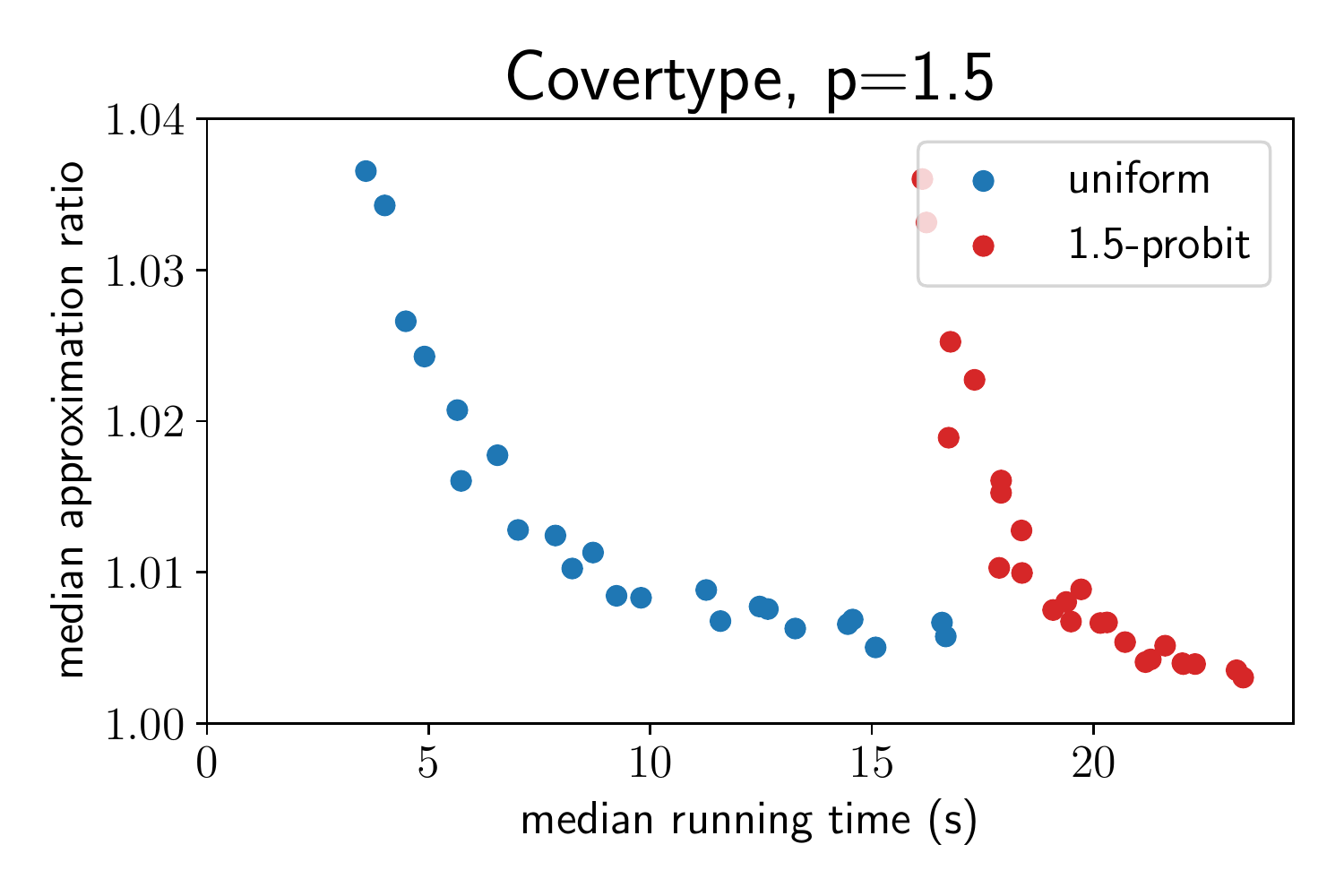}&
\includegraphics[width=0.309\linewidth]{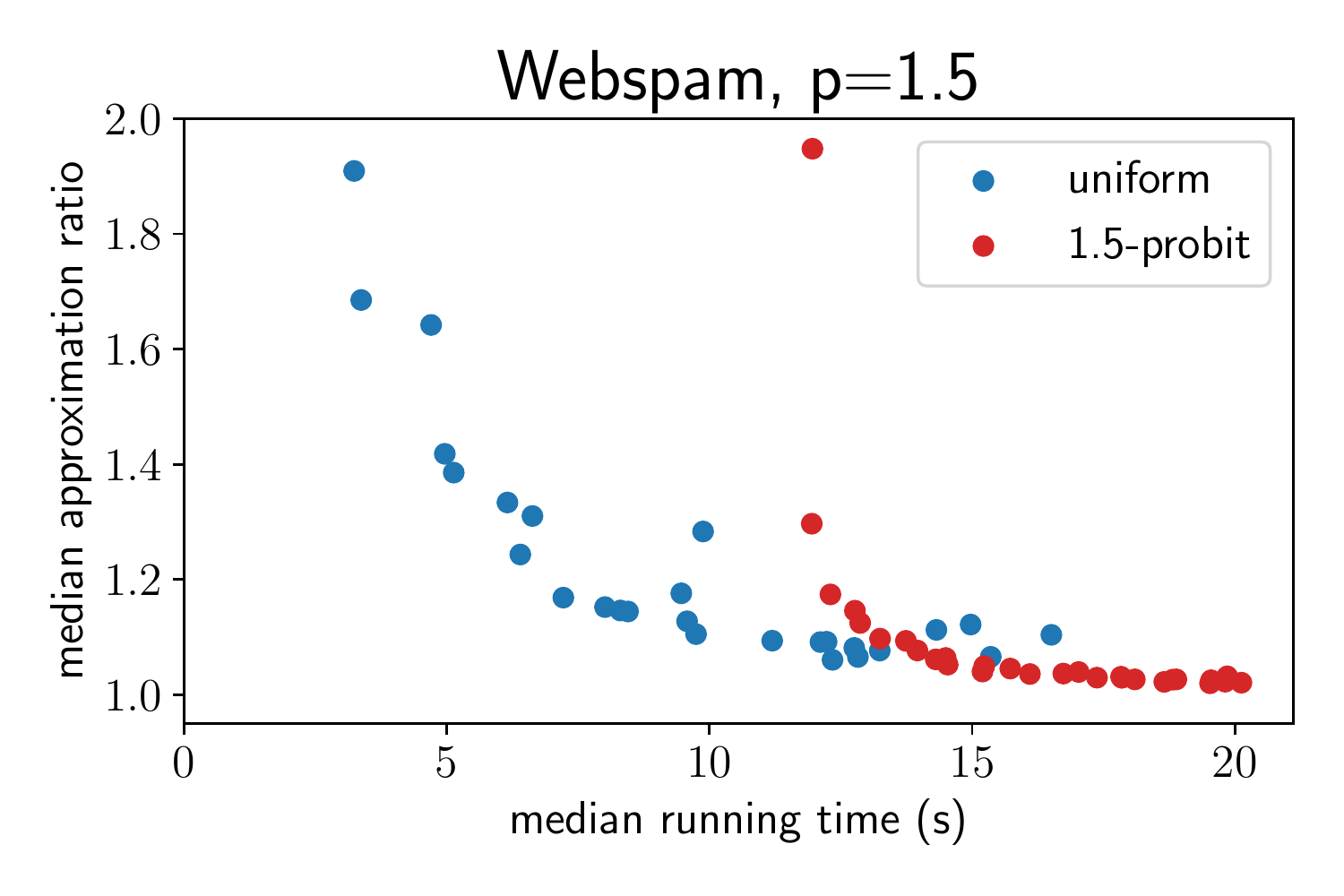}&
\includegraphics[width=0.309\linewidth]{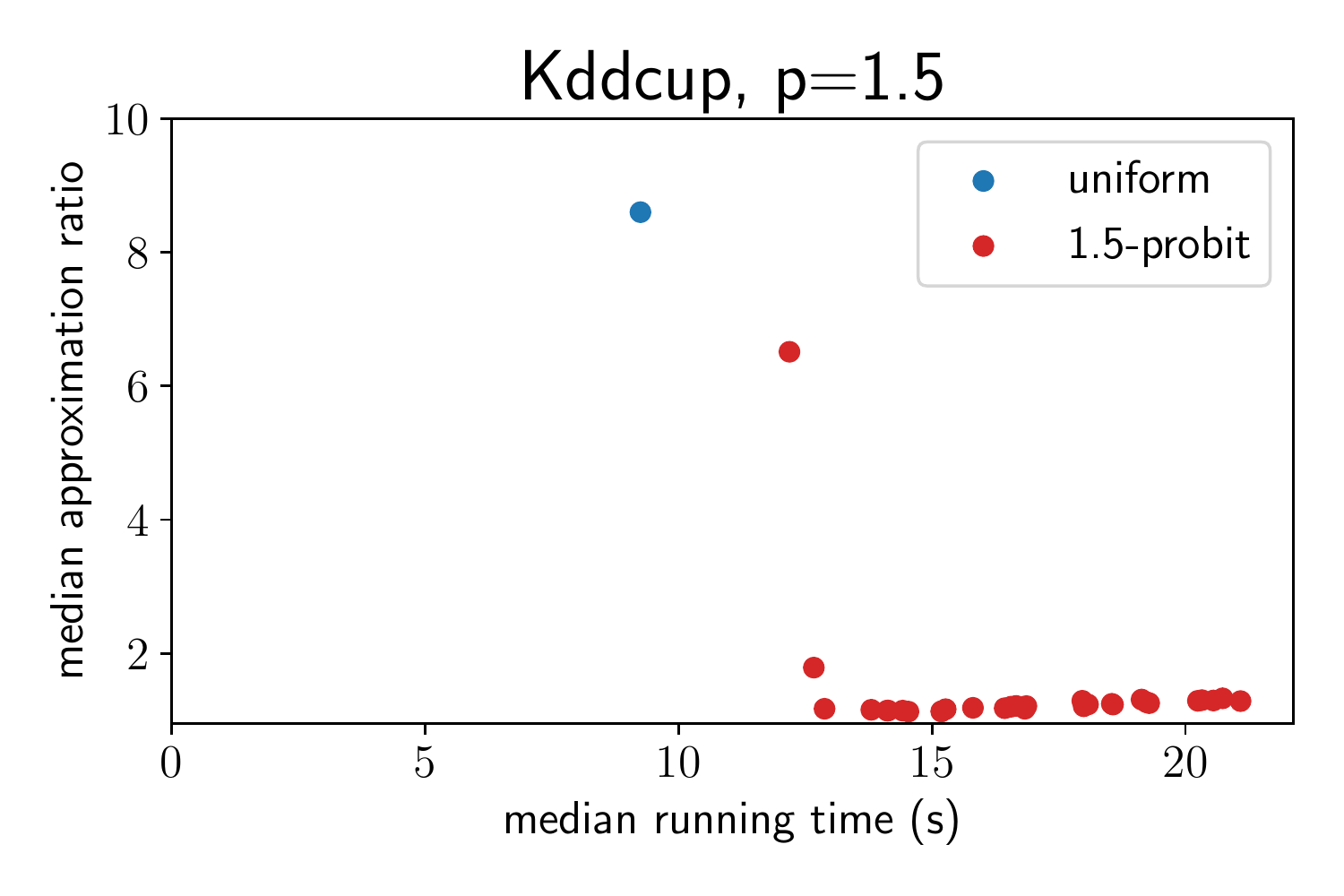}\\

\includegraphics[width=0.309\linewidth]{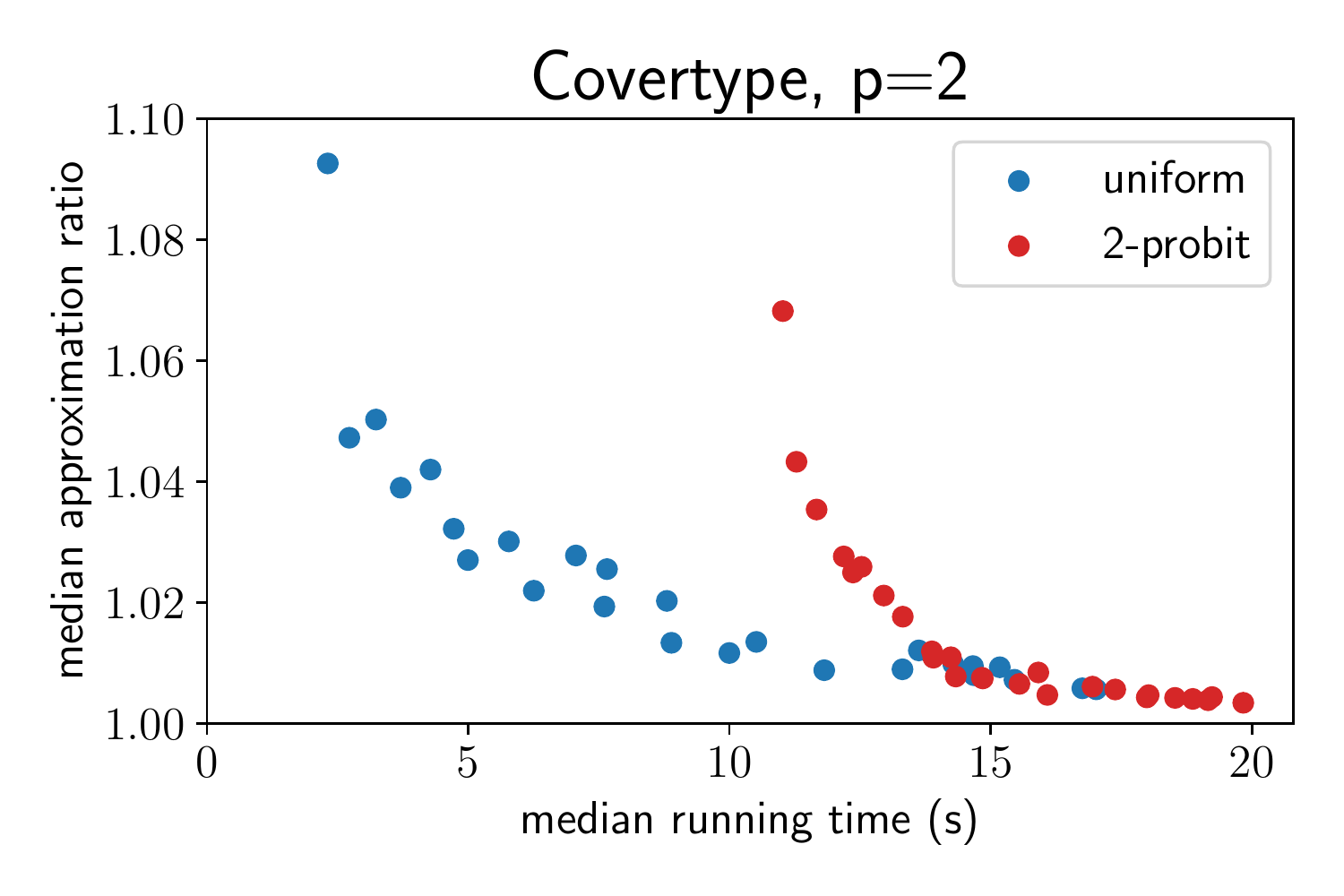}&
\includegraphics[width=0.309\linewidth]{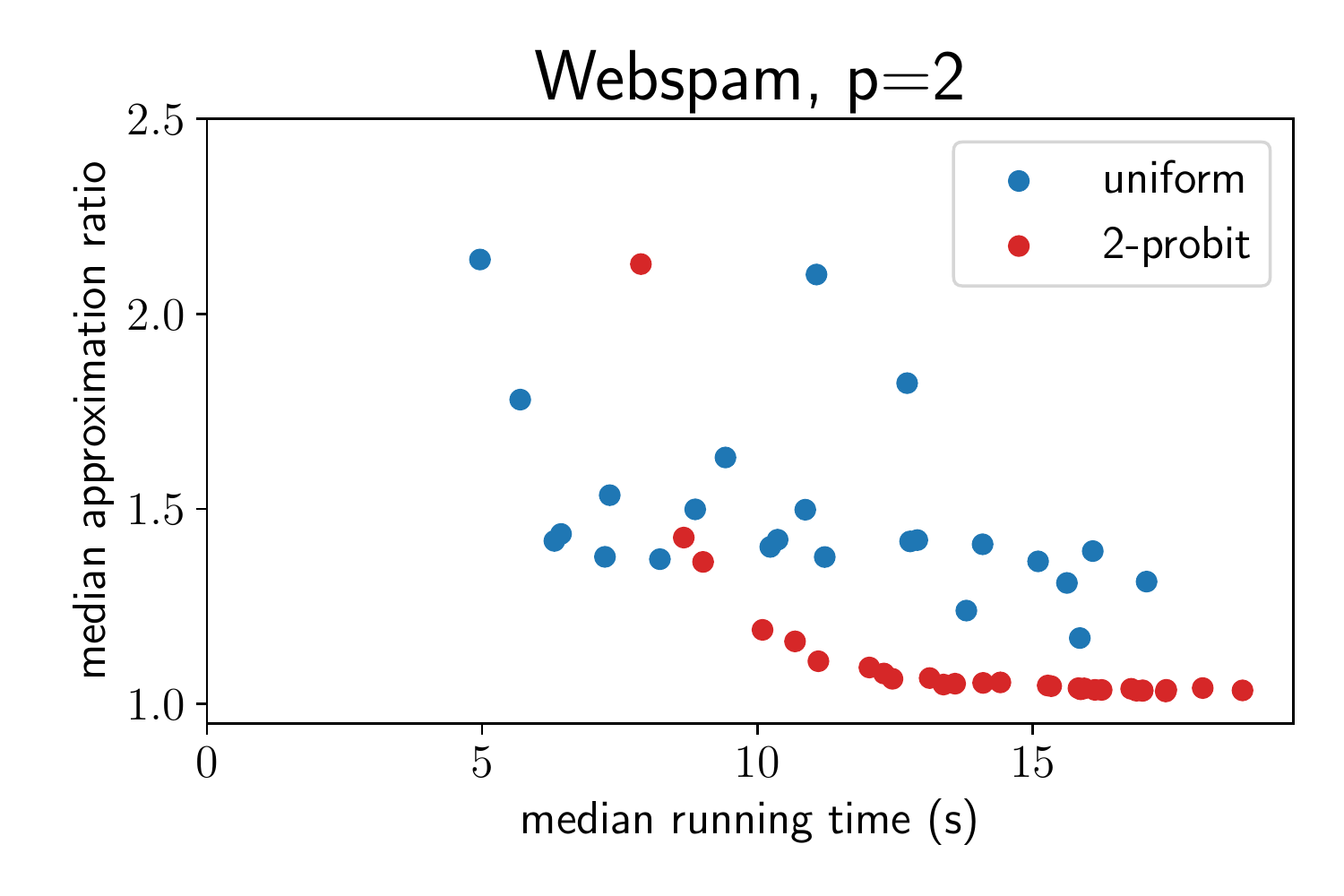}&
\includegraphics[width=0.309\linewidth]{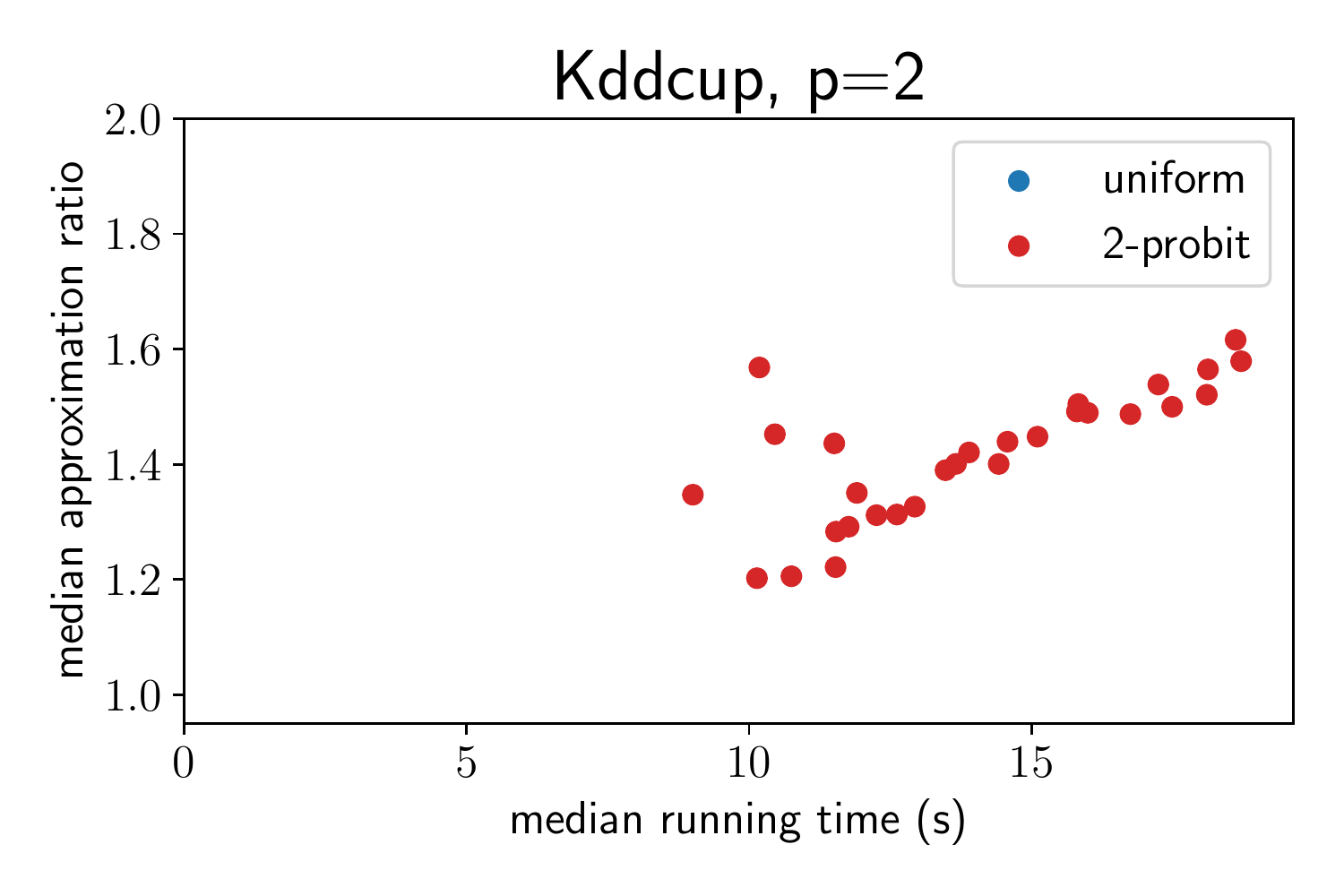}\\

\includegraphics[width=0.309\linewidth]{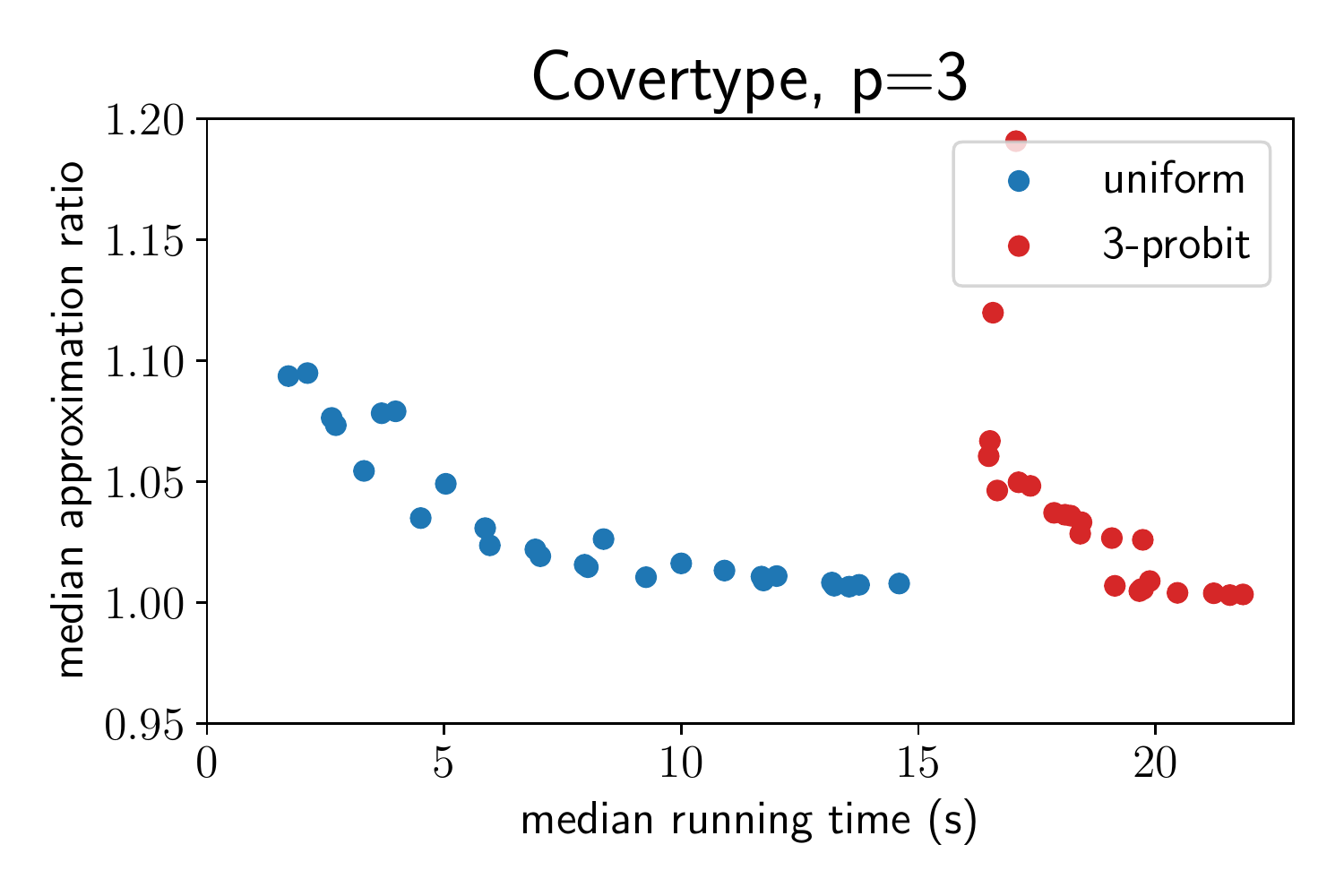}&
\includegraphics[width=0.309\linewidth]{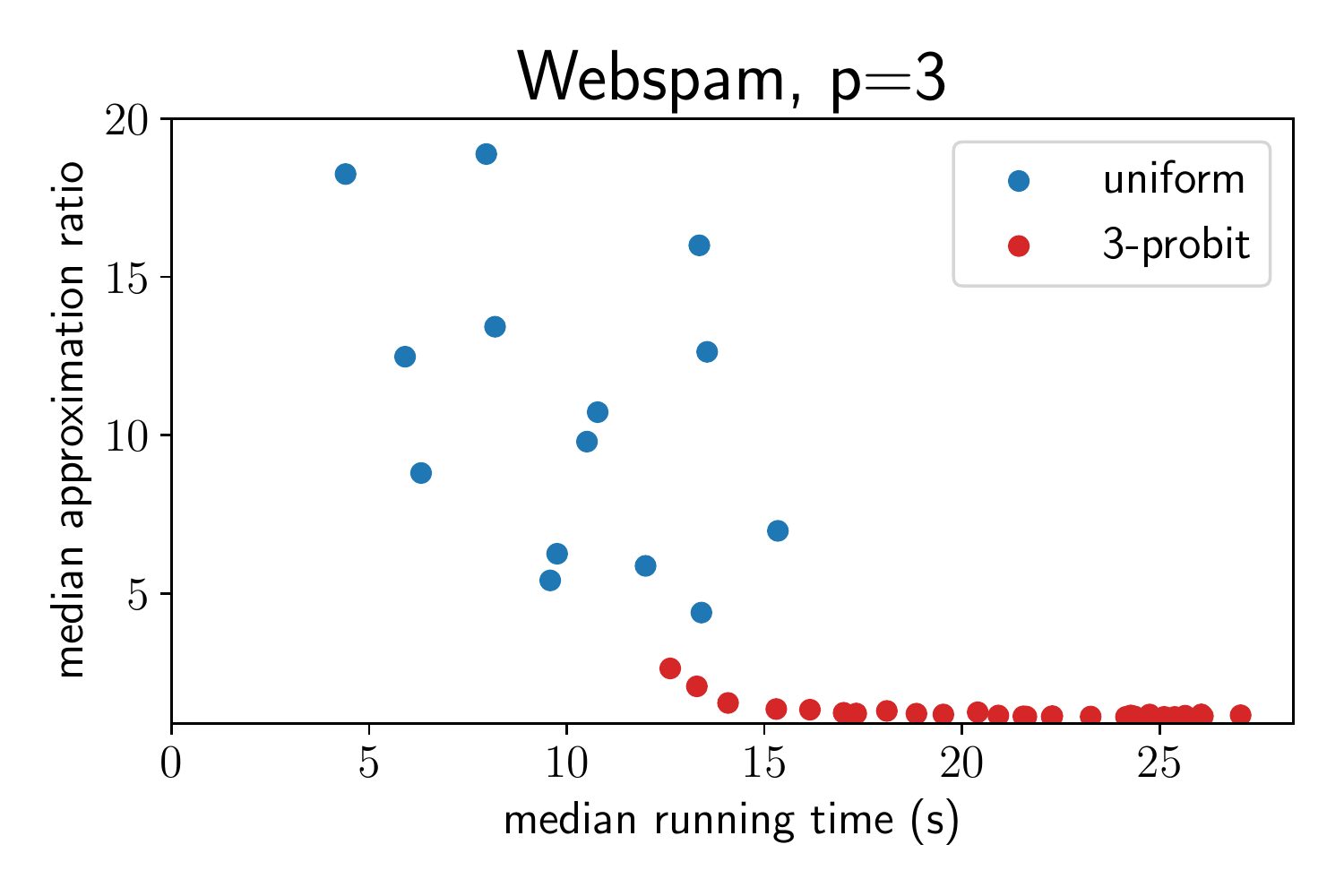}&
\includegraphics[width=0.309\linewidth]{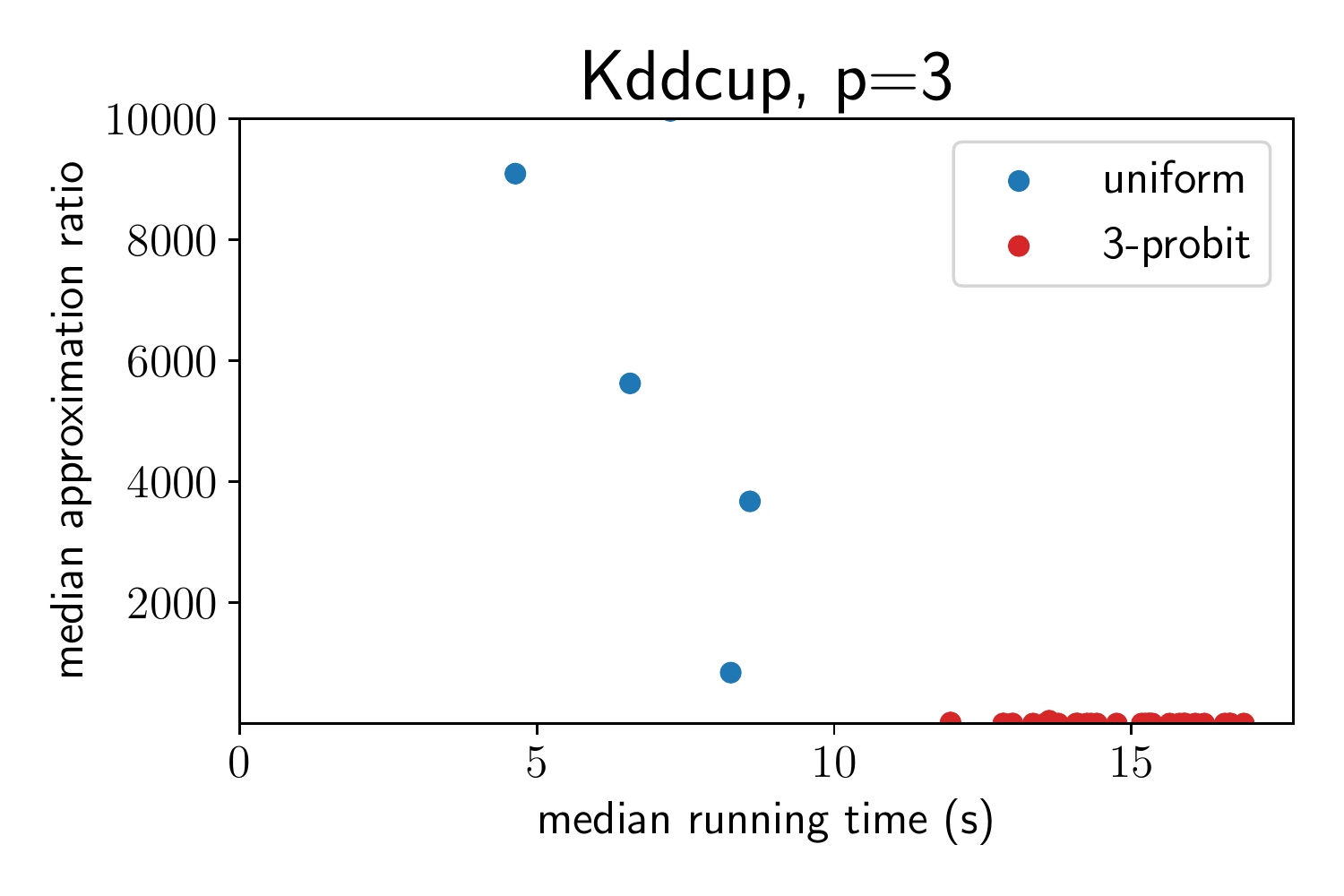}\\

\includegraphics[width=0.309\linewidth]{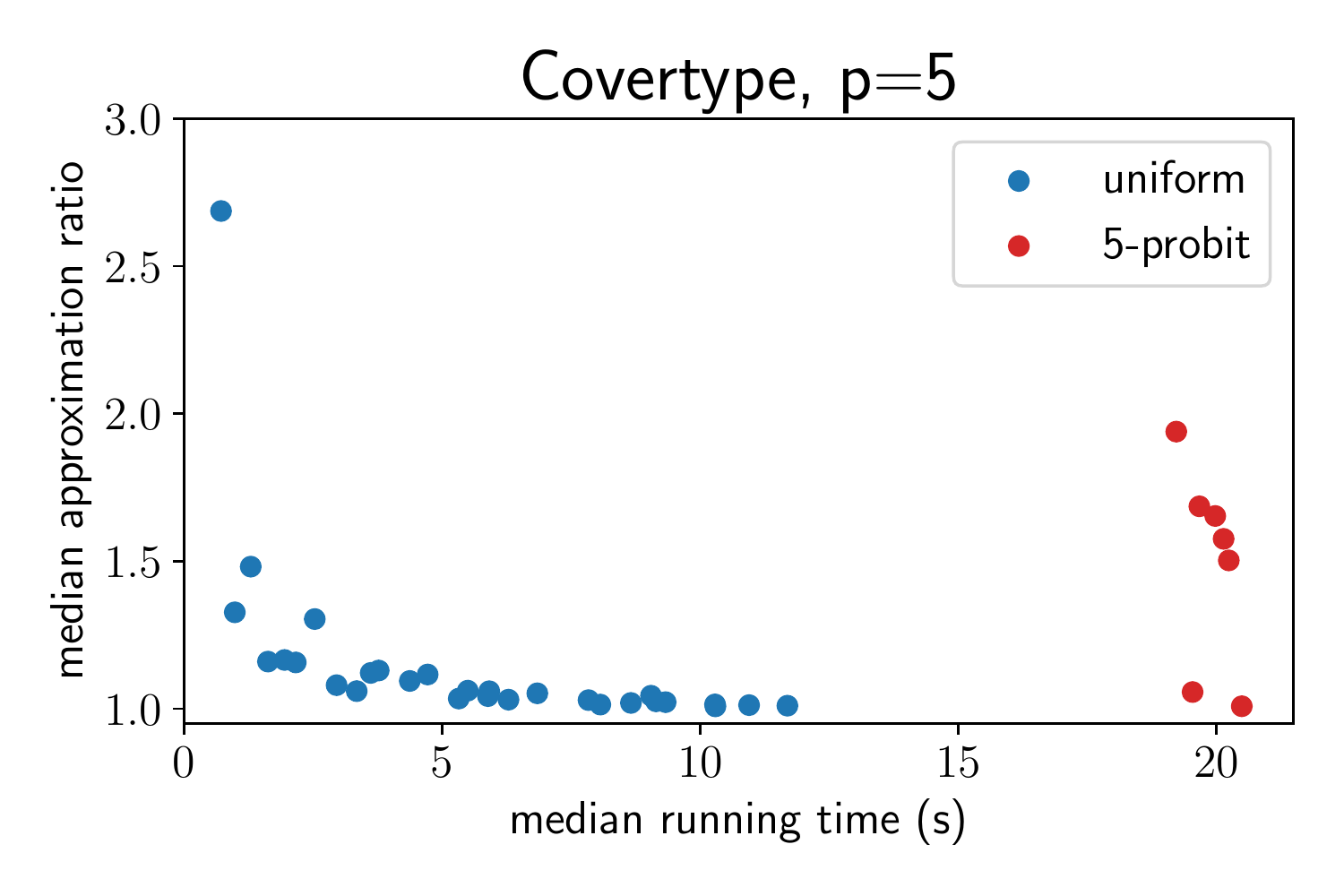}&
\includegraphics[width=0.309\linewidth]{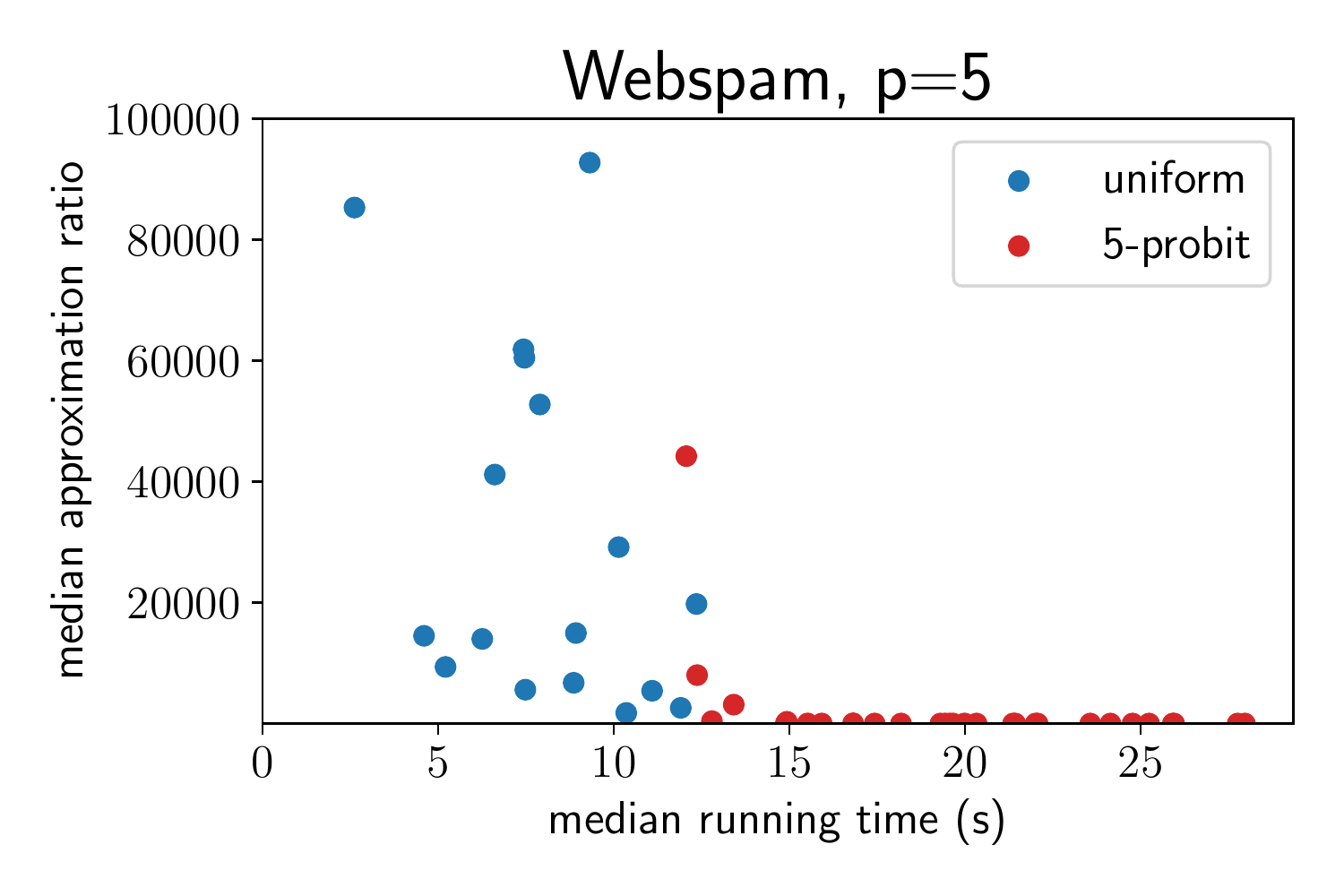}&
\includegraphics[width=0.309\linewidth]{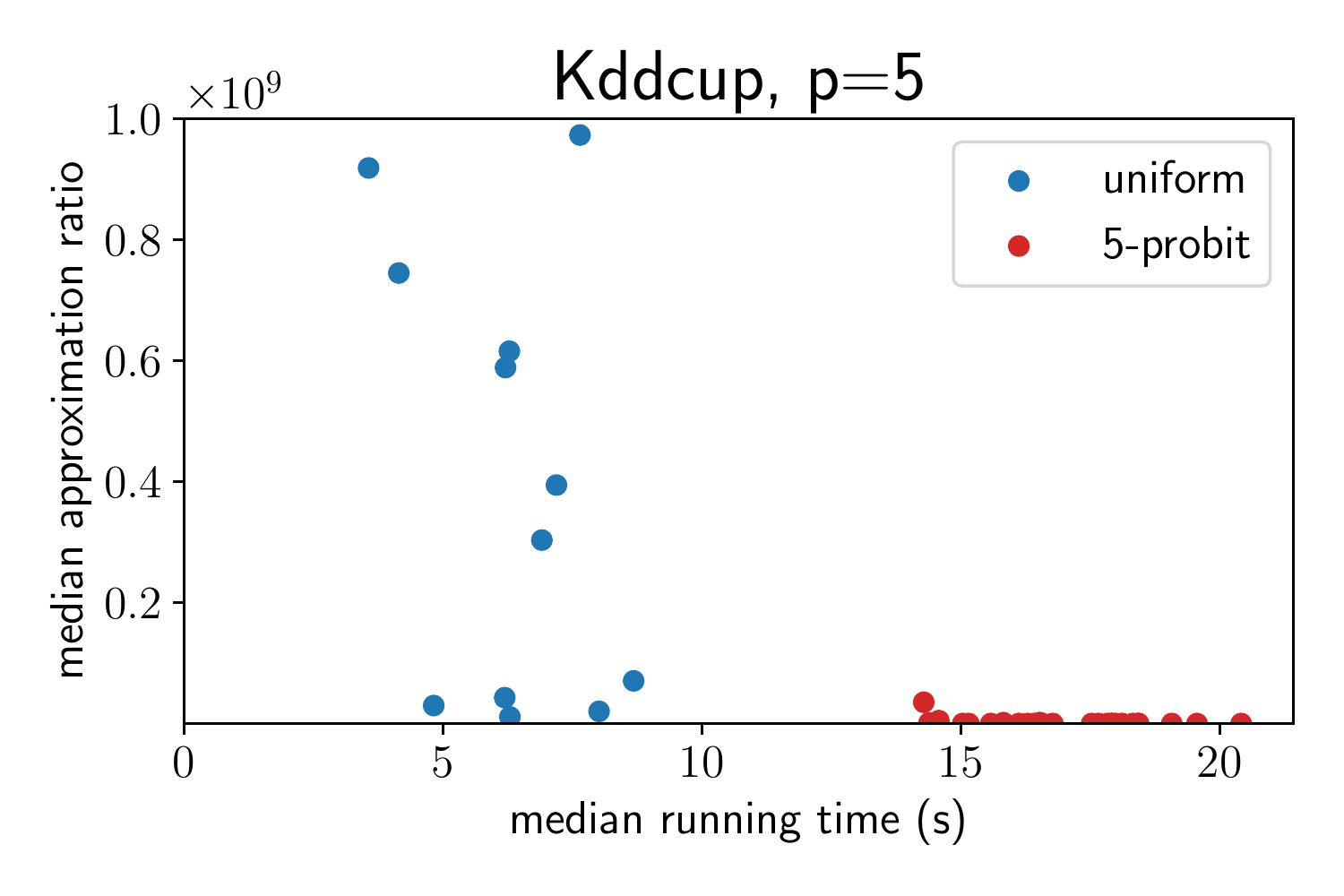}\\
\end{tabular}
\caption{Comparison of total running times (including optimization) vs. accuracy for $p\in \{1, 1.5, 2, 3, 5\}$ on the 
data sets Covertype, Webspam and Kddcup. The median approximation ratio (resp. running time) denotes for each sample size the median of the approximation ratios (resp. running times) over all repetitions.}
\label{fig:all-runtime-plots}
\end{center}
\end{figure*}
\end{document}